\documentclass[prd,tightenlines,nofootinbib,superscriptaddress]{revtex4}

\usepackage{pre_all}
\usepackage{pre_QI}
\allowdisplaybreaks

\hypersetup{
pdfstartview = {FitH},
}
\hypersetup{
	colorlinks=true,         
	linkcolor=brown,          
	citecolor=red,        
	urlcolor=blue            
}
\begin{document}

\title{Quantum Group Intertwiner Space From Quantum Curved Tetrahedron}

\author{{\bf Muxin Han}}\email{hanm@fau.edu}
\affiliation{Department of Physics, Florida Atlantic University, 777 Glades Road, Boca Raton, FL 33431, USA}
\affiliation{Institut f\"ur Quantengravitation, Universit\"at Erlangen-N\"urnberg, Staudtstr. 7/B2, 91058 Erlangen, Germany}

\author{{\bf Chen-Hung Hsiao}}\email{chsiao2017@fau.edu}
\affiliation{Department of Physics, Florida Atlantic University, 777 Glades Road, Boca Raton, FL 33431, USA}

\author{{\bf Qiaoyin Pan}}\email{qpan@fau.edu}
\affiliation{Department of Physics, Florida Atlantic University, 777 Glades Road, Boca Raton, FL 33431, USA}

\date{\today}

\begin{abstract}

In this paper, we develop a quantum theory of homogeneously curved tetrahedron geometry, by applying the combinatorial quantization to the phase space of tetrahedron shapes defined in \cite{Haggard:2015ima}. Our method is based on the relation between this phase space and the moduli space of SU(2) flat connections on a 4-punctured sphere. The quantization results in the physical Hilbert space as the solution of the quantum closure constraint, which quantizes the classical closure condition $M_4M_3M_2M_1=1$, $M_\nu\in \SU(2)$, for the homogeneously curved tetrahedron. The quantum group  $\mathcal{U}_q(\mathfrak{su}(2))$ emerges as the gauge symmetry of a quantum tetrahedron. The physical Hilbert space of the quantum tetrahedron coincides with the Hilbert space of 4-valent intertwiners of $\mathcal{U}_q(\mathfrak{su}(2))$. In addition, we define the area operators quantizing the face areas of the tetrahedron and compute the spectrum. The resulting spectrum is consistent with the usual Loop-Quantum-Gravity area spectrum in the large spin regime but is different for small spins. This work closely relates to 3+1 dimensional Loop Quantum Gravity in presence of cosmological constant and provides a justification for the emergence of quantum group in the theory.

\end{abstract}

\maketitle
\tableofcontents

\section{Introduction}

Quantum tetrahedron is a key building block in the theory of Loop Quantum Gravity (LQG) and plays a crucial role in the boundary states of the spinfoam amplitude of LQG. In LQG with vanishing cosmological constant, the physical Hilbert space of the quantum flat tetrahedron is the 4-valent $\SU(2)$ intertwiner space labeled by irreducible representation $K_\nu$'s, each assigned to a face $\ell_\nu$ of the quantum flat tetrahedron. Furthermore, the space is the solution space of the quantum flat closure condition, \ie $\sum^{4}_{\nu=1}\hat{\bf J}_\nu=0$, where $\hat{\bf J}_\nu=(\hat{J}_1,\hat{J}_2,\hat{J}_3)_\nu$ are $\mathfrak{su}(2)$ generators acting on the $\nu$-th copy of irreducible representation. The area spectrum of each face $\ell_\nu$ of the quantum flat tetrahedron is discrete and is characterized by a spin label $K_\nu$. 

The quantum flat closure condition is a quantization of the classical closure condition $\sum_{\nu=1}^4 a_\nu\hat{n}_\nu=0$, where $a_\nu,\hat{n}_\nu$ are the area and unit normal to the face $\ell_\nu$ respectively. Classically, the correspondence between a set of solutions of flat closure condition and flat tetrahedron is guaranteed by the \textit{Minkowski theorem} \cite{minkowski1897allgemeine}. This theorem has been generalized to the curved case \cite{Haggard:2015ima}, where a curved closure condition applies. The curved Minkowski theorem allows us to reconstruct homogeneously curved tetrahedra (spherical or hyperbolic tetrahedra) from a family of four $\SU(2)$ holonomies $M_\nu$ ($\nu=1,\cdots,4$) that satisfy the curved closure condition 
\be
M_{4}M_{3}M_{2}M_{1}=\Id_{\SU(2)}.
\label{eq:closure}
\ee
Here, $M_\nu$ is the holonomy around the $\nu$-th face of the homogeneously curved tetrahedron. Although the quantization of the closure condition for a flat tetrahedron has been extensively studied in LQG, the quantization of the curved closure condition and curved tetrahedron has not been explored yet. The homogeneously curved tetrahedron has played an important role in the recent construction of the spinfoam model with cosmological constant \cite{Haggard:2014xoa,Haggard:2015nat,Han:2021tzw} in 3+1 dimensional LQG. It is anticipated that the quantization of a curved tetrahedron should define the building block for the boundary Hilbert space of the spinfoam model. 

In this paper, we study the quantization of the curved closure condition and a homogeneously curved tetrahedron. At the classical level, the phase space corresponds to the solution spaces of the curved closure condition. The solution space coincides with the moduli space of $\SU(2)$ flat connections on a four-puncture sphere. In our approach, we adopt the Poisson bracket on the moduli space inspired by \cite{Fock:1997ai}, and then proceed with the quantization method known as combinatorial quantization \cite{Alekseev:1994au, Alekseev:1994pa, Alekseev:1995rn} as an approach of canonical quantization for the moduli space of flat connections. A quantum algebra $\mathfrak{M}$ of observables for flat connections is obtained on the four-puncture sphere. The quantum group $\mathcal{U}_q(\mathfrak{su}(2))$ emerges naturally as the gauge symmetry. It turns out that elements in $\mathfrak{M}$ enjoy the invariance under the quantum group action. 

The curved closure condition, when quantized, becomes a quantum constraint that is imposed on the quantum states to obtain the physical Hilbert space $W^0$ (we use bold letters to denote quantum operators) such that
\be
{\bf M}_{4}{\bf M}_{3}{\bf M}_{2}{\bf M}_{1}\Psi=\zeta\Psi,\qquad \forall\,\Psi\in W^0
\ee
where ${\bf M}_\nu$ is the quantum holonomy operator and the constant $\zeta$ has the classical limit $\zeta\xrightarrow{q\rightarrow1}1$. By using the representation theory of the quantum algebra $\mathfrak{M}$, we demonstrate that the solution space $W^0$ of the quantum closure condition exists and coincides with the intertwiner space of the quantum group $\mathcal{U}_q(\mathfrak{su}(2))$. Our result is valid for all $q=e^{i\theta}$ with $\theta\in(0,2\pi)$ which includes $q=e^{\frac{2\pi i}{k+2}}$ ($k\in\mathbb{N}$) being a root of unity. The intertwiner space depends on four irreducible $\mathcal{U}_q(\mathfrak{su}(2))$ representations $K_\nu\in \mathbb{N}/2$ (and additionally $K_\nu\leq k/2$ for $q$ root of unity), $\nu=1,...,4$ labeling the four punctures and corresponding to the quanta of the tetrahedron's areas. The quantization of the curved closure condition is equivalent to the quantization of the curved tetrahedra, with the solution space of the quantum closure condition encoding the geometric information of the quantum curved tetrahedra. The main result of this paper can be summarized by the following commuting diagram:

\be\ba{ccc}
\text{classical closure condition }\quad & \xrightarrow{\quad \text{ solution space }\quad}\quad & 
\ba{c}\text{ moduli space of flat connection}\\\text{as the physical phase space}\ea\\
\phantom{\text{quantization}}\left\downarrow\rule{0cm}{1cm}\right.\text{\small quantization} &&
\phantom{\text{quantization}}\left\downarrow\rule{0cm}{1cm}\right.\text{\small quantization}\\
\text{quantum closure condition }\quad & \xrightarrow{\quad \text{ solution space }\quad}\quad & \ba{c}\text{representation space of quantum algebra}\\\text{as the physical Hilbert space}\ea
\ea
\label{eq:road_map}
\ee

For $q$ a root of unity with $\theta=\frac{2\pi}{k+2}$ $(k\in\mathbb{N})$, the solution space derived from solving the quantum closure condition matches the physical Hilbert space of the Chern-Simons theory with a compact Lie group $\SU(2)$ on a compact surface with zero genus and four punctures, and $k+2={24 \pi}({\ell^2_{\rm p}\gamma|\Lambda|})^{-1}$ relates the Chern-Simons level to the cosmological constant. 
Here, $\gamma$ is the Barbero-Immirzi parameter, $\ell_{\rm p}=\sqrt{8\pi G\hbar/c^3}$ is the Planck length and $\Lambda$ is the cosmological constant. 
For generic $q=e^{i\theta}$, the phase $\theta$ relates to the cosmological constant by $\theta=\frac{1}{12}\ell_{\rm p}^2\gamma |\Lambda|$.

One of the geometric information of the curved tetrahedra is the areas $a_\nu$'s of the faces of a curved tetrahedron. The area $a_\nu$ relates to the holonomy $M_\nu$ in the closure condition \eqref{eq:closure} by $\frac{1}{2}\mathrm{Tr}(M_\nu)=\cos{(\frac{|\Lambda|}{6}
a_\nu)}$ \cite{Haggard:2015ima}, where $\Lambda$ is the cosmological constant.
The quantization relates the area operator to the  $q$-deformed Wilson loop. As a result, the area spectrum labeled by spin $K_\nu$ is given by
\be\label{area0000}
{\rm Spec}^{K_\nu}(\hat{a}_\nu)=\left\{\begin{array}{ll}
\gamma \ell_{\mathrm{p}}^2\left(K_\nu+\frac{1}{2}\right), & 0 \leq K_\nu<\frac{1}{2} B \\[0.1cm]
\frac{12 \pi}{|\Lambda|}-\gamma \ell_{\mathrm{p}}^2\left(K_\nu+\frac{1}{2}\right), & \frac{1}{2} B \leq K_\nu<B
\end{array}, \quad B=\frac{12 \pi}{\ell_{\mathrm{p}}^2 \gamma|\Lambda|}-1\right.,
\ee
where $\hat{a}_\nu$ is the quantum area operator for the $\nu$-th face. 
The area spectrum is bounded above and below. This formula is valid for both $q$ root of unit and generic $q$ as a phase. The cosmological constant provides a cut-off to the area spectrum.
For vanishing cosmological constant (so the second case in \eqref{area0000} disappears) and large spin $K_\nu$, the area spectrum reduces to $\gamma \ell^2_{\rm p} K_\nu$, which is consistent with the usual LQG area spectrum $\gamma\ell_{\rm p}^2\sqrt{K_\nu(K_\nu+1)}$ for large $K_\nu$. However, for small $K_\nu$, there is a significant difference from LQG with vanishing $\Lambda$. In particular, the area spectrum here is restrictively positive, whereas the usual LQG spectrum gives a trivial area at $K_\nu=0$.

This paper is organized as follows. Section \ref{sec:classical_closure} reviews the main ideas regarding classical curved closure conditions and its solution space, which is the moduli space of $\SU(2)$ flat connections. In this section, we also collect the main theorem of \cite{Haggard:2015ima}. In Section \ref{sec:moduli_space}, following Fock and Rosly's idea \cite{Fock:1998nu} that one can replace a 2-dimensional surface with a homotopically equivalent fat graph with an additional structure called ciliation, the graph we choose is called a simple graph which contains one base point and four loops that are generators of the fundamental group on the four-punctured sphere.  With this setup, we give the Poisson bracket of holonomies. The Poisson bracket is defined by the classical $r$-matrix, which is the solution of the classical Yang-Baxter equation. Section \ref{sec:quantum_moduli} reviews the construction of quantum algebra of moduli space of $\SU(2)$ flat connections (moduli algebra) and its representations following \cite{Alekseev:1994au, Alekseev:1994pa, Alekseev:1995rn}. In Section \ref{sec:quantum_closure}, we prove that the solution space of the quantum closure condition coincides with the intertwiner space of quantum group $\mathcal{U}_q(\mathfrak{su}(2))$, which is the representation space of moduli algebra. The intertwiner space is characterized by four irreducible representations of $\mathcal{U}_q(\mathfrak{su}(2))$ labeled by $K_\nu, \nu=1,...,4$, labeling the four punctures and corresponding to the quanta of the tetrahedron's areas. 
In Section \ref{sec:area}, we define the quantum area operator and calculate its spectrum. We also compare it with the area spectrum obtained from standard LQG. We conclude the paper in Section \ref{sec:conclusion}.

\section{Classical closure condition of a curved tetrahedron}
\label{sec:classical_closure}

We start by reviewing the closure condition of a convex homogeneously curved tetrahedron whose (Gaussian) curvature, identified with the cosmological constant $\Lambda$, can be either positive or negative. That is each face of the tetrahedron is flatly embedded in a three-sphere $S^3$ or a hyperbolic three-space $\bH^3$. Tetrahedra in flat space are then obtained when $\Lambda\rightarrow 0$ in both cases. 
To unify the notations, we denote the sign of the curvature as $s:=\sgn(\Lambda)$ and the $n$-dimensional homogeneously curved space as $\bE^{n,s}$ hence $\bE^{3,+}=S^3$ and $\bE^{3,-}=\bH^3$. 
Let us also fix the notations for the simplices. For a tetrahedron, we label a vertex $v$ and its opposite face $f$ by the same number and label an oriented edge connecting the target vertex $v_1$ and source vertex $v_2$ by the number pair $(v_1v_2)$. The same edge with the opposite orientation is denoted as $(v_1v_2)^{-1}\equiv(v_2v_1)$. 
An illustration of spherical and hyperbolic tetrahedra and the notations for simplices are given in fig.\ref{fig:tetrahedra}. 
We focus on non-degenerate tetrahedra in this paper. 
\begin{figure}[h!]
\centering
\begin{subfigure}{0.4\linewidth}
     \begin{tikzpicture}[scale=2]
  \draw[black, thick] (3,3) arc (0:60:1);
 \draw[black, thick] (2.5,3.866) arc (100:152.5:0.95);
 \draw[black, thick] (2.5,3.866) arc (120:200:1);
 \draw[black, thick] (2.06,2.66) arc (-139.5:-184:1);
 \draw[black, thick] (3,3) arc (-40:-100:1);
 \draw[black, thick] (3,3) arc (-70:-127:1);
 \draw[black, thick] (1.823,3.380) arc (-147:-132:1);
 \draw[black] (2.5,3.866)  node[anchor=south]{4};
 \draw[black] (3,3)  node[anchor=west]{3};
 \draw[black] (2.05,2.667) node[anchor=north]{2};
 \draw[black] (1.8,3.4) node[anchor=east]{1};
 \end{tikzpicture} 

	\subcaption{}
	\label{fig:SphericalTetra}
\end{subfigure}
\begin{subfigure}{0.4\linewidth}
    \begin{tikzpicture}[scale=3]
  \draw[black,thick] (-3,-3) arc (-27:36:-1);
 \draw[black,thick] (-3,-3) arc (-28.2:-67.8:1.3);
 \draw[black,thick] (-3,-3) arc (10:50:-1);
 \draw[black,thick] (-2.66,-3.59) arc (-124:-88:-0.69);
 \draw[black,thick] (-3.65,-3.59) arc (-64.1
 :-90.5:-1.13);
 \draw[black,thick] (-2.66,-3.59) arc (-70:11:-0.4);
 \draw[black,thick] (-2.912,-4.049) arc (-147.5:-96:-1);
 \draw (-3,-3) node[scale=1.3,anchor=south]{4};
 \draw (-2.66,-3.59) node[scale=1.3,right]{3};
 \draw (-2.912,-4.049) node[scale=1.3,below]{2};
 \draw (-3.692,-3.549) node[scale=1.3,below]{1};
  \end{tikzpicture}
	\subcaption{}
	\label{fig:HyperbolicTetra}
 \end{subfigure}
\caption{
{\it (a)} A tetrahedron flatly embedded in $S^3$. {\it (b)} A tetrahedron flatly embedded in $\bH^3$. For both tetrahedra, each one of the four numbers 1,2,3,4 labels a vertex as well as the face opposite to the vertex. An edge is labelled by a pair of numbers. For instance, the oriented edge connecting the source vertex 4 and the target vertex 2 is $(24)$ and the edge with the opposite orientation is $(42)\equiv(24)^{-1}$. The orientations of the edges are not specified here. }
\label{fig:tetrahedra}
\end{figure}
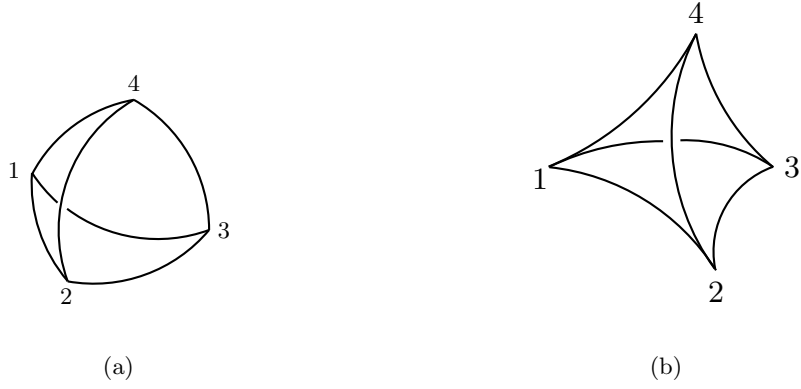

Each face of a tetrahedron is a triangle flatly embedded in a two-dimensional homogeneously curved subspace $\bE^{2,s}$ of $\bE^{3,s}$. The convexity guarantees that each edge of the triangle is the shortest geodesic on $\bE^{2,s}$ connecting the two end vertices of the edge. 
For each face, we choose a base point $p$ on the boundary and consider the oriented loop $\ell$ along the boundary starting and ending at $p$ whose orientation is counterclockwise when seen from the outside of the tetrahedron. Such an orientation generates an outward direction normal $\hat{n}_\ell(p)$ to the face at $p$ (and any other point within the face) by the right-hand rule, which is consistent with the topological orientation of the tetrahedron. We also denote the same loop with the opposite orientation as $\ell^{-1}$.

Indeed, a vector at $p$ tangent to the face gets rotated after parallel transport along $\ell$. The rotation angle is proportional to the area $a_\ell$ of the face enclosed by $\ell$. 
We denote the holonomy of the Levi-Civita connection along $\ell$ in the local frame of $p$ as $O_\ell(p)$. It is a group element of $\SO(3)$ for both curved tetrahedra embedded in $S^3$ and $\bH^3$ and encodes the information of $a_\ell$ by
\be\label{Ol0000}
O_\ell(p)=\exp\left[s \frac{|\Lambda|}{3}a_\ell\hat{n}_\ell(p)\cdot \vec{J}\right]\in \SO(3)\,,\qquad 
\frac{|\Lambda|}{3}a_\ell\in[0,2\pi]
\ee
where $\vec{J}=\{J_1,J_2,J_3\}$ are the generators of $\so(3)$ and the sign $s$ determines in which space the tetrahedron is embedded. 
For the convenience of constructing geometrical variables such as areas, dihedral angles and triple products, we lift the $\SO(3)$ holonomies $O_\ell$'s to the $\SU(2)$ ones $M_\ell$'s encoding the face areas in a similar way \cite{Haggard:2015ima}: 
\be
M_\ell(p)=\exp\left[s \frac{|\Lambda|}{3} a_\ell \hat{n}_\ell(p)\cdot\vec{\tau} \right]
\equiv \cos \left(\frac{|\Lambda|}{6}a_\ell\right)\Id-i\sin \left(\frac{|\Lambda|}{6}a_\ell\right)\vec{\sigma}\in\SU(2)\,,
\label{eq:H_definition}
\ee
where $\vec{\tau}=-\f{i}{2}\vec{\sigma}\in\su(2)$ and $\vec{\sigma}=\{\sigma_1,\sigma_2,\sigma_3\}$ are the Pauli matrices. 
Changing the base point corresponds to a conjugation action on $M_\ell(p)$ by an $\SU(2)$ group element, say $g$,
\be
M_\ell(p)\quad\longrightarrow\quad M_\ell(p')=gM_\ell(p)g^{-1}\,,\quad g\in\SU(2)\,.
\ee
In the rest of the paper, we only focus on $\SU(2)$ holonomies at the classical level and, when the base point is not necessarily specified, one denotes the holonomy simply by $M_\ell$. 

We have seen that the holonomy is taken in both the fundamental representation and the adjoint representation. In general, the holonomy can be expressed in an arbitrary irreducible representation of $\SU(2)$.
Every element of $\SU(2)$ is conjugated to an element in the abelian subgroup (maximal torus) consisting of all diagonal matrices in $\SU(2)$. 
The holonomy $M_{\ell}$ in \eqref{eq:H_definition} can be expressed as:
\be
M_{\ell}= g \begin{pmatrix}e^{i\frac{|\Lambda|}{6}a_{\ell}} & 0\\
0 & e^{-i\frac{|\Lambda|}{6}a_{\ell}}
\end{pmatrix}  g^{-1} =g e^{i\frac{|\Lambda|}{6}a_\ell H} g^{-1}\,,\quad
g\in \SU(2)\;,
\ee
where $H$ is the generator of the Cartan subalgebra of $\su(2)$. 
The right-most expression can be written in an arbitrary irreducible representation of $\SU(2)$. In particular, we are interested in the trace $\tr\lb M_\ell\rb$ of the holonomy, which stores the area information of the curved triangle that $\ell$ encircles (see \eqref{eq:half-trace_1}). Considering the irreducible represention labeled by $I\in\N/2$, the trace reads
\be
\tr^I(M^I_{\ell})=\tr^I(g^I (e^{i\frac{|\Lambda|}{6}a_{\ell}H})^I(g^{-1})^I)=\frac{\sin\left((2I+1)\frac{|\Lambda|}{6}a_{\ell}\right)}{\sin\lb\frac{|\Lambda|}{6} a_{\ell}\rb}\;.
\label{eq:area_general}
\ee
This result is the consequence of the conjugation-invariant property of the trace. We will see in Section \ref{sec:area} that Eq.\eqref{eq:area_general} has a straightforward quantum counterpart when we express the quantum version of $M_\ell$ in the $I$ representation. 

Changing the orientation of $\ell$ corresponds to changing $M_\ell$ to its inverse, $\ie$ $M_{\ell^{-1}}=M^{-1}_{\ell}$. 
For each curved tetrahedron, there exists a closure condition expressed as
\be
M_4M_3M_2M_1=\Id\,,\quad M_{\ell}\in\SU(2)\,,
\label{eq:closure_SU2}
\ee
 where all four holonomies are defined at the same base point. Indeed, it is easy to find a common point for three of the four holonomies. One then has to parallel transport the base point at least once through a specified path to define all the holonomies properly. As one of the simplest examples, choosing vertex 4 in fig.\ref{fig:tetrahedra} as the base point, $M_1(4),M_2(4),M_3(4)$ can all be defined directly by \eqref{eq:H_definition}. To define $M_4(4)$, we first define $M_4(2)$ based on vertex 2 by \eqref{eq:H_definition} and parallel transport it to vertex 4 through the edge (42).  We would like to stress here that, although holonomies in the closure condition \eqref{eq:closure_SU2} are written for the fundamental representation, they can be generalized to arbitrary irreducible representation $I\in\N/2$:
 \be
M^I_4M^I_3M^I_2M^I_1=\Id^I\,.
 \ee

A solution to \eqref{eq:closure_SU2} can be given by introducing the edge holonomy $h_{v_1v_2}$ for each oriented edge $(v_1v_2)$ with $h_{v_1v_2}^{-1}=h_{v_2v_1}$. Then
\be
\left\{\ba{l}
M_1=h_{43}h_{32}h_{24}\\
M_2=h_{41}h_{13}h_{34}\\
M_3=h_{42}h_{21}h_{14}\\
M_4=h_{42}M_4(2)h_{24}=h_{42}h_{23}h_{31}h_{12}h_{24}
\ea\right.
\label{eq:simple_solution}
\ee
is indeed a solution to \eqref{eq:closure_SU2}. The paths for the solution \eqref{eq:simple_solution} are illustrated in fig.\ref{fig:simple_pahts} for a spherical tetrahedron as an example and are the same for a hyperbolic tetrahedron. These paths are called the {\it simple paths} as they are the simplest set of paths up to the choice of the base point and the special edge \cite{Haggard:2015ima}.  
\begin{figure}[h!]
\begin{subfigure}[t]{0.2\linewidth}
 \begin{tikzpicture}[scale=2]
  \draw[red, thick,postaction={decorate},decoration={markings,mark=at position 0.6 with {\arrow[scale=1.5,>=stealth]{>}}}] (3,3) arc (0:60:1);
 \draw[black, thick, dashed] (2.5,3.866) arc (100:150:1);
 \draw[red, thick, postaction={decorate},decoration={markings,mark=at position 0.6 with {\arrow[scale=1.5,>=stealth]{>}}}] (2.5,3.866) arc (120:200:1);
 \draw[black, thick, dashed] (2.08,2.667) arc (-135:-180:1);
 \draw[red, thick, postaction={decorate},decoration={markings,mark=at position 0.6 with {\arrow[scale=1.5,>=stealth]{<}}}] (3,3) arc (-40:-100:1);
 \draw[black, thick, dashed] (3,3) arc (-70:-145:1);
 \filldraw[black] (2.5,3.866) circle(1pt) node[anchor=south]{4};
 \filldraw[black] (3,3) circle(1pt) node[anchor=west]{3};
 \filldraw[black] (2.05,2.667) circle(1pt) node[anchor=north]{2};
 \draw (1.8,3.3) node[anchor=east]{1};
 \end{tikzpicture} 
 \caption{}
\label{fig:ciliated_graph}
\end{subfigure}
\begin{subfigure}[t]{0.2\linewidth}
 \begin{tikzpicture}[scale=2]
  \draw[red, thick,postaction={decorate},decoration={markings,mark=at position 0.6 with {\arrow[scale=1.5,>=stealth]{<}}}] (3,3) arc (0:60:1);
 \draw[red, thick,postaction={decorate},decoration={markings,mark=at position 0.6 with {\arrow[scale=1.5,>=stealth]{<}}}] (2.5,3.866) arc (100:150:1);
 \draw[black, thick, dashed] (2.5,3.866) arc (120:200:1);
 \draw[black, thick, dashed] (2.08,2.667) arc (-135:-180:1);
 \draw[black, thick, dashed] (3,3) arc (-40:-100:1);
 \draw[red, thick, postaction={decorate},decoration={markings,mark=at position 0.6 with {\arrow[scale=1.5,>=stealth]{>}}}] (3,3) arc (-70:-147:1);
 \filldraw[black] (2.5,3.866) circle(1pt) node[anchor=south]{4};
 \filldraw[black] (3,3) circle(1pt) node[anchor=west]{3};
 \draw[black] (2.05,2.667) node[anchor=north]{2};
 \filldraw[black] (1.8,3.4) circle(1pt) node[anchor=east]{1};
 \end{tikzpicture} 
 \caption{}
\label{fig:ciliated_graph}
\end{subfigure}
\begin{subfigure}[t]{0.2\linewidth}
 \begin{tikzpicture}[scale=2]
  \draw[black,thick,dashed] (3,3) arc (0:60:1);
 \draw[red, thick,postaction={decorate},decoration={markings,mark=at position 0.6 with {\arrow[scale=1.5,>=stealth]{>}}}] (2.5,3.866) arc (100:150:1);
 \draw[red, thick, postaction={decorate},decoration={markings,mark=at position 0.6 with {\arrow[scale=1.5,>=stealth]{<}}}] (2.5,3.866) arc (120:200:1);
 \draw[red, thick, postaction={decorate},decoration={markings,mark=at position 0.6 with {\arrow[scale=1.5,>=stealth]{<}}}] (2.08,2.667) arc (-135:-180:1);
 \draw[black, thick, dashed] (3,3) arc (-40:-100:1);
 \draw[black, thick, dashed] (3,3) arc (-70:-147:1);
 \filldraw[black] (2.5,3.866) circle(1pt) node[anchor=south]{4};
 \draw[black] (3,3) node[anchor=west]{3};
 \filldraw[black] (2.05,2.667) circle(1pt) node[anchor=north]{2};
 \filldraw[black] (1.8,3.4) circle(1pt) node[anchor=east]{1};
 \end{tikzpicture} 
 \caption{}
\label{fig:ciliated_graph}
\end{subfigure}
\begin{subfigure}[t]{0.2\linewidth}
 \begin{tikzpicture}[scale=2]
  \draw[black,thick,dashed] (3,3) arc (0:60:1);
 \draw[black, thick,dashed] (2.5,3.866) arc (100:150:1);
 \draw[red, thick, postaction={decorate},decoration={markings,mark=at position 0.7 with {\arrow[scale=1.5,>=stealth]{>}}}] (2.48,3.864) arc (120:200:1);
 \draw[red, thick, postaction={decorate},decoration={markings,mark=at position 0.5 with {\arrow[scale=1.5,>=stealth]{<}}}] (2.53,3.866) arc (120:200:1);
 \draw[red, thick, postaction={decorate},decoration={markings,mark=at position 0.6 with {\arrow[scale=1.5,>=stealth]{>}}}] (2.04,2.667) arc (-138:-182:1);
 \draw[red, thick, postaction={decorate},decoration={markings,mark=at position 0.6 with {\arrow[scale=1.5,>=stealth]{>}}}] (3,3) arc (-40:-100:1);
 \draw[red, thick, postaction={decorate},decoration={markings,mark=at position 0.6 with {\arrow[scale=1.5,>=stealth]{<}}}] (3,3) arc (-70:-126:1);
 \draw[red, thick, postaction={decorate},decoration={markings,mark=at position 0.6 with {\arrow[scale=1.5,>=stealth]{>}}}] (1.8,3.38) arc (-140:-126:1);
 \filldraw[black] (2.5,3.866) circle(1pt) node[anchor=south]{4};
 \filldraw[black] (3,3) circle(1pt) node[anchor=west]{3};
 \filldraw[black] (2.05,2.667) circle(1pt) node[anchor=north]{2};
 \filldraw[black] (1.8,3.4) circle(1pt) node[anchor=east]{1};
 \end{tikzpicture} 
 \caption{}
\label{fig:ciliated_graph}
\end{subfigure}
\caption{The set of simple paths ({\it in red}) for holonomies $\{M_1,M_2,M_3,M_4\}$ defined in \eqref{eq:simple_solution} with vertex 4 as the base point and edge (42) as the special edge. They satisfy the closure condition \eqref{eq:closure_SU2}. }
\label{fig:simple_pahts}
\end{figure}
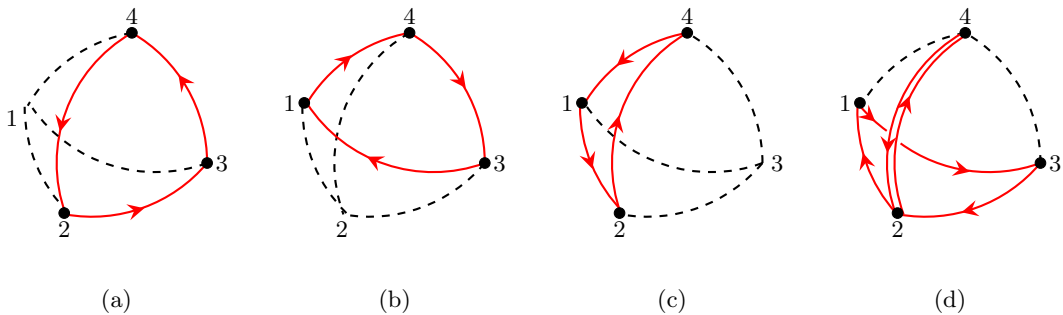

Note that, given a curved tetrahedron whose curvature is $\Lambda$, the full geometrical information can be described by the four holonomies $M_\ell$'s. We collect these geometrical interpretations in Appendix \ref{app:GeoTetra}. What is important for the rest of the construction is the sign, denoted by $\sgn(\det \Gram(M_\ell))$, of the determinant of the Gram matrix $\Gram(M_\ell):=\Gram(\cos\theta_{\ell_1\ell_2})$ defined in \eqref{eq:cos_H}. 

The key result of \cite{Haggard:2015ima} is that, reversely, each closure condition written in the form of \eqref{eq:closure_SU2} with a few extra restrictions identically defines a curved tetrahedron up to translation. This is called the {\it curved Minkowski theorem} for tetrahedra. In the flat limit, it coincides with the well-known Minkowski theorem for flat tetrahedra which was proven in 1897 \cite{minkowski1897allgemeine}. 
The curved Minkowski theorem is stated as follows. 
\begin{theorem}[The curved Minkowski theorem for tetrahedron]
\label{theorem:Minkowski}
\cite{Haggard:2015ima} Given four $\SU(2)$ holonomies $M_\ell$'s satisfying the non-degeneracy condition $\det \Gram(M_\ell)\neq 0$ and the closure condition $M_4M_3M_2M_1=\Id$, one can uniquely determine a non-degenerate curved tetrahedron in the following way.
\begin{enumerate}
	\item Label the sub-simplices of the tetrahedron as in fig.\ref{fig:tetrahedra}. The tetrahedron is flatly embedded in $S^3$ if $\sgn(\det \Gram(M_\ell))>0$ and flatly embedded in $\bH^3$ if $\sgn(\det \Gram(M_\ell))<0$;
	\item The holonomies $M_\ell$'s are associated to a set of simple paths with either the base point at vertex 4 and special edge (42) or the base point at vertex 3 and special edge (31) and the orientation of the paths determine the orientation of the face surrounded by the path;
	\item Each holonomy $M_\ell$ encodes the area $a_\ell$ of face $\ell$ and the outward direction normal $\hat{n}_\ell$ (when parallel transported to the base point) in its parametrization $M_\ell=\exp\lb s \frac{|\Lambda|}{6} a_\ell \hat{n}_\ell \cdot\vec{\tau}\rb$ with $s:=\sgn(\det \Gram(M_\ell))$.
\end{enumerate}
\end{theorem}
We refer interested readers to \cite{Haggard:2015ima} for detailed proof of this theorem. Let us only make two comments on the theorem. Firstly, the requirement of $\det \Gram(M_\ell)\neq 0$ is to exclude the degenerate curved tetrahedra, which cannot be reconstructed uniquely with holonomies in the above way. $\det \Gram(M_\ell)=0$ does not correspond to a flat tetrahedron for nonzero $|\Lambda|$. Indeed, the theorem is based on a constant $\Lambda$, while the flat limit corresponds to $\Lambda\rightarrow 0$, and the linearization of the closure condition recovers the closure condition for the flat tetrahedron. Secondly, the two choices of pairs of (base point, special edge), $\ie$ (4, (42)) and (3, (31)) reproduce the same curved tetrahedron. The solution of the four holonomies $M_\ell$'s in terms of the edge holonomies $h_{v_1v_2}$'s for the first case is given in \eqref{eq:simple_solution} and those for the second case is 
\be
\left\{\ba{l}
M_1=h_{34}h_{42}h_{23}\\
M_2=h_{31}h_{14}h_{43}\\
M_3=h_{31}M_3(1)h_{13}=h_{31}h_{12}h_{24}h_{41}h_{13}\\
M_4=h_{32}h_{21}h_{13}
\ea\right..
\label{eq:simple_solution_2}
\ee
These holonomies will be used to describe the moduli space of flat connections on a four-punctured sphere in the next section, which are closely related to the curved tetrahedron encoded in $\{M_\ell\}$. 

\section{Moduli space $\cM_{\Flat}$ of flat connections}
\label{sec:moduli_space}

In the previous section, we have seen that a homogeneously curved tetrahedron can be one-to-one described by four $\SU(2)$ holonomies satisfying the closure constraint (with some extra conditions) as given in Theorem \ref{theorem:Minkowski}. The closure constraint \eqref{eq:closure_SU2} appears in the same shape as the defining equation for the moduli space of flat connection on a four-punctured sphere, denoted as $\Sfour$ (0 denoting the genus of a sphere). The goal of this section is to show that the $\SU(2)$ holonomies fully describing a homogeneously curved tetrahedron are exactly those describing the moduli space of $\su(2)$ flat connection on a $\Sfour$. 

Let us consider a {\it simple graph} on $\Sfour$ containing one node and four loops $\ell_1,\ell_2,\ell_3,\ell_4$, each of which starts and ends at the node and surrounds a puncture as illustrated in fig.\ref{fig:simple_graph}. The orientations of the loops match that of the two-sphere. 
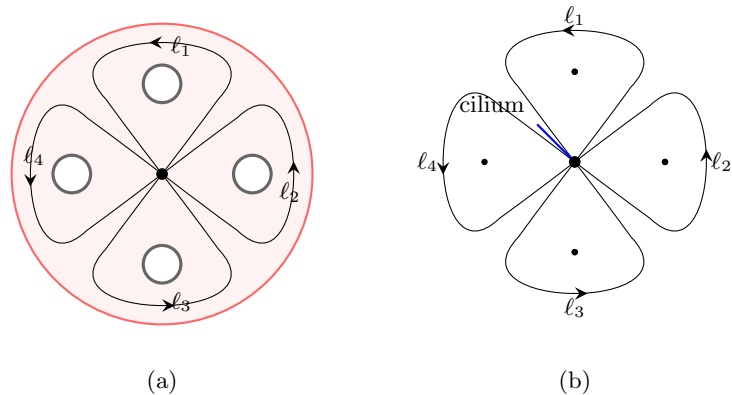
\begin{figure}
\begin{subfigure}[t]{0.3\linewidth}
\begin{tikzpicture}
 \filldraw[color=red!60, fill=red!5, thick](0,0) circle (2);
 \puncture{0,0}{0.55,0.42}{1,0.75}{2,1.7}{2,-1.7}{1,-0.75}{0.55,-0.42}; 
 \puncture{0,0}{-0.42,0.55}{-0.75,1}{-1.7,2}{1.7,2}{0.75,1}{0.42,0.55};
 \puncture{0,0}{-0.55,-0.42}{-1,-0.75}{-2,-1.7}{-2,1.7}{-1,0.75}{-0.55,0.42};
 \puncture{0,0}{0.42,-0.55}{0.75,-1}{1.7,-2}{-1.7,-2}{-0.75,-1}{-0.42,-0.55};
 \filldraw[black] (0,0) circle(2pt);
 \draw (0,1.7) node[anchor=west] {$\ell_1$};
 \filldraw[black] (0,1.2) circle(1pt);
 \draw (0,-1.7) node[anchor=west] {$\ell_3$};
 \filldraw[black] (0,-1.2) circle(1pt);
  \draw (1.7,0) node[anchor=north] {$\ell_2$};
 \draw (-1.7,0) node[anchor=south] {$\ell_4$};
 \filldraw[black] (1.2,0) circle(1pt);
 \filldraw[color=black!60, fill=white!5, very thick](-1.2,0) circle (0.25);
 \filldraw[color=black!60, fill=white!5, very thick](1.2,0) circle (0.25);
 \filldraw[color=black!60, fill=white!5, very thick](0,-1.2) circle (0.25);
 \filldraw[color=black!60, fill=white!5, very thick](0,1.2) circle (0.25);
\end{tikzpicture}
\caption{}
\label{fig:simple_graph}
\end{subfigure}
\begin{subfigure}[t]{0.3\linewidth}
\begin{tikzpicture}
   \pun{0,0}{0.55,0.42}{1,0.75}{2,1.7}{2,-1.7}{1,-0.75}{0.55,-0.42}; 
 \pun{0,0}{-0.42,0.55}{-0.75,1}{-1.7,2}{1.7,2}{0.75,1}{0.42,0.55};
 \pun{0,0}{-0.55,-0.42}{-1,-0.75}{-2,-1.7}{-2,1.7}{-1,0.75}{-0.55,0.42};
 \pun{0,0}{0.42,-0.55}{0.75,-1}{1.7,-2}{-1.7,-2}{-0.75,-1}{-0.42,-0.55}; 
 \draw[thick,blue] (0,0) -- (-0.5,0.5);
 \draw (-0.55,0.55) node[anchor=south east]{cilium};
 \filldraw[black] (0,0) circle(2pt);
 \draw (0,1.7) node[anchor=south] {$\ell_1$};
 \filldraw[black] (0,1.2) circle(1pt);
 \draw (0,-1.7) node[anchor=north] {$\ell_3$};
 \filldraw[black] (0,-1.2) circle(1pt);
  \draw (1.7,0) node[anchor=west] {$\ell_2$};
 \draw (-1.7,0) node[anchor=east] {$\ell_4$};
 \filldraw[black] (1.2,0) circle(1pt);
 \filldraw[black](-1.2,0) circle (1pt);
 \filldraw[black](1.2,0) circle (1pt);
 \filldraw[black](0,-1.2) circle (1pt);
 \filldraw[black](0,1.2) circle (1pt);
\end{tikzpicture}
\caption{}
\label{fig:ciliated_graph}
\end{subfigure}
\caption{{\it (a)} A four-punctured sphere $\Sfour$ ({\it surface in red}) and the simple graph ({\it loops in black})  on it. Every loop $\ell_\nu\,(\nu=1,2,3,4)$ surrounding the $\nu$-th puncutre in the simple graph is oriented counter-clockwise when seen outside $\Sfour$. {\it (b)} A ciliated graph based on the simple graph in (a). A cilium ({\it in blue}) is added between loops $\ell_1$ and $\ell_4$ to fix the linear order $\ell_1\prec\ell_2\prec\ell_3\prec\ell_4$.}
\end{figure}
(The homotopy equivalence classes of) these loops generate the fundamental group $\pi_1(\Sfour)$ on the four-punctured sphere 
\be
\pi_1(\Sfour)=\{\ell_1,\ell_2,\ell_3,\ell_4: \ell_4\circ\ell_3\circ\ell_2\circ\ell_1=1\}\,.
\ee
The representation of $\pi_1(\Sfour)$ can be used to define the moduli space of $G$ flat connections on $\Sfour$ as follows. 
\be
\cM_{\Flat}^0(\Sfour, G):=\{\text{flat } G \text{ connection on } \Sfour \}/\text{gauge}\simeq 
\Hom(\pi_1(\Sfour),G)/G\,,
\ee
where the quotient in the last expression is by the conjugate action of $G$. This means the moduli space of $\SU(2)$ flat connections can be represented as
\be
\cM_{\Flat}^0(\Sfour, \SU(2))=\{G_1,G_2,G_3,G_4\in\SU(2):G_4G_3G_2G_1=\Id_{\SU(2)}\}/\SU(2)\,,
\label{eq:flat_connection_def}
\ee
where each $G_\nu\, (\nu=1,2,3,4)$ is the $\SU(2)$ holonomy of the loop $\ell_\nu$ around the $\nu$-th puncture. One immediately observes the form of the closure condition in \eqref{eq:flat_connection_def}. 
Remarkably, if we identify these holonomies $\{G_\nu\}$ with the holonomies $\{M_\nu\}$ that identify a constant curved tetrahedron as described in Theorem \ref{theorem:Minkowski}, $\cM_{\Flat}(\Sfour, \SU(2))$ 
can be one-to-one mapped to a constant curved tetrahedron (up to orientation).

$\cM_{\Flat}^0(\Sfour, \SU(2))$ is equipped with a Poisson bracket induced from the Chern-Simons theory on the three-manifold $\Sfour\times \R$
\be
\{A^a_i(x_1),A^b_j(x_2)\}=-\f{2\pi}{k} \delta^{ab}\epsilon_{ij}\delta^{(2)}(x_1-x_2)\,,\quad x_1,x_2\in \Sfour\,,
\label{eq:Poisson_cont}
\ee
where $A$ is an $\su(2)$-valued one-form, $a,b$ are Lie algebra indices, $i,j$ are coordinate indices on $\Sigma_{0,4}$, and the prefactor $-\f{2\pi}{k}$ is introduced here to relate to the Chern-Simons theory with $k\in\N$ being the Chern-Simons level. We are interested in the symplectic leaves obtained from $\cM_{\Flat}^0(\Sfour, \SU(2))$ by fixing the conjugacy classes of the holonomies $\{G_\nu\}_{\nu=1,2,3,4}$ each labelled by the eigenvalue $\lambda_\nu$, which relates to the area of the face $\nu$ in the geometrical interpretation of tetrahedron. 
We denote this symplectic space as $\cM_{\Flat}(\Sfour, \SU(2))$. This is the phase space of the Chern-Simons theory equipped with the symplectic form known as the Atiyah-Bott-Goldmann two-form \cite{Atiyah:1982fa,goldman1984symplectic,Alekseev:1993rj}:
\be
\Omega=\int_{\Sigma_{0,4}}\tr\lb \delta A \wedge \delta A \rb\,,
\label{eq:symplectic_form}
\ee
where $\delta$ is the exterior differential on the phase space and $\tr$ is the non-degenerate Killing form on $\su(2)$. 

Quantizing \eqref{eq:Poisson_cont} directly leads to an infinite-dimensional quantum algebra which is complicated to deal with. An alternative approach is called the {\it combinatorial quantization}  \cite{Alekseev:1994pa,Alekseev:1994au,Alekseev:1995rn} whose classical setup 
follows the description of Fock and Rosly \cite{Fock:1998nu}. 
The idea is to replace the Riemann surface, here $\Sfour$, by a so-called {\it ciliated fat graph} and represent the moduli space $\cM_{\Flat}(\Sfour,\SU(2))$ by a finite-dimensional space whose Poisson structure of graph connections is consistent with the one given in \eqref{eq:Poisson_cont}. The consistency is checked by the fact that the Poisson brackets of observables recover those in the continuous theory. 

To build the ciliated flat graph on $\Sfour$, we start with the simple graph as shown in fig.\ref{fig:simple_graph}.  
The cyclic order of (half-)links incident to the node is set to be clockwise.  
 On top of that, we also (randomly) fix a linear order of these links. One can graphically add a {\it cilium} sitting at the node to separate the first and the last links. We say that link $l_1$ is of a lower order than $l_2$, denoted as $l_1\prec l_2$ or $l_2\succ l_1$, if the cilium sweeps through $l_1$ before $l_2$ in the clockwise direction. An example of linear order is illustrated in fig.\ref{fig:ciliated_graph}. 
The Poisson structure is represented by the Poisson brackets of holonomies $\{U_l\}$ along the links $l$'s. The linear order on the node is important to present the full Poisson algebra \cite{Fock:1998nu}. 

To express the Poisson bracket in a concise way, one makes use of the {\it classical $r$-matrix} $r\equiv \sum_a r_a^{(1)}\otimes r_a^{(2)}\in \su(2)\otimes\su(2)$ which is a solution to the {\it classical Yang-Baxter equation} (CYBE) \cite{Alekseev:1994pa,chari1995guide,kosmann2004lie}:
\be
[r_{12},r_{13}]+[r_{12},r_{23}]+[r_{13},r_{23}]=0\,,
\label{eq:CYBE}
\ee
where $r_{12}=\sum_a r_a^{(1)}\otimes r_a^{(2)}\otimes \Id,r_{13}=\sum_a r_a^{(1)}\otimes \Id\otimes r_a^{(2)}$ and $r_{23}=\sum_a  \Id \otimes r_a^{(1)}\otimes r_a^{(2)}$. 
We also denote the transpose of the $r$-matrix by $r':=\sum_a r_a^{(2)}\otimes r_a^{(1)}$. 
Under the basis $\tau_a=\f{1}{2i}\sigma_a$ of $\su(2)$, the $r$-matrix can be expressed as 
\be
r=\tau_3\otimes \tau_3 + 2\,\tau_+\otimes \tau_-\,,\quad
r'=\tau_3\otimes \tau_3 + 2\,\tau_-\otimes \tau_+\,,
\label{eq:r_su2}
\ee
where $\tau_\pm\equiv \tau_1\pm i \tau_2$. 
The full algebra is rather complicated for a general Riemann surface. Here we only give the two types of Poisson brackets relevant to us and refer interested readers to \cite{Fock:1998nu} for other types\footnotemark{}. Denote the holonomy along a loop $\ell$ by $M_\ell$ and $\Ml_\ell:=M_\ell\otimes \Id, \Mr_\ell:=\Id \otimes M_\ell$. 
\begin{enumerate}[i]
	\item For a single loop $\ell$, the Poisson bracket of the holonomy $M_\ell$ reads
\be
\{\Ml_\ell, \Mr_\ell \} = r \Ml_\ell \Mr_\ell + \Mr_\ell r'\Ml_\ell - \Ml_\ell r \Mr_\ell - \Ml_\ell \Mr_\ell r'\,.
\label{eq:Poisson_1}
\ee
\item For two loops $\ell$ and $\ell'$ with linear order $\ell\prec \ell'$, the Poisson bracket of the holonomies $M_\ell$ and $M_{\ell'}$ reads
\be
\{\Ml_\ell, \Mr_{\ell'}\} = r \Ml_\ell \Mr_{\ell'}+\Ml_\ell \Mr_{\ell'} r - \Mr_{\ell'}r\Ml_\ell -\Ml_\ell r \Mr_{\ell'}\,.
\label{eq:Poisson_2}
\ee
\end{enumerate}
 \footnotetext{On a general graph embedded on a Riemann surface, a classical $r$-matrix is assigned to each node and only the symmetric part of these $r$-matrices are required to be equal for all nodes \cite{Fock:1998nu}. The CYBE \eqref{eq:CYBE} ensures that the Poisson brackets expressed in terms of the $r$-matrices satisfy the Jacobi identity.}

As in the lattice gauge theory, the gauge transformation acts at the nodes of the graph. For a given oriented link $l$, denote the source and target as $s(l)$ and $t(l)$ respectively. The holonomy $U_l$ transforms under gauge transformation as
\be
U_l \quad\rightarrow \quad U_l^g=g_{s(l)}^{-1} \,U_l\, g_{t(l)}\,.
\label{eq:gauge_symmetry}
\ee
In the simple graph that we are interested in, only one node is relevant hence the gauge transformation acts on the holonomies by conjugation
\be
M_\ell \quad\rightarrow \quad M_\ell^g = g^{-1}M_\ell g\,.
\label{eq:gauge_symmetry2}
\ee
To preserve the Poisson brackets of holonomies, the gauge transformation $g_v$ at each node $v$ is equipped with the following Poisson brackets.
\be
\{\gl_v,\gr_v\}=\gl_v\gr_vr-r\gr_v\gl_v\,,\quad
\{\gl_v,\gr_{v'}\}=0\,,\quad v\neq v'\,.
\label{eq:Poisson_Lie}
\ee
It shows that the symmetry group forms a Poisson Lie group. 
 Observables are constructed by functions of holonomies which are invariant under the gauge transformation. Examples are the traces of products of holonomies as given in \eqref{eq:half-traces}.
Our goal next is to quantize the holonomy algebra given by \eqref{eq:Poisson_1} and \eqref{eq:Poisson_2} as well as the gauge symmetry \eqref{eq:gauge_symmetry2} into quantum algebras and construct the quantum observables.


\section{Quantization of $\cM_{\text{flat}}(\Sfour,\SU(2))$}
\label{sec:quantum_moduli}

Thanks to the expression with the classical $r$-matrix, the Poisson structure of $\cM_{\Flat}(\Sfour,\SU(2))$ possesses a natural quantization given by the commutation relations in terms of quantum matrices and the quantum version of the $r$-matrix.
The quantum theory is described 
based on the so-called quasitriangular ribbon (quasi) Hopf algebra $\UQ$ which is a deformation of the universal enveloping algebra of $\su(2)$ with a deformation parameter $q\in\bC$ and some extra algebraic structures, \eg the $*$-structure. 
It is important to note that $\UQ$ has different algebraic as well as representation structures for different values of $q$. We are in particular interested in the case that $q$ is a root-of-unity, \ie $q^p=1$ for some $p\in\N_+,p\geq 2$ and that $q$ is a phase but not a root-of-unity, \ie $|q|=1,q^p \neq 1\,,\forall q\in\N_+$. (In this paper, we denote the case of $q$ root-of-unity as $q^p=1$ and $q$ not a root-of-unity as $|q|=1,q^p\neq1$.) 
$\UQ$ is a Hopf algebra in the latter case while it is a quasi Hopf algebra in the former case (\eg \cite{chari1995guide,majid2000foundations}). 
We collect the necessary knowledge of $\UQ$ for both cases and their representation theory in Appendix \ref{sec:math}. 

The $\SU(2)$ gauge transformation is also quantized to the non-commutative quantum symmetry described by the (quasi) Hopf-$*$ algebra $\UQ$ in order to preserve the commutation relations. In this section, we first construct the quantum observables which are conservative under the quantum symmetry on $\Sfour$. They are mathematically described by the {\it invariant algebra}. These quantum observables are then used to build up the {\it moduli algebra} which is recognized as the quantum version of $\cM_\Flat(\Sfour,\SU(2))$. We will finally construct the representation of this moduli algebra. Both cases of $|q|=1,q^p\neq1$ and $|q|=1,q^p=1$ will be considered in this section without specifying unless necessary and keeping in mind that admissible irreducible representations are $I\in\N/2$ in the former case but only $I\in\N/2,\ 0\leq 2I \leq p-2$ in the latter case. $q$ being a root of unity relates to the SU(2) Chern-Simons theory by $q=\exp(\frac{2\pi i}{k+2})$, \ie $p=k+2$. Both cases can describe the quantization of a curved tetrahedron. 

\subsection{Quantum holonomies and quantum symmetry}

Denote the quantum holonomy for a loop $\ell$ as $\bM_\ell$. Given the quantum version of the $r$-matrix, the Poisson brackets \eqref{eq:Poisson_1} and \eqref{eq:Poisson_2} can be naturally quantized to commutators of quantum holonomies. Explicitly, \eqref{eq:Poisson_1} is quantized to
\be
R^{-1}\,\bMl_\ell\, R\,\bMr_\ell = \bMr_\ell\, R'\, \bMl_\ell\, {R'}^{-1}
\label{eq:Commutation_1}
\ee
and \eqref{eq:Poisson_2} is quantized to
\be
R^{-1}\,\bMl_\ell\,R\,\bMr_{\ell'} = \bMr_{\ell'}\,R^{-1}\,\bMl_\ell\,R\,,
\label{eq:Commutation_2}
\ee
where $R$ is the quantum $\cR$-matrix of $\UQ$ defined in \eqref{eq:R_def}. 
One can easily check that they recover the Poisson brackets at the first-$\hbar$ order by taking the expansion \eqref{eq:R_r} of $R$.  
We are more interested in the irreducible representations $\rho^I\otimes\rho^J$ of the above commutation relations:
\begin{subequations}
\begin{align}
(R^{-1})^{IJ}\,\bMlI\, R^{IJ}\,\bMrJ &= \bMrJ\, (R')^{IJ}\, \bMlI\, ({R'}^{-1})^{IJ}\,,
\label{eq:commutation_rep_1}
\\
(R^{-1})^{IJ}\,\bMlI\,R^{IJ}\,\bMrJp &= \bMrJp\,(R^{-1})^{IJ}\,\bMlI\,R^{IJ}\,,\quad  \ell\prec\ell'\,,
\label{eq:commutation_rep_2}\\
(R')^{IJ}\,\bMlI\,(R'^{-1})^{IJ}\,\bMrJp &= \bMrJp\,(R')^{IJ}\,\bMlI\,(R'^{-1})^{IJ}\,,\quad  \ell\succ\ell',.
\label{eq:commutation_rep_3}
\end{align}
\label{eq:commutation_rep}
\end{subequations}
We demand that inverting the orientation of the loop $\ell\rightarrow \ell^{-1}$ maps $\bM^I_\ell$ to $\bM^I_{\ell^{-1}}$ such that
\be
\bM^I_\ell \bM^I_{\ell^{-1}}=\bM^I_{\ell^{-1}} \bM^I_\ell=e^I\,.
\label{eq:inverse_M}
\ee

As the classical holonomies $\{M_{\ell}\}$ are in $\SU(2)$. A product of two holonomies for the same loop in two different representations should admit the recoupling theory of $\SU(2)$ by means of the CG maps. Explicitly,
\be
\MlI \MrJ=\sum_K C_0[IJ|K]^* M_\ell^K C_0[IJ|K]\,,
\label{eq:M_CG}
\ee
where $C_0[IJ|K]:V^I\otimes V^J\rightarrow V^K$ and $C_0[IJ|K]^*:V^K\rightarrow V^I\otimes V^J$ are the CG maps of $\SU(2)$. 
Such a recoupling relation needs to be deformed for $\bM_\ell$ in order to preserve the commutation relations \eqref{eq:commutation_rep}. The result is given in terms of the CG maps for $\UQ$ (see \eqref{eq:CG_def1}) and \eqref{eq:CG_def2}) and reads
\begin{align}
 \bMlI\,R^{IJ}\, \bMrJ = \sum_K C[IJ|K]^* \bM_\ell^K C[IJ|K]\,.
 \label{eq:holonomy_recoupling}
\end{align}
Mathematically, the quantum holonomies $\{\bM^I_\ell\}$ of the simple graph on $\Sfour$ can be used to define the  {\it graph algebra} denoted as $\cL_{0,4}$.

\medskip

\noindent {\it Graph algebra $\cL_{0,4}$. }
The graph algebra $\cL_{0,4}$ for the simple graph $\Gamma$ on $\Sfour$ is generated by the matrix elements of $\{\bM^I_\ell\}_{I\in\N/2,\ell\in\Gamma}$ associated to loops $\ell$'s in $\Gamma$ satisfying the commutation relations \eqref{eq:commutation_rep}, the invertibility relation \eqref{eq:inverse_M} and the recoupling relations \eqref{eq:holonomy_recoupling} \cite{Alekseev:1994au,Alekseev:1995rn}. Here, $I$ runs through all the irreducible representations of $\UQ$. 
In this sense, one can understand $\bM_\ell^I$ as an element in $\End(V^I)\otimes \cL_{0,4}$ and each matrix element is in $\cL_{0,4}$.
Consistently, the Poisson bracket \eqref{eq:Poisson_Lie} of the gauge transformation elements is quantized to 
\be
R\,\bgl\bgr = \bgr\bgl R\,,\quad \bgl\bgrp=\bgrp\bgl\,,\quad v\neq v'\,,
\label{eq:quantum_symmetry}
\ee
where ${\bf g}_v$ is the quantized version of $g_v$. This commutation relation reveals that the underlying quantum symmetry is $\UQ$. 

We define the action of a $\UQ$ element $\xi$ on the generators $\bM^I_\ell$ of $\cL_{0,4}$ as follows.
\be
\xi(\bM^I_\ell)=\sum_a \rho^I(S(\xi^{(1)}_a))\,\bM^I_\ell\,\rho^I(\xi^{(2)}_a)\,.
\label{eq:quantum_gauge_action}
\ee
This is the quantization of the conjugacy action \eqref{eq:gauge_symmetry2} and can be checked to preserve \eqref{eq:commutation_rep} subject to the commutation relation \eqref{eq:quantum_symmetry} of the gauge symmetry. 
Take $\xi$ to be the generators $\{X,Y,q^{\f{H}{2}}\}$ of $\UQ$, the actions are explicitly
\begin{subequations}
\begin{align}
q^{\frac{H}{2}}(\bM_\ell^I)&=\rho^{I}(q^{-\frac{H}{2}})\bM_\ell^{I}+\bM_\ell^{I}\rho^{I}(q^{\frac{H}{2}})\,,\\
X(\bM_\ell^{I})&=-q\rho^{I}(X)\bM_\ell^{I}\rho^{I}(q^{\f{H}{2}})+\rho^{I}(q^{\f{H}{2}})\bM_\ell^{I}\rho^{I}(X)\,,\\
Y(\bM_\ell^{I})&=-q^{-1}\rho^{I}(Y)\bM_\ell^{I}\rho^{I}(q^{\f{H}{2}})+\rho^{I}(q^{\f{H}{2}})\bM_\ell^{I}\rho^{I}(Y)\,.
\end{align}
\end{subequations}
The action can be generalized to the whole $\cL_{0,4}$ by the property
\be
\xi(\alpha\beta)=\sum_a\xi^{(1)}_a(\alpha)\xi^{(2)}_a(\beta)\,,\quad \forall \alpha,\beta\in\cL_{0,4}\,.
\ee

Especially, for a one-punctured sphere, there is only one holonomy $\bM_\ell$ along the single puncture. In this case, the graph algebra $\cL_{0,1}$ generated by matrix elements of $\bM_\ell^I$ is also called the {\it loop algebra}. We denote the loop algebra corresponding to the $\nu$-th puncture ($\nu=1,2,3,4$) of $\Sfour$ as $\cL_\nu$ and denote the matrix of its generators as $\bM_\nu^I$ for notation conciseness. The centers of these loop algebras (introduced below) will be shown to be important for defining the moduli algebra. 

As a graph algebra, $\cL_{0,4}$ is not equipped with a natural $*$-operation. However, to define the observables later, it is necessary to define a $*$-algebra. To this end, we consider the semi-direct product of $\cL_{0,4}$ and the gauge algebra $\UQ$ and construct $\cS_{0,4}=\UQ\ltimes \cL_{0,4}$ \cite{Alekseev:1994au,Alekseev:1995rn}. 
It is generated by generators of $\cL_{0,4}$ and $\mu^I(q^{\frac{H}{2}}),\mu^I(X)$ and $\mu^I(Y)$ defined in the following way.
\be
\mu^I(\xi):=(\rho^I\otimes\Id)\Delta(\xi)\,,\quad \forall \xi\in\UQ\,.
\label{eq:def_mu_xi}
\ee
The commutation relation between $\mu^I(\xi)$ and $\bM_\ell^I$ is\footnotemark{}
\be
\mu^I(\xi)\bM_\ell^I=\bM_\ell^I\mu^I(\xi)\,,
\label{eq:commute_xi_M}
\ee
which is also called the covariance property. 
\footnotetext{
The commutation relation \eqref{eq:commute_xi_M} takes a simple form because we are studying a simple graph on $\Sfour$. For a general graph with more nodes, the commutation relation of a gauge element, say $\xi_v$, on node $v$ and a quantum holonomy, say $U^I_l$, on link $l$ also relies on the relative location between $v$ and $l$. See \cite{Alekseev:1994pa,Alekseev:1994au} for more details.
}
One can now define the $*$-operation on $\bM_\ell^I$ which is now viewed as an element in $\End(V^I)\otimes \cS_{0,4}$:
\be
(\bM_\ell^I)^*:=(e^I\otimes\kappa^{-1}) (R^I\bM^I_{\ell^{-1}}(R^{-1})^I)\,(e^I\otimes \kappa)
\in\End(V^I)\otimes \cS_{0,4}\,.
\label{eq:star_graph_algebra}
\ee
Here, $R^I:=(\rho^I\otimes\Id)(R)$ and similarly for $(R^{-1})^I$. $\kappa$ is the unitary (in the sense that $\kappa^*=\kappa^{-1}$) element of $\UQ$ obtained from the central element $v\equiv \kappa^2$ defined in \eqref{eq:def_v}. Such a $*$-operation preserves the commutation relation \eqref{eq:commutation_rep} when extended to $\cS_{0,4}$.

\subsection{Moduli algebra $\fM_{0,4}^{K_\nu}$ as the quantization of $\cM_\Flat(\Sfour,\SU(2))$ }
\label{subsec:moduli_algebra}

The graph algebra $\cL_{0,4}$ contains the quantized Poisson structure of quantum holonomies in the simple graph $\Gamma$ and it allows us to define the quantum gauge action. To proceed, we look for the subalgebra of $\cL_{0,4}$ which is invariant under the quantum gauge action \eqref{eq:quantum_gauge_action} hence the algebra of observables. It is called the {\it invariant algebra}, denoted as $\cA_{0,4}$ \cite{Alekseev:1994pa,Alekseev:1994au,Alekseev:1995rn}. 

\medskip

\noindent {\it Invariant algebra $\cA_{0,4}$. }
The algebra $\cA_{0,4}$ is defined as a subalgebra of the graph algebra $\cL_{0,4}$ containing all elements in $\cL_{0,4}$ that are invariant with respect to the action \eqref{eq:quantum_gauge_action} of $\UQ$, \ie 
\be
\cA_{0,4}=\{A\in\cL_{0,4}|\xi(A)=A\epsilon(\xi)\}\,.
\ee
The elements of $\cA_{0,4}$ are linear combinations of the form \cite{Alekseev:1995rn} 
\be
\tr_{q}^{I}(C[I_1I_2I_3I_4|I]\bM_{1}^{I_1}\bM_{2}^{I_2}\bM_{3}^{I_3}\bM_{4}^{I_4}C[I_1I_2I_3I_4|I]^*)\,,
\label{eq:generator_A04}
\ee
where $\tr_q^I$ is the quantum trace defined in \eqref{eq:quantum_trace_g} and $C[I_1I_2I_3I_4|I]$ is the intertwining map for $\UQ$ action defined similarly as in \eqref{eq:CG_def1} and $C[I_1I_2I_3I_4|I]^*$ is defined similarly as in \eqref{eq:CG_def2}. Explicitly, they are defined through the intertwining relations 
\begin{align}
C[I_1I_2I_3I_4|I]\lb\rho^{I_1}\otimes\rho^{I_2}\otimes\rho^{I_3}\otimes\rho^{I_4}\rb\Delta^{(3)}(\xi)&=\rho^J(\xi)C[I_1I_2I_3I_4|I]\,,\\
\lb\rho^{I_1}\otimes\rho^{I_2}\otimes\rho^{I_3}\otimes\rho^{I_4}\rb\Delta^{'(3)}(\xi)C[I_1I_2I_3I_4|I]^*&=C[I_1I_2I_3I_4|I]^*\rho^I(\xi)\,.
\end{align}
$C[I_1I_2I_3I_4|I]$ can be separated into a series of three CG maps and different ways of separation are related through the $6j$-symbols of $\UQ$ as in \eqref{eq:CG-6j-identity_1} for the case $|q|=1,q^p\neq 1$ and as in \eqref{eq:CG-6j-identity} for the case $q^p=1$. 
The $*$-operation on $\cA_{0,4}$ inherits from that of $\cS_{0,4}$.

For the loop algebra $\cL_\nu$, the invariant subalgebra is simply proportional to (the representation of) the quantum trace $\tr^I_q(\bM_{\nu}^I)$ defined in \eqref{eq:quantum_trace_g}. We define the {\it central elements} $\{c_\nu^I\}_{I}$ of $\cL_\nu$ as
\be
c_\nu^I=\kappa^I \tr^I_q(\bM_{\nu}^I)\,.
\label{eq:c_def}
\ee
They are invariant elements of $\cL_\nu$
and they satisfy the {\it fusion rule} 
\be\ba{l}
c_{\nu}^{I}c_{\nu}^{J}=c_{\nu}^K\,,\quad\\[0.15cm]
c_{\nu}^{I}c_{\mu}^{J}=c_{\mu}^{J}c_{\nu}^{I}\,,\quad \mu\neq\nu\,,
\ea
\label{eq:fusion_rules}
\ee
which can be proved using the properties \eqref{eq:g_properties} of the group-like element $g$, in particular $\Delta(g)=g\otimes g$ and $gS(\xi)=S^{-1}(\xi)g$, the recoupling \eqref{eq:holonomy_recoupling} of the quantum holonomies and the normalization \eqref{eq:CG_normalization} of the CG map. See also \cite{Alekseev:1994au}. 

The central element has $*$-operation given as $(c_{\nu}^{I})^{*}=c_{\nu}^{\Bar{I}}$, where the $\Bar{I}$ is the representation conjugate to $I$. 
The element $c_{\nu}^I$ commutes with holonomies around punctures other than $\nu$, \ie $c_{\nu}^{I}\bM_{\mu}^{J}=\bM_{\mu}^{J}c_{\nu}^{I}$ for $\mu\neq \nu$.

The above concepts can be well-defined for both cases $q^p\neq 1$ and $q^p=1$. To proceed, we will have to consider these two cases separately in order to define the moduli algebra properly. 
\medskip

\noindent {\bf The case $q^p=1$. }

Since the root-of-unity case has finitely many physical representations, one can define a symmetric and invertible ``$S$-matrix'' $S_{IJ}$ in the following way.
\be
S_{IJ}:=\cN (\tr^I_q\otimes\tr^J_q)(R'R)\,,\quad
\cN=\f{1}{\lb\sum_I d_I^2\rb^\f12}\,,
\label{eq:S-matrix_def}
\ee
where the summation runs through all the physical representations $I=0,\f12,\cdots,\f{p-2}{2}$. 
Indeed, if $I$ can take up to $\infty$, which is the case for $q^p\neq 1$, $\cN\rightarrow 0$ hence $S_{IJ}$ is ill-defined.  
The $S$-matrices satisfy the following properties \cite{Frohlich:1990ww,Alekseev:1994au}.
\be
S_{IJ}=S_{JI}\,,\quad S_{0J}=\mathcal{N}d_{J}\,,\quad
\sum_{J}S_{IJ}S_{JK}=\delta_{IK}\,,\quad
\sum^{u(I,J)}_{K=|I-J|}S_{KL}=\f{S_{JL}S_{IL}}{\cN d_{L}}\,,
\label{eq:S-property}
\ee
where $u(I,J)=\min(I+J,p-2-I-J)$.
The properties above show that the inverse element of the $S$-matrix is $S_{IJ}$ itself. 

With the help of the $S$-matrix, the {\it characters} in $\cA_{0,4}$, denoted as $\chi_\nu^J$ with $\nu=0,1,2,3,4$, can be defined as
\be
\chi_{\nu}^{J}=\cN d_{J}\sum_KS_{JK}c_{\nu}^{\Bar{K}}\,.
\label{eq:character_def}
\ee
A character is indeed a central element. It is also easy to prove that it is an orthogonal projector in $\cA_{0,4}$ satisfying
\be
\chi_{\nu}^{I}\chi_{\nu}^{J}=\delta_{IJ}\chi_{\nu}^{I}\,,\quad
(\chi_{\nu}^{J})^{*}=\chi_{\nu}^{J}\,.
\ee
Specially, for $\nu=0$,  $\bM_{0}^{J},\, c^J_0$ and $\chi^0$ are defined as
\be
\bM_0^J:=\kappa_{J}^{3}\bM_{4}^{J}\bM_{3}^{J}\bM_{2}^{J}\bM_{1}^{J}\,,\quad
c_{0}^{J}:=\kappa_J\tr_{q}^{J}(\bM_{0}^{J})\,,\quad
\chi_0^0= \cN^2\sum_J d_J\kappa_J^4 \tr_q^J\lb\bM_{4}^{J}\bM_{3}^{J}\bM_{2}^{J}\bM_{1}^{J}\rb\,.
\label{eq:bM_0}
\ee

We are now ready to give the definition of moduli algebra for the root-of-unity case.

\medskip

\noindent {\it Moduli algebra $\fM^{K_\nu}_{0,4}$ for $q^p=1$. } 
The moduli algebra $\fM_{0,4}^{K_\nu}$ of a four-punctured sphere, each of which is associated with an irreducible physical representation $K_{\nu}\,(\nu=1,2,3,4)$, is a $*$-algebra defined as \cite{Alekseev:1995rn}
\be
\fM_{0,4}^{K_{\nu}}:=\chi_{0}^{0}\chi_{1}^{K_{1}}\chi_{2}^{K_{2}}\chi_{3}^{K_{3}}\chi_{4}^{K_{4}}\cA_{0,4}\,.
\label{eq:def_moduli_algebra}
\ee
The expression above means that each element in $\fM_{0,4}^{K_{\nu}}$ is obtained by an element in the invariant algebra $\cA_{0,4}$ multiplied by the five characters
$\chi_{0}^{0}\,,\,\chi_{1}^{K_{1}}\,,\,\chi_{2}^{K_{2}}\,,\,\chi_{3}^{K_{3}}\,,\,\chi_{4}^{K_{4}}\in\cA_{0,4}$.

\medskip

\noindent {\bf The case $|q|=1, q^p\neq 1$. }
\medskip

For the case that $q$ is not a root-of-unity, the unboundedness of representations results in an ill-defined $S$-matrix \eqref{eq:S-matrix_def}, then the moduli algebra needs to be defined separately. In fact, the re-definition should start from the loop algebra and the graph algebra. 

We first consider the loop algebra in this case. Define $\bX=\kappa^{-1}\bX_+(\bX_-)^{-1}$ with $\bX_+\equiv R'$ and $\bX_-\equiv R^{-1}$. Then its representation reads $\bX^{I}=\kappa^{-1}_I(R'R)^{I}\equiv (\rho^I\otimes \Id)(R'R)$, where the pre-factor $\kappa^{-1}_I$ depends on the normalization condition of the CG maps. 
One can easily check that $\bX^I$ satisfies the same properties \eqref{eq:commutation_rep_1} and \eqref{eq:holonomy_recoupling} as the generators $\{\bM_\ell^I\}$ of the loop algebra. 
This shows that the loop algebra $\cL_{0,1}$ is isomorphic to $\UQ$ in this case\footnote{
Strictly speaking, the isomorphism is between $\cL_{0,1}$ and $\UQsl$ as proven in \cite{Alekseev:1993gh}. By such isomorphism, the $*$-structure of $\UQsl$, leading to $\UQ$, can induce a $*$-structure on $\cL_{0,1}$, which means that $\cL_{0,1}$ is $*$-isomorphic to $\UQ$. 
For the same reason, $\cL_{0,4}$ is $*$-isomorphic to $\UQ^{\otimes4}$.
\label{footnote:isomorphism}} 
\cite{Alekseev:1993gh,Alekseev:1995rn}. 

The isomorphism can be extended to the graph algebra as $\cL_{0,4}\cong \UQ^{\otimes 4}$. More precisely, we assign two elements $\bX_{\nu,+}\equiv R'_{1,\nu+1}$ and $\bX_{\nu,-}\equiv R_{1,\nu+1}^{-1}$ to the $\nu$-th puncture and define
\be
\bX_{\nu}^{I}:=\kappa_{I}^{-1} \bX^I_{1,+}\bX^I_{2,+}\cdots \bX^I_{\nu,+}(\bX^I_{\nu,-})^{-1}\cdots (\bX^I_{2,+})^{-1}(\bX^I_{1,+})^{-1}\in\End(V^I)\otimes \UQ^{\otimes \nu}\otimes e^{\otimes (4-\nu)}\,,
\label{eq:iso_M_graph_algebra}
\ee
where $e$ is the identity in $\UQ$.
They satisfy the commutation relations \eqref{eq:commutation_rep} and the recoupling relation \eqref{eq:holonomy_recoupling}. (See the detailed proof in Appendix \ref{app:isomorphism}.) 
We also define 
\be
\bX_0=\kappa^{-1}\bX_{0,+}(\bX_{0,-})^{-1}\,,\quad \text{ where } \bX_{0,\pm}:=\prod_{\nu=1}^{4}\bX_{\nu,\pm}\,,\quad 
\ee
and $\bX_0^I$ is its representation in $V^I$. 
$\bX_0$ can also be viewed as an embedding $\UQ\hookrightarrow \bigotimes_{\nu=1}^4\UQ$. 
Then we observe the same expression as in \eqref{eq:bM_0}, \ie
\be
\bX_0^I=\kappa_I^{-1}\bX_{0,+}^I\lb\bX_{0,-}^I\rb^{-1}
=\kappa^3_I \bX_{4}^I\bX_{3}^I\bX_{2}^I\bX_{1}^I\,.
\label{eq:X0}
\ee
From the expression above, the $\bX^I_0$ is the holonomy of the biggest cycle. This will be related to the quantum closure condition we define in Section \ref{sec:quantum_closure} below.

The element $c_\nu^I:=\kappa_I \tr_q^I(\bX_{\nu}^I)$ is defined in the same way as \eqref{eq:c_def} but with the new holonomy definition \eqref{eq:iso_M_graph_algebra}. 

The quantum trace of any element $Y^I\in\End(V^I)\otimes e\otimes\UQ$ ($e$ is the identity element in $\UQ$) satisfies the property 
\be
\tr_q^I(Y^I )=\tr_q^I((R^{-1})^{I} Y^I R^{I})=\tr_q^I(R'^{I}Y^I(R'^{-1})^{I})\,.
\label{eq:trace_lemma_0}
\ee
It is proven in Lemma \ref{lemma:trace_lemma}.
One can then show that $c_\nu^I$ is a central element of the copy of $\UQ$ corresponding to $\nu$-th puncture. That is, it commutes with $\bX_{\mu}^J\,,\,\forall\mu=1,\cdots, 4$ as we now prove. We consider $\mu=\nu$ and $\mu\neq \nu$ separately
for any representation $J\in\N/2$, 
\begin{subequations}
\be\begin{split}
\rho^J(c^I_\nu) \bX_{\nu}^J = \kappa_I\tr_q^I((R^{-1})^{IJ}\bXlnuI R^{IJ})\bX_{\nu}^J 
&\equiv \kappa_I\tr_q^I((R^{-1})^{IJ}\bXlnuI R^{IJ}\bXrnuJ)\\
&=\kappa_I\tr_q^I(\bXrnuJ (R')^{IJ}\bXlnuI (R'^{-1})^{IJ})
=\bX_{\nu}^J\rho^J(c^I_\nu)\,, 
\end{split}\ee
\be\begin{split}
\rho^J(c^I_\nu) \bX_{\mu}^J = \kappa_I\tr_q^I((R^{-1})^{IJ}\bXlnuI R^{IJ})\bX_{\mu}^J 
&\equiv \kappa_I\tr_q^I((R^{-1})^{IJ}\bXlnuI R^{IJ}\bXrmuJ)\\
&=\kappa_I\tr_q^I(\bXrmuJ (R^{-1})^{IJ}\bXlnuI R^{IJ})
=\bX_{\mu}^J\rho^J(c^I_\nu) \,,\quad
\ell_\nu\prec\ell_\mu\,,
\end{split}\ee
\be\begin{split}
\rho^J(c^I_\nu) \bX_{\mu}^J = \kappa_I\tr_q^I((R')^{IJ}\bXlnuI (R'^{-1})^{IJ})\bX_{\mu}^J 
&\equiv \kappa_I\tr_q^I((R')^{IJ}\bXlnuI (R'^{-1})^{IJ}\bXrmuJ)\\
&=\kappa_I\tr_q^I(\bXrmuJ (R')^{IJ}\bXlnuI (R'^{-1})^{IJ})
=\bX_{\mu}^J\rho^J(c^I_\nu) \,,\quad
\ell_\nu\succ\ell_\mu\,,
\end{split}\ee
\end{subequations}
where we have used \eqref{eq:trace_lemma_0} for the first and last equalities and \eqref{eq:commutation_rep} for the third equalities in all three equations. 

Due to the isomorphism with $\UQ$, one can define the moduli algebra for the case $|q|=1, q^p\neq 1$ in terms of the central elements, instead of the characters, of the loop algebra $\cL_\nu$. We define the representations of the central element $c^{I_\nu}_\nu$ as
\be
\rho^{K_\nu}(c^{I_\nu}_\nu)=\f{s_{I_\nu K_\nu}}{d_{K_\nu}}\Id_{K_\nu}\,, \quad \text{ with }
s_{IJ}:=(\tr_q^I\otimes \tr_q^J)(R'R)\,.
\label{eq:rep_c_qgen}
\ee

The moduli algebra can then be defined in the following way. 

\medskip

\noindent {\it Moduli algebra $\fM^{K_\nu}_{0,4}$ for $|q|=1,q^p\neq 1$. } 
Let $K_1,K_2,K_3,K_4$ be a set of 4 irreducible representations of $\UQ$ with $|q|=1,q^p\neq 1$ assigned to the four punctures of $\Sfour$. The moduli algebra $\fM^{K_\nu}_{0,4}$ is defined as the quotient of the invariant algebra $\cA_{0,4}$ by the relations $\rho^0(c^I_0)\,,\rho^{K_1}\lb c^{I_1}_1\rb\,,\rho^{K_2}\lb c^{I_2}_2\rb\,,\rho^{K_3}\lb c^{I_3}_3\rb\,,\rho^{K_4}\lb c^{I_4}_4\rb$, \ie \cite{Alekseev:1995rn,Alekseev:1993gh}
\be
\fM^{K_\nu}_{0,4}= \cA_{0,4}/\left\{\rho^0(c^I_0)\,,\rho^{K_1}\lb c^{I_1}_1\rb\,,\rho^{K_2}\lb c^{I_2}_2\rb\,,\rho^{K_3}\lb c^{I_3}_3\rb\,,\rho^{K_4}\lb c^{I_4}_4\rb\right\}_{I,I_{\nu}}\,,
\label{eq:def_M0}
\ee
where $I,I_1,I_2,I_3,I_4$ run through all irreducible representation of $\UQ$. 

From the construction above, it is reasonable to view the invariant algebra $\cA_{0,4}$ as the quantization of $\SU(2)$ flat connection space on $\Sfour$ and the moduli algebra $\fM^{K_\nu}_{0,4}$ as the quantization of the moduli space $\cM_\Flat(\Sfour,\SU(2))$ of $\SU(2)$ flat connections on $\Sfour$ with fixed eigenvalues of holonomies around all punctures. In the next subsection, we will construct the (irreducible) representations of the moduli algebra which allows us to build the Hilbert space of intertwiners. 

\subsection{Representation of the moduli algebra and the physical Hilbert space}
\label{sec:physical_Hilbert_space}

Having defined the moduli algebra $\fM^{K_\nu}_{0,4}$, we now proceed to construct its representation theory. The ``representation labels'' $K_1,\cdots,K_4$ defining the moduli algebra are representations of holonomies around punctures of $\Sfour$, which are fixed when defining the generators of $\cL_\nu,\cL_{0,4}$ and $\fM^{K_\nu}_{0,4}$, and do not {\it in priori} denote the representation of the moduli algebra itself. 
We will nevertheless find that only the irreducible representation space with the same labels can carry the representation of the moduli algebra. 

As the moduli algebra has a $*$-structure inherited from its generators, one can naturally construct the $*$-representations of it, which compose a Hilbert space. We denote such Hilbert space as the {\it physical Hilbert space} in the sense that states therein are invariant under the quantum gauge transformations. 
We will start by constructing the representation of the loop algebra $\cL_{0,1}$. Recall that the subalgebra $\cL_\nu$ of the graph algebra $\cL_{0,4}$ defined on each puncture $\nu$ is isomorphic to $\cL_{0,1}$. Then the representation of $\cL_{0,4}$ can be constructed by the tensor product of four representations of $\cL_{0,1}$ \cite{Alekseev:1995rn}. 
In this section, we will mainly focus on the case of $q^p=1$. For $|q|=1,q^p\neq1$ case, the results turn out to be the same as the $q^p=1$ case hence we skip the full derivations. We refer to \cite{Alekseev:1994au,Alekseev:1993gh} and Appendix \ref{app:reps_q_generic} for details. 

When $q$ is a root-of-unity, $\UQ$ is a finite-dimensional algebra whose total number of irreducible representations is finite. However, it has no semi-simplicity property, which brings obstacles in constructing the representation theory for the moduli algebra, which requires semi-simplicity. In order to restore the semi-simplicity, one considers its truncated algebra $\TUQ$ which has a simple representation theory despite a complex underlying algebraic structure. 
The representation theory of moduli algebra with quantum symmetry $\TUQ$ is constructed by applying the {\it substitution rule} (see \eqref{eq:substitution_rule} below) \cite{Alekseev:1994au,Alekseev:1995rn} on the representation theory of moduli algebra constructed below. 
In the following, we construct the representation theory of quantum algebra with quantum symmetry $\TUQ$, which can be easily generalized for semi-simple Hopf algebra with a finite number of irreducible representations \cite{Alekseev:1995rn}.

\medskip
\noindent {\it Representation of $\cL_{0,1}$. }
The loop algebra $\cL_{0,1}$ has a series of representations $\{D^{I}\}$ realized in the representation spaces $V^{I}$ of $\UQ$, $I\in \N/2$. The generator $\bM_\ell^J\in\End(V^J)\otimes \cL_{0,1}$ can be expressed in the representation $D^I$ as \cite{Alekseev:1995rn}

\be
D^{I}(\bM_\ell^{J}):=\kappa^{-1}_{J}(R'R)^{JI}\,.
\label{eq:rep_loop_algebra}
\ee

This representation preserves the commutation relation \eqref{eq:commute_xi_M} due to the intertwining property of $R'R$\footnotemark{}. 
\footnotetext{
Take the $J$ representation of \eqref{eq:commute_xi_M}, one gets $D^J(\mu^I(\xi)\bM_\ell^I)=(\rho^I\otimes \rho^J)\Delta(\xi)(R'R)^{IJ}=(R'R)^{IJ}(\rho^I\otimes \rho^J)\Delta(\xi)=D^J(\bM_\ell^I\mu^I(\xi))$ due to the definition \eqref{eq:def_mu_xi} of $\mu^I(\xi)$ and the property $\Delta(\cdot)(R'R)=(R'R)\Delta(\cdot)$ of $R'R$. See the last equality in \eqref{eq:def_v}. This proves that the representation \eqref{eq:rep_loop_algebra} of the loop algebra generator preserves the commutation relation \eqref{eq:commute_xi_M}.}
It has been proven in \cite{Alekseev:1995rn} that the set of representations $\{D^{I}\}$ defined above is faithful and, in the absence of truncation, each $D^{I}$ is an irreducible representation. Here faithfulness means that the dimension of loop algebra and the dimension of space of representation matrices under ${D^{I}}$ are the same.  

The representation of the central element $c^I$ (defined in \eqref{eq:c_def}) and the projector $\chi^I$ (defined in \eqref{eq:character_def}) of the loop algebra is given by\footnote{Notice that $c^J$ is a central element, $D^I(c^J)$ is proportional to $\Id_I$ by Schur's lemma. The proportionality is obtained by taking the trace of the representation. }

\be
 D^I(c^{J})=\frac{s_{JI}}{d_I} \Id_{I}\,,\quad
 D^{I}(\chi^J)=\delta_{IJ}\,,\quad
 \text{ where }\, 
 s_{JI}=(\tr_q^J\otimes \tr_q^I)(R'R)\,.
\label{eq:rep_central}
\ee

\medskip
\noindent {\it Representation of $\cL_{0,4}$. }
The representation of each sub-algebra $\cL_\nu$ can be realized in the space $V^{I_{\nu}}$ of $\UQ$, where the label $I_\nu$ is assigned at $\nu$-th puncture.  
Denote the representation of $\cL_{0,4}$ as $D^{{I_1},{I_2},{I_3},{I_4}}$. It is realized in the tensor product space $\cT(I_{1},I_{2},I_{3},I_{4})$ which is decomposable due to semi-simplicity:
\be
\cT (I_{1},I_{2},I_{3},I_{4}):=V^{I_1}\otimes V^{I_2} \otimes V^{I_3} \otimes V^{I_4} 
= \bigoplus_{I}V^{I}\otimes W^{I}(I_{1},I_{2},I_{3},I_{4}) \,.
\label{eq:T1234}
\ee
Here, $I$ runs through all admissible irreducible representations of $\UQ$ and $W^{I}(I_{1},I_{2},I_{3},I_{4})$ is the multiplicity space whose dimension is $\sum_{K_1K_2}N^{I_1I_2}_{K_1}N^{K_1I_3}_{K_2}N^{K_2I_4}_{I}$ with $N_I^{JK}=1$ if $I,J,K$ satisfy the triangular inequality $|I-J|\leq K\leq I+J$ and $N_K^{IJ}=0$ otherwise. 

 The representation of the generators of $\cL_{0,4}$ can be expressed as \cite{Alekseev:1995rn}
 \be
 D^{{I_1},{I_2},{I_3},{I_4}}(\bM_{\nu}^J)=(\rho^J\otimes\iota_{\nu})(R')D_{\nu}^{I_\nu}(\bM_{
 \ell_\nu}^J)(\rho^J\otimes\iota_{\nu})((R')^{-1})\in\End(V^J)\otimes\lb\bigotimes\limits_{\mu=1}^4\End(V^{I_\mu})\rb\,,\quad\nu=1,2,3,4\,.
 \label{eq:rep_graph_algebra}
\ee
The notations above mean:
\be
\begin{split}
D_{\nu}^{I_\nu}(\bM_{\nu}^J)&=\Id_{I_1}\otimes...\otimes \Id_{I_{\nu-1}}\otimes D^{I_\nu}(\bM_{\nu}^J)\otimes \Id_{I_{\nu+1}}\otimes...\otimes \Id_{I_4}\,,\\
\iota_{\nu}(\xi)&=(\rho^{I_1}\boxtimes...\boxtimes \rho^{I_{\nu-1}})(\xi)\otimes \Id_{I_{\nu}}\otimes...\Id_{I_4}\,,
\quad \forall\, \xi\in\UQ\,.
\end{split}
\ee
where $\rho^{I_1}\boxtimes...\boxtimes\rho^{I_m}(\xi)=(\rho^I_1\otimes...\otimes\rho^{I_m})\Delta^{(m-1)}(\xi)$. 
Explicitly,
\be
D^{I_1,I_2,I_3,I_4}(\bM_{\nu}^{J})=\kappa_{J}^{-1} \lb R'_{12}R'_{13}\cdots R'_{1\nu} R'_{1,\nu+1}R_{1,\nu+1}R^{'\,-1}_{1\nu}\cdots R^{'\,-1}_{13}R^{'\,-1}_{12}\rb^{J I_1\cdots I_\nu}\otimes e^{I_{\nu+1}}\otimes...\otimes e^{I_{4}}\,,
\label{eq:rep_M_explicit}
\ee
In order to define the representation for $(\bM_{\nu}^J)^*$, it is necessary to 
extend the representation $D^{{I_1},{I_2},{I_3},{I_4}}$ to quantum symmetry $\UQ$ by 

\be
 D^{{I_1},{I_2},{I_3},{I_4}}(\xi)=\iota_{4+1}(\xi)\equiv(\rho^{I_1}\boxtimes \rho^{I_2} \boxtimes \rho^{I_3} \boxtimes \rho^{I_4})(\xi)\,,\quad \forall\, \xi\in\UQ\,.
 \label{eq:rep_QU}
\ee
Same as the representations $\{D^I\}$ of the $\cL_{0,1}$, the representations $\{D^{{I_1},{I_2},{I_3},{I_4}}\}$ are also faithful and, in the absence of truncation, irreducible.
However, different from the fact that $D^{I}$ is a *-representation, the $*$-property of $D^{{I_1},{I_2},{I_3},{I_4}}$, if exists, is {\it not} compatible with the standard inner product of vector space $\cT(I_1,I_2,I_3,I_4)$. 
This is because the tensor product of two unitary representations is in general not unitary. 
Nevertheless, we can define an inner product that preserves the unitarity for the tensor product of two unitary representations. Such an inner product was first proposed in \cite{Durhuus:1991dk} (see also \cite{Alekseev:1995rn}) to be 
\be
\la x,y\ra_{R}:=\la x,(\rho\otimes\rho')(R\Delta(\kappa)(\kappa^{-1}\otimes\kappa^{-1}))y\ra\,,\quad
\forall\,x,y\in V\otimes V'
\label{eq:Scalar_reps}
\ee
where $\rho,\rho'$ are two $*$-representations of $\UQ$ on Hilbert spaces $V$ and $V'$, and $\la\cdot,\cdot\ra$ denotes the standard scalar product on $V\otimes V'$ \st $\la v_1\otimes v_1',v_2\otimes v_2'\ra=\la v_1,v_2\ra\la v_1',v_2'\ra,\,\forall v_1,v_2\in V,\,\forall v'_1, v'_2\in V'$. 

The scalar product $\la x,y\ra_R$ defined as such is positive definite, with respect to which $D^{{I_1},...,{I_4}}$ are $*$-representations.
To show that the inner product \eqref{eq:Scalar_reps} is preserved under the same unitary transformation, we consider $x,y\in V^I\otimes V^J$ and $\xi\in\UQ$. Then\footnotemark{}
\be\begin{split}
\la \rho^I\boxtimes\rho^J(\xi)x,y\ra_{R}&=
q^{I(I+1)+J(J+1)}\la \rho^I\boxtimes\rho^J(\xi)x,R^{IJ}(\rho^I\otimes\rho^J)\Delta(\kappa) y\ra \\
&=q^{I(I+1)+J(J+1)} \la x, \lb(\rho^I\otimes\rho^J)\Delta(\xi)\rb^* R^{IJ}(\rho^I\otimes\rho^J)\Delta(\kappa)y\ra\\
&=q^{I(I+1)+J(J+1)}\la x,(\rho^I\otimes\rho^J)\Delta'(\xi^*)R^{IJ} (\rho^I\otimes\rho^J)\Delta(\kappa)y\ra\\
&=q^{I(I+1)+J(J+1)}\la x, R^{IJ}(\rho^I\otimes\rho^J)\Delta(\kappa\xi^*)y\ra \\
&= \la x, \rho^I\boxtimes\rho^J(\xi^*) y\ra_R
\equiv \la x, \lb \rho^I\boxtimes\rho^J(\xi)\rb^* y\ra_R\,,
\end{split}\ee
where we have used the representation $\kappa^I=q^{-I(I+1)}$ of $\kappa$ (see \eqref{eq:v_kappa_J}), which is a scalar hence can be extracted out of the inner product, on the first line, the identity $\Delta'(\xi)R=R\Delta(\xi)\,,\forall\,\xi\in\UQ$ to obtain the second line, and the fact that $\kappa$ is a central element, \ie $\kappa\xi=\xi\kappa\,,\forall\,\xi\in\UQ$, to obtain the third line. 
\footnotetext{
The inner product $\la\cdot,\cdot \ra_R$ is preserved under unitary transformation even without the insertion of $\Delta(\kappa)(\kappa^{-1}\otimes\kappa^{-1})$ in \eqref{eq:Scalar_reps}. Such an insertion is there for a normalization purpose. In particular, given $x=(R')^{IJ}C[IJ|K]^*z\,,\, y=(R')^{IJ}C[IJ|K]^*z' \in V^I\otimes V^J=\oplus_K V^K$ such that $z,z'\in V^K$. The inner product defined as in \eqref{eq:Scalar_reps} realizes a simple decomposition $\la x,y \ra_R=\oplus_K \la z,z' \ra$ with no extra factors \cite{Durhuus:1991dk}.
}

\medskip
\noindent {\it Representation of $\cA_{0,4}$. }
 The $*$-algebra $\cA_{0,4}$ of gauge invariant elements is the subalgebra of $\cL_{0,4}$. Therefore, the representations $D^{I_{1},I_{2},I_{3},I_{4}} $ can be restricted to $\cA_{0,4}$ directly. 
 Since $\cA_{0,4}$ is an invariant algebra, its representation restricted to $V^I$ in the decomposition \eqref{eq:T1234} must be carried in the multiplicity space $W^I(I_1,I_2,I_3,I_4)$. 
 On the other hand, the dimension of $\cA_{0,4}$, counted from the number of its generators \eqref{eq:generator_A04}, is $\sum_{I,I_1,I_2,I_3,I_4}\dim \lb C[I_1I_2I_3I_4|I] \rb^2$, which matches exactly the dimension of the representation space carried by $\bigoplus_{I,I_1,I_2,I_3,I_4}W^I(I_1,I_2,I_3,I_4)$.  
 The faithfulness of the representations of $\cA_{0,4}$, inherited from the faithfulness of the representation of $\cL_{0,4}$, guarantees that the full multiplicity space $W^{I}(I_{1},I_{2},I_{3},I_{4})$ for each admissible set $(I,I_1,I_2,I_3,I_4)$ carries an irreducible $*$-representation of $\cA_{0,4}$.

 \medskip
\noindent {\it Representation of $\fM^{K_\nu}_{0,4}$. }
 Recall that the moduli algebra \eqref{eq:def_moduli_algebra} is defined as $\chi^{0}_{0}\prod_{\nu}^{4}\chi^{K_\nu}_{\nu}\cA_{0,4}$. 
 Given the representation of the invariant algebra above, the only extra ingredients to obtain the representation of the moduli algebra are the representation of the characters $\chi^0_0$ and $\chi^{K_\nu}_\nu\,,\,\nu=1,\cdots,4$. We first consider the latter one. Using the same definition \eqref{eq:rep_graph_algebra} of the representation, it is easy to compute 
 \be
D^{I_1,I_2,I_3,I_4}\lb \chi_\nu^{K_\nu} \rb =\delta_{I_\nu K_\nu} \Id_{I_1}\otimes \cdots \otimes \Id_{I_4}\,,\quad
\forall \, \nu=1,\cdots,4\,.
 \ee
Now that $\prod_{\nu=1}^{4}\delta_{I_\nu K_\nu}$ is imposed, we consider the following representation of $\chi^0_0$:
\be
D^{K_1,K_2,K_3,K_4}(\chi_0^0) 
=\cN^2\sum_{K} d_K \kappa_K^4 \tr_q^K \left[ D^{K_1,K_2,K_3,K_4}\lb \bM_{4}^K \bM_{3}^K \bM_{2}^K \bM_{1}^K \rb  \right]\,.
\ee
Using the definition \eqref{eq:rep_graph_algebra} (or directly the explicit expression \eqref{eq:rep_M_explicit}), one gets 
\be
D^{K_1,K_2,K_3,K_4}(\chi_0^0) 
= \cN^2\sum_K d_K \tr_q^K \left[\Delta^{(3)}\lb  R'R\rb\right]^{KK_1K_2K_3K_4} 
\equiv\cN^2 \sum_K d_K \tr_q^K\left[ (\rho^K\otimes \iota_5)(R'R) \right]\,,
\ee
where $\iota_5=\iota_{4+1}$ is defined in \eqref{eq:rep_QU}. 
Since $\fM_{0,4}$ is a subalgebra of $\cA_{0,4}$, its $J$ representation lives in the multiplicity space $W^J(K_1,K_2,K_3,K_4)$. For this reason, we are only interested in the representation restricted to $W^J(K_1,K_2,K_3,K_4)$, which we denote as $W^J$ for conciseness:
\be
D^{K_1,K_2,K_3,K_4}(\chi_0^0)|_{W^J} 
=\cN^2 \sum_K d_K \tr_q^K \lb\rho^K\otimes\rho^J \rb (R'R) \equiv \cN^2\sum_K d_K \f{S_{KJ}}{\cN \kappa_K d_K} \Id_{W^J} 
= \sum_K \f{S_{0K}S_{KJ}}{\kappa_J d_J}\Id_{W^J} =\delta_{0J} \Id_{W^0}\,,
\label{eq:chi0_rep}
\ee
where we have used the representation \eqref{eq:rep_central} of the central element and the property \eqref{eq:S-property} of the $S$-matrix.  

We finally conclude that, given four punctures labeled by representations $K_1,\cdots,K_4$ respectively, there exists only one irreducible $*$-representation space of the moduli algebra $\fM_{0,4}^{K_\nu}$, which is $W^0(K_1,K_2,K_3,K_4)$. 
Because of the faithfulness of the representations theory, other representations cannot exist.

\medskip

Up to now, we have dealt with the representation theory of $\cL_{0,1}$, $\cL_{0,4}$, $\cA_{0,4}$ and $\fM_{0,4}^{K_\nu}$ for $\UQ$ with $q$ a root-of-unity assuming the semi-simplicity of the quasi Hopf algebra. However, semi-simplicity is there only for the truncated algebra $\TUQ$ for $q^p=1$. The representation theory of $\TUQ$ is achieved from the one described above followed by applying the certain  substitution rule\footnotemark{} \cite{Alekseev:1994au,Alekseev:1995rn}:
\be
\begin{split}
&{C}^{a}[IJ|K] \rightarrow  \widetilde{C}^{a}[IJ|K] := C[IJ|K](\varphi^{-1})^{IJ}\,, \qquad
{C}^{a}[IJ|K]^{*}\rightarrow  \widetilde{C}^{a}[IJ|K]^{*} := (\varphi_{213}')^{IJ}C[IJ|K]^{*}\,, \\
&R^{IJ} \rightarrow \mathcal{R}^{IJ} := (\rho^{I}\otimes \rho^{J}\otimes \Id)(\varphi_{213}R\varphi^{-1})\,,\qquad
 d_I \rightarrow \tilde{d}_I := \tr^I(\rho^I(gS(\beta)\alpha))\,,\qquad
 R^I  \rightarrow  R^I \equiv(\rho^I\otimes \Id)R\,,\\
 &\tr_{q}^{I}(X)\rightarrow \widetilde{\tr}_q^I(X):=\tr^{I}(m^{I}Xw^{I}g^{I})\,,\quad
 \text{with }\; m^{I}=\rho^{I}(S(\phi^{(1)})\alpha \phi^{(2)})\phi^{(3)}\,,\quad  w^{I}=\rho^{I}(\varphi^{(2)}S^{-1}(\varphi^{(1)}\beta))\varphi^{(3)}\,,
\end{split}  
\label{eq:substitution_rule}
\ee
where $\varphi,\phi\equiv\varphi^{-1}\in\TUQ\otimes\TUQ\otimes\TUQ$ and $\alpha,\beta\in\TUQ$ are the defining elements for $\TUQ$ (see Section \ref{subsec:TUQ}). 
\footnotetext{
Due to some subtlety, the substitution rule \eqref{eq:substitution_rule} can not be applied to formulas containing the grouplike element $g$ (see \eqref{eq:g_properties}) except for computing the quantum trace $\tr_q$, which is the only places in this paper where we encounter $g$. We refer to \cite{Alekseev:1994au} for dealing with this subtlety.
}
Upon the imposition of the substitution rule, all the formulas of representation theory above in this section work for $\TUQ$. 
Due to the theorem in \cite{MACK1992185}, the physical representations of $\TUQ$ and $\UQ$ with $q^p=1$ are the same, hence the representation theory of $\TUQ$ constructed above can also represent the physical representation of $\UQ$ with $q^p=1$. 

\medskip

Let us now briefly go through the case for $\UQ$ with $|q|=1,q^p\neq1$ (see Appendix \ref{app:reps_q_generic} for some details). 
Since the loop algebra $\cL_{0,1}$ is isomorphic to $\UQ$ in this case, its irreducible representation is the same as the one of $\UQ$, in particular, we choose the $*$-irreducible representations that preserve the $*$-structure of $\UQ$ for $|q|=1,q^p\neq 1$ labeled by $J\in \N/2$. The representation of the generator $\bX^I\equiv \kappa_I^{-1}(R'R)^I$ of $\cL_{0,1}$ in the space $V^J$ can be directly calculated to be
\be
\rho^{J}(\bX^I)=\kappa_I^{-1} (R'R)^{IJ}\,, 
\label{eq:rep_X}
\ee
which matches the definition \eqref{eq:rep_loop_algebra} for the case $q^p=1$. It turns out naturally that the representation theories of $\cL_{0,1}\,,\,\cL_{0,4}\,,\,\cA_{0,4}$ and $\fM^{K_\nu}_{0,4}$ for case $q^p=1$ discussed above before applying the substitution rule are the same for the case $|q|=1, q^p\neq 1$. In particular, \eqref{eq:rep_M_explicit} can be directly used to express the representation of generators for $\cL_{0,4}$, and the representation spaces for $\cA_{0,4}$ and $\fM^{K_\nu}_{0,4}$ are $\bigoplus_{I,I_1,I_2,I_3,I_4}W^I(I_1,I_2,I_3,I_4)$ and $W^0(K_1,K_2,K_3,K_4)$ respectively\footnotemark{}.
\footnotetext{
Indeed, since the definitions of $\cA_{0,4}$ and $\fM^{K_\nu}_{0,4}$ vary for the cases $q^p=1$ and $|q|=1,q^p\neq 1$ as described in Section \ref{subsec:moduli_algebra}, their representations are also defined differently. 
Nevertheless, the resulting representation spaces for both cases are the same. We refer interested readers to \cite{Alekseev:1994au} for the explicit construction of the representations for case $|q|=1,q^p\neq 1$. See also Appendix \ref{app:reps_q_generic}.
}

\medskip

The analysis for $|q|=1, q^p\neq 1$ can be directly generalized to $|q|\neq 1$ and the same results follow except for the ones related to $*$-structure (hence we are dealing with $\UQsl$ instead of $\UQ$). We therefore conclude that the representation space of the $\UQsl$ moduli algebra on $\Sfour$ for with any $q\in \bC$ is $W^0(K_1,K_2,K_3,K_4)$, given that the four punctures carry the irreducible representations $K_1,K_2,K_3,K_4$ respectively. 
$W^0(K_1,K_2,K_3,K_4)$ is nothing but the space of intertwiner $C[K_1K_2K_3K_4|0]$. 
In the next section, we will see that it is indeed the solution space of the quantized version of the closure condition hence it can be naturally understood as the quantization of $\cM_\Flat(\Sfour,\SU(2))$ given fixed eigenvalues of the holonomies, each around one puncture of $\Sfour$.

\section{Quantum closure condition and the solution space -- the intertwiner space}
\label{sec:quantum_closure}

We have seen in Section \ref{sec:moduli_space} that, classically, the moduli space $\cM_\Flat(\Sfour,\SU(2))$ of $\SU(2)$ flat connection on $\Sfour$ is the solution space of the closure condition $M_{4}M_{3}M_{2}M_{1}=\Id$ where the $M_{\nu}\in\SU(2)$ is the holonomy around the $\nu$-th ($\nu=1,\cdots,4$) puncture of $\Sfour$. 
On the other hand, the curved Minkowski theorem (Theorem \ref{theorem:Minkowski}) states that a set of four holonomies $\{M_{\nu}\}$ satisfying the closure condition uniquely maps to a homogeneously curved tetrahedron whose faces $\{\ell_\nu\}$ carry areas $\{a_{\nu}\}$ corresponding to the eigenvalues of $\{M_{\nu}\}$. 
Here, each face $\ell_\nu$ is isomorphic to the loop around the $\nu$-th puncture of $\Sfour$. 
Inspired by LQG, the areas should be quantized to operators with discrete spectra, each characterized by an irreducible representation of the underlying algebra, which is $\UQ$ in our case. 

In this spirit, given fixed $\{a_{\nu}\}$, the moduli space is quantized to the moduli algebra $\fM_{0,4}^{K_\nu}$ where each representation label $K_\nu\in\N/2$ encodes the information of (quantum) area of face $\ell_\nu$. 
It is then natural to expect that the representation space of such moduli algebra can be viewed as the solution space of some quantum version of the closure condition. We illustrate in this section that this is the case. 
More precisely, we show that the quantum closure condition is naturally read out from the construction of the moduli algebra as described in Section \ref{sec:quantum_moduli}, which realizes the left downward arrow of \eqref{eq:road_map} hence closes the loop.
Due to the different definitions of moduli algebra for the $|q|=1,q^p\neq1$ and $q^p=1$ cases, we will need to define the quantum closure conditions for the two cases separately. This section and the next contribute the core of this paper. 

\begin{definition}[Quantum closure condition for $\UQ$]
\label{def:quantum_closure}
For case $q^p=1$, given the quantum holonomy $\bM_0^I$ represented in $V^I$ representation defined in \eqref{eq:bM_0}, the quantum closure condition is 
\be
\bM_0^I\equiv
\kappa_I^3 \bM_{4}^I \bM_{3}^I \bM_{2}^I \bM_{1}^I
=\kappa_I^{-1} e^I\,,
\label{eq:quantum_closure_1}
\ee
where $\kappa_I\equiv q^{-I(I+1)}$ is a central element and $e^I=\rho^I(e)$ is the $I$ representation of the identity element of $\UQ$.   

For case $|q|=1,q^p\neq 1$, given $\cL_{0,4}$ generators $\bX^I_0\equiv \kappa_I^{-1}\bX^I_{0,+}\lb\bX^I_{0,-}\rb^{-1}$ defined in \eqref{eq:X0}, the quantum closure condition is defined as
\be
\bX^I_{0}\equiv \kappa^3_I \bX_{4}^I\bX_{3}^I\bX_{2}^I\bX_{1}^I
=\kappa_I^{-1}e^I\,.
\label{eq:quantum_closure_2}
\ee
\end{definition}
It is easy to see that the above expressions of the quantum closure conditions take a similar form as the classical closure condition \eqref{eq:closure_SU2}. \eqref{eq:quantum_closure_1} and \eqref{eq:quantum_closure_2} are the quantized version of \eqref{eq:closure_SU2} (restricted to $V^I$) under the following natural quantization map
\be
\Id_{\SU(2)}\rightarrow \kappa_I^{-1}e^I\,,\quad
M_{\nu} \rightarrow \left\{\ba{l}
\bM_{\nu}^I \,,\quad \text{for }\, q^p=1\\[0.15cm]
\bX_{\nu}^I\,,\quad \text{for }\, |q|=1,q^p\neq 1
\ea\right.,\quad \forall \, \nu=1,\cdots,4\,.
\ee

\begin{theorem}[Solution to the quantum closure condition]
Given the representation labels $K_1,K_2,K_3,K_4$ assigned to the punctures of $\Sfour$, the representation space $W^0(K_1,K_2,K_3,K_4)$ of the moduli algebra $\fM_{0,4}^{K_\nu}$ for $\UQ$, or the intertwiner space, is and is the only solution space to the quantum closure condition \eqref{eq:quantum_closure_1} or \eqref{eq:quantum_closure_2} in the sense that 
\begin{align}
\bM_0^I\lb  W^0(K_1,K_2,K_3,K_4) \rb &= \kappa^{-1}_I e^I \lb  W^0(K_1,K_2,K_3,K_4)\rb \,,\quad \text{for }\, q^p=1\,,\\
\bX_0^I\lb  W^0(K_1,K_2,K_3,K_4) \rb &= \kappa^{-1}_I e^I \lb  W^0(K_1,K_2,K_3,K_4)\rb \,,\quad \text{for }\, |q|=1,q^p\neq 1\,.
\end{align}
\end{theorem}
\begin{proof}

We first consider the case of $q^p=1$. 
The proof relies on the ``flatness condition''\footnote{We remind the readers that the term ``flatness'' here means the flatness of the Chern-Simons connection \cite{Alekseev:1994au} and does not refer to zero curvature as the flatness in gravity context. }
\be
\chi_0^0\bM_0^I=\chi_0^0\kappa_I^{-1}e^I
\label{eq:flatness}
\ee
that is proven in \cite{Alekseev:1994au} for both truncated and before-truncation cases, which we  provide in Appendix \ref{app:details} (Lemmas \ref{lemma:flatness_gen} and \ref{lemma:flatness_trun}). 

The goal is to look for the subspace, say $\cH$, of $\cT(K_1,K_2,K_3,K_4)\equiv \bigoplus_J V^J\otimes W^J(K_1,K_2,K_3,K_4)$ \st 
\be
\bM_0^I\,\cH = \kappa_I^{-1}e^I \cH\,.
\label{eq:eigen_function_H}
\ee
By \eqref{eq:flatness}, it implies
\be
\chi_0^0\bM_0^I \,\cH = \bM_0^I \chi_0^0 \,\cH =\kappa_I^{-1}e^I\chi_0^0\,\cH\,,
\ee
where the first equality comes from the fact that $\chi_0^0$ is a central element. We look for such $\cH$ by testing the action of $\bM_0^I\chi_0^0$ on the multiplication space $W^J(K_1,K_2,K_3,K_4)\equiv W^J$ of each $V^J$ separately, 
which is simply the representation of $\bM_0^I\chi_0^0$ restricted on the space $W^J$. 
Recalling \eqref{eq:chi0_rep} and \eqref{eq:flatness}, we immediately have
\be
\bM_0^I\chi_0^0\, W^J \equiv D^{K_1,K_2,K_3,K_4}(\bM_0^I\chi_0^0)|_{W^J} 
=\delta_{0J}  \kappa_I^{-1}  e^I\Id_{W^0}
\equiv \delta_{0J} \kappa_I^{-1}e^I \chi_0^0\, W^0\,.
\label{eq:Mchi_act_on_Wj}
\ee

This implies that $\cH\subseteq W^0$. 
Since $W^0$ is defined from a projection of $\chi_0^0$ \ie $\chi_0^0\,W^0=W^0$,  
\eqref{eq:Mchi_act_on_Wj} also implies $\bM_0^I\,W^0 =\kappa_I^{-1}e^I \,W^0$, which means the full space $W^0$ is the solution of $\cH$ to \eqref{eq:eigen_function_H}.  

\medskip

We next consider the case of $|q|=1,q^p\neq 1$. In this case, we can no longer define $\chi_0^0$ but the proof becomes simpler due to the isomorphism of $\cL_{0,4}$ and $\UQ^{\otimes 4}$. The representation of $\bX_0^I$ in $\cT(K_1,K_2,K_3,K_4)$ can be directly calculated:
\be
\begin{aligned}
\rho^{K_1,K_2,K_3,K_4}(\bX^I_0)&=\kappa^{-1}_I(\rho^I\otimes (\rho^{K_1}\otimes\rho^{K_2}\otimes\rho^{K_3}\otimes\rho^{K_4}))(R'_{12}R'_{13}R'_{14}R'_{15}R_{15}R_{14}R_{13}R_{12})\\
&=\kappa^{-1}_I(\rho^I\otimes (\rho^{K_1}\otimes\rho^{K_2}\otimes\rho^{K_3}\otimes\rho^{K_4}))(\Delta^{(3)}(R'R))\\
&=\kappa^{-1}_I(\rho^I\otimes \iota_5)(R'R)\,,
\end{aligned}
\ee
whose restriction on each invariant subspace $W^J$ of $\cT(K_1,K_2,K_3,K_4)$ gives 
\be
\rho^{K_1,K_2,K_3,K_4}(\bX^I_0)|_{W^J}=\kappa^{-1}_I(\rho^I\otimes \rho^J)(R'R)\,.
\ee
Only when $J=0$ the action of $X^I_0$ gives us the closure condition. That is,
\be
\bX_0^I W^0 := \rho^{K_1,K_2,K_3,K_4}(\bX^I_0)|_{W^0}=\kappa^{-1}_I(\rho^I\otimes \epsilon)(R'R)=\kappa^{-1}_I e^I|_{W^0}=\kappa^{-1}_I e^I\,W^0\,.
\ee
This completes the proof that $W^0$ is and is the only solution to the quantum closure condition for both cases of $q$. 
\end{proof}

\section{Area operator}
\label{sec:area}

Having found the eigenspace $W^0$ of the quantum closure condition, it remains to show that the representation labels $K_1,K_2,K_3,K_4$ do represent the spectra of quantum areas for completing the loop of \eqref{eq:road_map}. 
Classically, as described in Section \ref{sec:classical_closure}, the area $a_{\nu}$ of a face $\ell_\nu$ of a curved tetrahedron is encoded in the trace of holonomy $M_{\nu}$ around the $\nu$-th puncture written in the fundamental representation \eqref{eq:H_definition}:
\be
\f12\tr(M_{\nu}) =\cos \left(\frac{|\Lambda|}{6}a_{\nu}\right) \,.
\ee
Noting that $M_{\nu}$ is quantized to the quantum holonomy $\bM_{\nu}$, we propose a quantization of $\cos (\f{|\Lambda|}{6}a_{\nu}) $ with a  natural spectrum labeled by $K_\nu\in\N/2$, denoted as ${\rm spec}^{K_\nu}$ (Recall the geometrical interpretation of the holonomies in \eqref{eq:H_definition}): 
\be\begin{split}
\cos \left(\f{|\Lambda|}{6}a_{\nu}\right) 
\quad\longrightarrow\quad 
\cos \lb \frac{|\Lambda|}{6}{\rm spec}^{K_\nu}(\hat{a}_{\nu})\rb 
:=&\f{\kappa_{1/2}}{2}D^{K_\nu}\lb \tr_q^{1/2}\lb \bM_{\nu}^{1/2} \rb \rb\equiv 
\left\{\ba{ll}
\f12 D^{K_\nu}\lb c_{\nu}^{1/2} \rb \,,\quad & \text{ for } q^p=1\\[0.15cm]
\f12 \rho^{K_\nu}\lb c_{\nu}^{1/2} \rb \,,\quad & \text{ for }|q|=1\,, q^p\neq1\\[0.15cm]
\ea\right..
\end{split}
\label{eq:def_area_op}
\ee
We require that the area spectrum only takes value in a finite interval
\be
\frac{|\Lambda|}{6}{\rm spec}^{K_\nu}(\hat{a}_{\nu})\in [0,\pi).
\label{eq:interval}
\ee
It can also be understood by defining the area operator with the inverse function: 
\be
\frac{|\Lambda|}{6}{\rm spec}^{K_\nu}(\hat{a}_{\nu}):=\arccos[\f{\kappa_{1/2}}{2}D^{K_\nu}( \tr_q^{1/2}( \bM_{\nu}^{1/2}))]\,,
\label{eq:spectrum_def}
\ee
where $\arccos(x)\in[0,\pi)$ is single-valued.
The area spectrum takes a simple form, which is given in the following theorem. 

\begin{theorem}[Spectrum of the area operator] \label{area spectrum}
Given $q=e^{i\theta}$ being a phase with $\theta\in(0,2\pi)$ and $K_\nu\in \N/2$ characterizing the area spectrum ${\rm spec}^{K_\nu}(\hat{a}_{\nu})$ defined in \eqref{eq:def_area_op}, ${\rm spec}^{K_\nu}(\hat{a}_{\nu})$ is given by
\be
{\rm spec}^{K_\nu}(\hat{a}_{\nu}) =
\begin{cases}
\frac{6}{|\Lambda|}(2K_\nu+1)\theta,& 0\leq K_\nu<\frac{1}{2}(\frac{\pi}{\theta}-1)\\
\frac{6}{|\Lambda|}\left [2\pi -(2K_\nu+1)\theta\right],& \frac{1}{2}(\frac{\pi}{\theta}-1)\leq  K_\nu < \frac{1}{2}(\frac{2\pi}{\theta}-1)
\end{cases}\,,\quad \nu=1,\cdots,4\,.
\label{eq:area_spectrum}
\ee
If $q=e^{\f{2\pi i}{k+2}}\ (k\in \N)$ is a root-of-unity and $K_\nu\leq \f{k}{2}$, 
\be
 {\rm spec}^{K_\nu}(\hat{a}_{\nu}) =
\begin{cases}
\frac{6}{|\Lambda|}(2K_\nu+1)\frac{2\pi}{k+2},& 0\leq K_\nu<\frac{k}{4}\\
\frac{6}{|\Lambda|}\left [2\pi -(2K_\nu+1)\frac{2\pi}{k+2}\right],& \frac{k}{4}\leq  K_\nu\leq \frac{k}{2}
\end{cases}\,,\quad \nu=1,\cdots,4\,.
\label{eq:area_spectrum1111}
\ee   
\end{theorem}
\begin{proof}
 The area spectrum comes from the direct calculation of $D^{K_\nu}\lb c_{\nu}^{1/2}\rb$ for the case $q^p=1$ and $\rho^{K_\nu}\lb c_{\nu}^{1/2} \rb$ for the other case, which has been revealed in \eqref{eq:rep_c_qgen} and \eqref{eq:rep_central} followed by the substitution rule \eqref{eq:substitution_rule} for the former case. Therefore,
\be
D^{K_\nu}\lb c_{\nu}^{\f12} \rb = \f{1}{d_{K_\nu}}\lb \tr_q^{1/2}\otimes \tr_q^{K_\nu} \rb(\cR'\cR)\,,\quad 
\rho^{K_\nu}\lb c_{\nu}^{\f12} \rb =  \f{1}{d_{K_\nu}}\lb \tr_q^{1/2}\otimes \tr_q^{K_\nu} \rb(R'R)\,.
\label{eq:rep_c_both}
\ee
To proceed, we use the results in Lemmas \ref{lemma:RpR_qgen} and \ref{lemma:RpR_qroot} below, which give the same result for $K_\nu$ representation of $c^I_{\ell}$ for any $q=e^{i\theta}$. 
For the case $|q|=1,q^p\neq 1$, we calculate that
\be
\rho^{1/2}\otimes \rho^{K_\nu} (R'R)=\rho^{K_\nu+\f12}(R'R) + \rho^{K_\nu-\f12}(R'R)=
\f{v^{\f12}v^{K_\nu}}{v^{K_\nu+\f12}}\Id_{K_\nu+\f12}+ \f{v^{\f12}v^{K_\nu}}{v^{K_\nu-\f12}}\Id_{K_\nu-\f12}\,.
\label{eq:rho_RpR}
\ee
For the case $q^{k+2}=1$, on the other hand, the same expression holds only when $K_\nu+\f12\leq \f{k}{2}$, leading to $u(K_\nu,\f12)=K_\nu+\f12$. In this case,
\be
\rho^{\f12}\otimes \rho^{K_\nu} (\cR'\cR)
= \rho^{K_\nu+\f12}(\cR'\cR) + \rho^{K_\nu-\f12}(\cR'\cR) 
= \f{v^{\f12}v^{K_\nu}}{v^{K_\nu+\f12}}\Id_{K_\nu+\f12}+ \f{v^{\f12}v^{K_\nu}}{v^{K_\nu-\f12}}\Id_{K_\nu-\f12}\,.
\label{eq:rho_RpR_2}
\ee 
The irreducible representation of the ribbon element $v$ takes the form $v^I=q^{-2I(I+1)}$ (see Lemma \ref{lemma:vJ}). 
The quantum trace of the identity matrix simply gives the quantum dimension of the representation space \ie $\tr_q^I\lb \Id_I \rb =d_I\equiv [2I+1]_q$. Combining these facts, \eqref{eq:rep_c_both} reads 
\be\begin{split}
D^{K_\nu}\lb c_{\nu}^{\f12} \rb = \rho^{K_\nu}\lb c_{\nu}^{\f12} \rb 
&=\f{1}{d_{K_\nu}}\lb \f{q^{-\f32}q^{-2K_\nu(K_\nu+1)}}{q^{-2(K_\nu+\f12)(K_\nu+\f32)}} d_{K_\nu+\f12}+
\f{q^{-\f32}q^{-2K_\nu(K_\nu+1)}}{q^{-2(K_\nu-\f12)(K_\nu+\f12)}} d_{K_\nu-\f12}\rb\\
&=\f{1}{[2K_\nu+1]_q}\lb q^{2K_\nu} [2K_\nu+2]_q + q^{-2K_\nu-2}[2K_\nu]_q\rb\\
&=\f{[2(2K_\nu+1)]_q}{[2K_\nu+1]_q}
=\f{e^{i2(2K_\nu+1)\theta}- e^{-i 2(2K_\nu+1)\theta}}{{e^{i(2K_\nu+1)\theta}- e^{-i (2K_\nu+1)\theta}}}=\f{\sin\lb 2(2K_\nu+1)\theta\rb}{\sin\lb (2K_\nu+1)\theta\rb}\\
&= 2\cos\lb(2K_\nu+1)\theta \rb \,,
\label{eq:D_c-result}
\end{split}
\ee
where we have used the definition of a quantum number $[n]_q:=\f{q^n-q^{-n}}{q-q^{-1}}$ and that $q^{i\theta}$ is a pure phase. 

As a special case, when $q^{k+2}=1$ and $K_\nu=k/2$, $u(K_\nu,\f12)=K_\nu-1/2$ and hence only one representation $\rho^{K_\nu-\f12}$ is left in the recoupling theory. This means, instead of \eqref{eq:rho_RpR_2}, we have
\be
\rho^{\f12}\otimes \rho^{K_\nu} (\cR'\cR)=\rho^{K_\nu-\f12}(\cR'\cR)=\f{v^{\f12}v^{K_\nu}}{v^{K_\nu-\f12}}\Id_{K_\nu-\f12}
=q^{-(k+2)}[k]_q \equiv [k]_q
\,\,\Rightarrow\,\,
D^{K_\nu}\lb c^{\f12}_{\ell_\nu} \rb = \f{[k]_q}{[k+1]_q}
\equiv 2\cos\lb \f{k+1}{k+2}2\pi \rb
\,,
\ee
which is the same as the result of \eqref{eq:D_c-result} when $K_\nu=k/2$ and $\theta=\f{2\pi}{k+2}$.
Adding the factor $\f{6}{|\Lambda|}$, by definition \eqref{eq:spectrum_def}, \eqref{eq:area_spectrum} and \eqref{eq:area_spectrum1111} are proven\footnote{More generally, $D^{K_\nu}\lb c_{\ell}^{J} \rb=\rho^{K_\nu}\lb c_{\ell}^{J}\rb=\f{[(2J+1)(2K_\nu+1)]_q}{[2K_\nu+1]_q}$ for any $q\in \bC$, which is proved in Lemma \ref{lemma:tr_RpR_gen}. }. 
\end{proof}

As described in Section \ref{sec:classical_closure}, the holonomies can be written in an arbitrary irreducible representation. This allows us to read the area spectrum from any irreducible representation. Recall that, classically, the area $a_{\nu}$ of a face $\ell_\nu$ of a curved tetrahedron is encoded in the trace of the holonomy $M_{\nu}$ around the $\nu$-th puncture written in an arbitrary representation \eqref{eq:area_general}:
\be
\tr^I(M^I_{\ell_\nu})=\tr^I(g^I (e^{i\frac{|\Lambda|}{6}a_{\nu}H})^I(g^{-1})^I)=\frac{\sin((2I+1)\frac{|\Lambda|}{6}a_{\nu})}{\sin\lb\frac{|\Lambda|}{6} a_{\nu}\rb}\;.
\ee
As a natural generalization of \eqref{eq:def_area_op}, we define the quantization of the {\it r.h.s} as 
\be
\begin{split}
\frac{\sin((2I+1)\frac{|\Lambda|}{6}{\rm spec}^{K_\nu}(\hat{a}_{\nu}))}{\sin \lb\frac{|\Lambda|}{6}{\rm spec}^{K_\nu}(\hat{a}_{\nu})\rb}:= 
\kappa_I D^{K_\nu}\lb \tr^I_q(\bM^I_{\nu})\rb\equiv 
\left\{\ba{ll}
 D^{K_\nu}\lb c_{\nu}^{I} \rb \,,\quad & \text{ for } q^p=1\\[0.15cm]
 \rho^{K_\nu}\lb c_{\nu}^{I} \rb \,,\quad & \text{ for }|q|=1\,, q^p\neq1\\[0.15cm]
\ea\right..
\label{eq:define_gen_area_op}
\end{split}
\ee

With the results from the lemmas  \ref{lemma:RpR_qgen}, \ref{lemma:RpR_qroot}, \ref{lemma:tr_RpR_gen}, and \ref{lemma:tr_RpR_trun}, when $q=e^{i\theta}$ is a phase, we have
\be
D^{K_\nu}\lb c_{\ell}^{I} \rb=\rho^{K_\nu}\lb c_{\ell}^{I}\rb=\f{[(2J+1)(2K_\nu+1)]_q}{[2K_\nu+1]_q}=\frac{\sin\lb(2I+1)(2K_\nu+1)\theta\rb}{\sin\lb(2K_\nu+1)\theta\rb}\;.
\label{eq:D_c_sin}
\ee
Therefore, the spectrum in \eqref{eq:area_spectrum} also satisfies  \eqref{eq:define_gen_area_op}. This shows that Eq.\eqref{eq:area_spectrum} defines the quantization of area independent of the choice of representation $I$ for the closure condition.

Below we prove the lemmas that are used in the proof of Theorem \ref{area spectrum}.

\begin{lemma}
\label{lemma:RpR_qgen}
Let $\rho^{I}$ and $\rho^{J}$ be two irreducible representations of $\UQ$ with $q^p\neq1$ with highest weights $I$ and $J$ respectively. Let $\rho^K$ be the irreducible component with the highest weight $K$ in the decomposition of the tensor product of representations $\rho^{I}$ and $\rho^{J}$. Then
\be
(\rho^I\otimes\rho^J)(R'R)=\sum^{I+J}_{K=I-J}\frac{v_Iv_J}{v_K}\Id_K\;.
\ee
\end{lemma}
\begin{proof}
Since $\UQ$ with $q^p\neq 1$ possesses the semisimplicity property, for every admissible tuple $(I, J, K)$, the tensor product of representations is decomposable, \ie $(\rho^I\boxtimes\rho^J)(\xi) = \sum^{I+J}_{K=|I-J|} \rho^K(\xi)$. Now, we will show that the action of $\rho^{I}\otimes\rho^{J}(R'R)$ on each subspace $V^{K}$ is given by $\frac{v_{I}v_{J}}{v_{K}}\Id_{K}$, where $\Id_K$ represents the identity map of $V^K$. It is important to note that $R'R$ commutes with $\Delta(\xi),\,\forall\,\xi\in\UQ$. By Schur's lemma, the restriction of $R'R$ onto the subspace $V^K$ is equal to the identity map $\Id_K$ up to a factor. 

We now show that this factor is $\frac{v_{I}v_{J}}{v_{K}}$. Recall that the defining relation of the ribbon element $v$: $\Delta(v)=(R'R)^{-1}(v\otimes v)$ or equivalently $(R'R)\Delta(v)=(v\otimes v)$. Taking the presentation $\rho^I\otimes\rho^J$ of the latter relation, we get 
\be
\rho^{I}\otimes\rho^{J}(R'R\Delta(v))\equiv
\rho^{I}\otimes\rho^{J}(R'R)\rho^{I}\otimes\rho^{J}(\Delta(v))=\rho^{I}\otimes\rho^{J}(v\otimes v)
\equiv v_{I}v_{J}(\Id_{I}\otimes \Id_{J})\,.
\ee
This relation holds for each subspace $V^K$, which gives 
\be
v_K(R'R)^{IJ}_{K}=v_{I}v_{J}\Id_K
\quad\Longrightarrow\quad
(R'R)_{K}^{IJ}=\frac{v_{I}v_{J}}{v_{K}}\Id_{K}\,,
\ee
where $(R'R)^{IJ}_K$ denotes the $\rho^I\otimes\rho^J$ representation of $R'R$ restricted to the subspace $V^K$ and we have used the definition $\rho^I(v)=v_I\Id_I$ of the representation of the ribbon element and that $\rho^I\otimes\rho^J(\Delta(v))|_{V^K}=v_K\Id_K$. 
We have, therefore, completed the proof for the lemma.
\end{proof}
\begin{lemma}
\label{lemma:RpR_qroot}
  For the truncated case $\mathcal{U}^T_q(\mathfrak{su}(2))$ with $q=e^{\f{2\pi i}{k+2}}$, consider two physical representations, $\rho^{I}$ and $\rho^{J}$, of $\mathcal{U}^T_q(\mathfrak{su}(2))$, where $I$ and $J$ can take values in half-integers from $0$ to $\frac{k}{2}$. Let $\rho^K$ represent the irreducible component in the decomposition of the tensor product of these representations, $\rho^{I}$ and $\rho^{J}$. Applying the substitution rule \eqref{eq:substitution_rule}, we obtain the following identity.
  \be
  (\rho^I\otimes \rho^J)(\mathcal{R'}\mathcal{R})=(\rho^I\otimes \rho^J)(\varphi^{IJ}(R'R)^{IJ}(\varphi^{-1})^{IJ})=\sum^{u(I,J)}_{K=|I-J|} \frac{v_Iv_J}{v_K} \Id_K\;.
  \ee
 \end{lemma}
 \begin{proof}
 We follow the same idea as we used to prove the Lemma \ref{lemma:RpR_qgen}. For the truncated algebra $\mathcal{U}^T_q(\mathfrak{su}(2))$, the semisimplicity property holds, \ie 
 \be(\rho^I\boxtimes\rho^J)(\xi)=\sum_{K=|I-J|}^{u(I,J)} \rho^K(\xi)\;, \forall \xi\in \mathcal{U}^T_q(\mathfrak{su}(2)),
 \ee
 where $u(I,J)=\text{Min}(I+J,k-I-J)$.
 Moreover, the defining relations for the ribbon element and quasi-triangularity still hold \cite{ Altschuler:1992zz}:
\be
\Delta(v) = (R'R)^{-1}(v\otimes v)\quad,\quad
(R'R)\Delta(\xi) = \Delta(\xi)(R'R),
\ee
where $R$ and $R'$ satisfy the quasi-Yang-Baxter equation. For $\cU^T_q(\mathfrak{su}(2))$, the action of $\rho^{I}\otimes\rho^{J}(R'R)$ on each subspace $V^{K}$ is given by $\frac{v_{I}v_{J}}{v_{K}}\Id_K$ as

\be
(\rho^{I}\otimes \rho^{J})(\varphi((R'R)(\varphi^{-1}))|_{V^K} = \varphi_{K}^{IJ}(R'R)^{IJ}_K (\varphi^{-1})_{K}^{IJ}
= \varphi_{K}^{IJ}\frac{v_{I}v_{J}}{v_{K}}\Id_{K}(\varphi^{-1})_{K}^{IJ}
= \frac{v_{I}v_{J}}{v_{K}}\Id_{K}\,.
\ee

By combining the semisimplicity property of $\mathcal{U}^T_q(\mathfrak{su}(2))$, the proof is complete.
\end{proof}

\medskip

We compare the area spectrum in \eqref{eq:area_spectrum} and \eqref{eq:area_spectrum1111} to the area spectrum in LQG. The standard area spectrum in 4D LQG without cosmological constant takes the form 
${\rm spec}^{K_\nu}(\hat{a})=\gamma \ell_{\text{p}} ^2 \sqrt{K_\nu(K_\nu+1)}$, where $K_\nu$ is the SU(2) spin. Identifying $K_\nu$ to the $\mathcal{U}_q(\mathfrak{su}(2))$ spin and requiring \eqref{eq:area_spectrum} and \eqref{eq:area_spectrum1111} consistent with the standard LQG area in the limit $\theta,\,|\Lambda|\to 0$ and large $K_\nu$ (\ie the area grows linearly in $K_\nu$\footnote{The second case in \eqref{eq:area_spectrum} disappears in the limit $\theta\to0$.}) results in that $\theta=\frac{1}{12}\ell_{\rm p}^2\gamma|\Lambda|$, or $k+2=\f{24\pi }{\ell_{\rm p}^2 \gamma|\Lambda|}\in\mathbb{N}$ in the case of $q=\exp(\frac{2\pi i}{k+2})$. $k+2$ here equals to two times the integer level in 3+1 dimensional spinfoam model with cosmological constant \cite{Haggard:2014xoa,Han:2021tzw,Han:2023hbe}. The area spectrum \eqref{eq:area_spectrum} and \eqref{eq:area_spectrum1111} reduce to 
\begin{equation}
{\rm spec}^{K_\nu}(\hat{a}_{\nu})=\begin{cases}\gamma \ell_{\rm p}^2 \left(K_\nu+\frac{1}{2}\right),& 0\leq K_\nu<\frac{1}{2}B\\
\frac{12\pi }{|\Lambda|}-\gamma \ell_{\rm p}^2 \left(K_\nu+\frac{1}{2}\right), &\frac{1}{2}B\leq K_\nu < B
\end{cases},\qquad B=\frac{12\pi}{\ell_{\rm p}^2\gamma|\Lambda|}-1.
\label{areasp1111}
\end{equation}
As a key difference from LQG with vanishing $\Lambda$, the area spectrum is bounded from above,
\be
{\rm spec}^{K_\nu}(\hat{a}_{\nu})\leq \frac{6\pi}{|\Lambda|}\,.
\ee 
Although the resulting area spectrum is consistent with the standard LQG result for $\Lambda\to 0$ and large $K_\nu$, a noticeable difference shows up in the regime of small $K_\nu$. In particular, when $K_\nu=0$, the area does not become trivial, in contrast to the area spectrum in the standard LQG. Indeed, the area spectrum here is restrictively positive, and in the case of $q=\exp(\frac{2\pi i}{k+2})$,
\be
\mathrm{Min}\left[{\rm spec}^{K_\nu}(\hat{a}_{\nu})\right]=\frac{1}{2}\gamma\ell_{\rm p}^2.
\ee
The minimum can be smaller for generic $q=e^{i\theta}$:
\be
\mathrm{Min}\left[{\rm spec}^{K_\nu}(\hat{a}_{\nu})\right]=\frac{1}{2}\gamma\ell_{\rm p}^2 \left(2B-\lfloor2B\rfloor\right)
\ee
where $B$ is defined in \eqref{areasp1111} and $\lfloor\cdot\rfloor$ is the floor function.

Interestingly, our result \eqref{areasp1111} is the same as the area spectrum obtained in \eg \cite{Alekseev:2000hf}, which suggests the value of $\gamma$ different from the standard LQG by the black hole microstate counting.

\section{Conclusion and discussion}
\label{sec:conclusion}

In this work, we show that the solution space of the quantum curved closure condition coincides with the intertwiner space, $W^0(K_1,K_2,K_3,K_4)$ of $\mathcal{U}_q(\mathfrak{su}(2))$ with the quantum deformation parameter $q$ being a phase for both cases $|q|=1,q^p\neq 1$ and $q^p=1$, where $K_\nu$ is the $\mathcal{U}_q(\mathfrak{su}(2))$ spin.
Inspired by the LQG, the geometrical quantities are quantized as the operators such that their actions on the Hilbert space give discrete spectra. Classically, the area $a_{\nu}$ of the face $\ell_\nu$ of a curved tetrahedron can be calculated from the holonomy $M_\nu$ surrounding $\ell_\nu$ in the fundamental representation by $\frac{1}{2}\tr(M_\nu)=\cos{(\frac{|\Lambda|}{6}
a_{\nu})}$. Moreover, we generalize the holonomy $M^I_\ell$ in the classical theory to arbitrary representation $I$ and the area is also related to the trace of the holonomy. In the quantum theory, the area spectrum is discrete and bounded from above and below.

Recently, an improved spinfoam model was proposed  \cite{Han:2021tzw}. This model formulates the 3+1 dimensional Lorentzian spinfoam amplitude with a non-zero cosmological constant in terms of the Chern-Simons theory with the complex group $\SL(2,\mathbb{C})$. It has been demonstrated to possess finite spinfoam amplitudes and exhibit the correct semiclassical behavior. The amplitude involves sum over quantum areas $j_f$, and a natural cutoff is provided by the Chern-Simons level $k$. The space of boundary data can be identified as the phase space of shapes of a homogeneously curved tetrahedron, which is equivalent to the phase space of $\SU(2)$ flat connections on a 4-punctured sphere. The quantization of the space of boundary data is then the moduli algebra that we study in this paper. The corresponding Hilbert space is the intertwiner space $W^0(K_1,K_2,K_3,K_4)$ of $\UQ$. 

Given that the boundary data of the spinfoam model is semiclassical, to gain a clearer understanding of the relationship between the quantization and the boundary data, it becomes interesting to construct coherent states in $W^0(K_1,K_2,K_3,K_4)$ parameterized by phase space variables. The result on this aspect will be reported elsewhere \cite{toappear}.

This paper makes the first step toward reformulating the LQG kinematics in order to make it compatible with the spinfoam theory with nonzero cosmological constant. This work focuses on the quantization of a curved tetrahedron, which is at the same level as a single intertwiner in the LQG Hilbert space. We have to generalize the quantization to arbitrary 3D cellular complexes to reformulate the entire LQG Hilbert space and geometrical operators. The result in \cite{Han:2016dnt} shows that the moduli space of SU(2) flat connections on higher-genus surfaces closely relates to the LQG phase space on cellular complexes. Then we expect that applying the combinatorial quantization to flat connections on higher-genus surfaces should lead to an interesting reformulation of the LQG kinematics.

\begin{acknowledgements}
The authors would like to acknowledge Yuting Hu for various helpful discussions. This work receives support from the National Science Foundation through grants PHY-2207763, PHY-2110234, the Blaumann Foundation and the Jumpstart Postdoctoral Program at Florida Atlantic University. The authors acknowledge IQG at FAU Erlangen-N\"urnberg, Perimeter Institute for Theoretical Physics, and University of Western Ontario for the hospitality during their visits.   

\end{acknowledgements}

\appendix
\renewcommand\thesection{\Alph{section}}

\section{Geometry of a homogeneously curved tetrahedron in terms of holonomies}
\label{app:GeoTetra}

In this appendix, we collect the geometrical information of a homogeneously curved tetrahedron stored in the holonomies $M_1,M_2,M_3$ and $M_4$ which can be decomposed in the way of \eqref{eq:simple_solution}. These materials can also be found in \cite{Haggard:2015ima}. 

For convenience, let us first introduce the half-traces of the products of one, two and three holonomies, respectively.
\begin{subequations}
\begin{align}
	\la M_\ell \ra &=\f12\tr(M_\ell)\,,
	\label{eq:half-trace_1}\\
	\la M_{\ell_1}M_{\ell_2}\ra 
	&=\f12 \tr(M_{\ell_1}M_{\ell_2})-\f14 \tr(M_{\ell_1})\tr(M_{\ell_2})\,,
	\label{eq:half-trace_2}\\
	\la M_{\ell_1}M_{\ell_2}M_{\ell_3} \ra 
	&=\f12 \tr(M_{\ell_1}M_{\ell_2}M_{\ell_3}) 
	-\f14\left[\tr(M_{\ell_1})\tr(M_{\ell_2}M_{\ell_3})+ \text{cyclic}\right]
	+\f14 \tr(M_{\ell_1})\tr(M_{\ell_2})\tr(M_{\ell_3})\,.
	\label{eq:half-trace_3}
\end{align}
\label{eq:half-traces}
\end{subequations}
The half-trace \eqref{eq:half-trace_1} of one holonomy $M_\ell$ around a face $\ell$ encodes the area $a_\ell$ of the face; the half-trace \eqref{eq:half-trace_2} of two holonomies $M_{\ell_1}$ and $M_{\ell_2}$ encodes the dihedral angle $\theta_{\ell_1\ell_2}$ of the two faces $\ell_1$ and $\ell_2$; the half-trace \eqref{eq:half-trace_3} of three holonomies $M_{\ell_1}, M_{\ell_2}$ and $M_{\ell_3}$ encodes the triple product of the normals $(\hat{n}_{\ell_1}\times\hat{n}_{\ell_2})\cdot\hat{n}_{\ell_3}$ to the three faces $\ell_1,\ell_2,\ell_3$ calculated at the common vertex of the three faces. 
Explicitly,
\begin{align}
\cos (s a_\ell) =& \epsilon_\ell \la M_\ell\ra\,,\\
\cos \theta_{\ell_1\ell_2} :=& \hat{n}_{\ell_1}\cdot\hat{n}_{\ell_2} = -\f{\epsilon_{\ell_1}\epsilon_{\ell_2}\la M_{\ell_1}M_{\ell_2}\ra}{\sqrt{1-\la M_{\ell_1}\ra^2}\sqrt{1-\la M_{\ell_2}\ra^2}}\,,
\label{eq:cos_H}\\
(\hat{n}_{\ell_1}\times\hat{n}_{\ell_2})\cdot\hat{n}_{\ell_3} =&
-\f{\epsilon_{\ell_1}\epsilon_{\ell_2}\epsilon_{\ell_3}\la M_{\ell_1}M_{\ell_2}M_{\ell_3} \ra}{\sqrt{1-\la M_{\ell_1}\ra^2}\sqrt{1-\la M_{\ell_2}\ra^2}\sqrt{1-\la M_{\ell_3}\ra^2}}\,.
\end{align}
The signs $\{\epsilon_\ell=\pm\}$ are fixed by 
requiring the following inequalities for triple products of the normals evaluated at vertex 4 (referring to fig.\ref{fig:tetrahedra}):
\be
\left\{\ba{rl}
(\hat{n}_1\times\hat{n}_2)\cdot \hat{n}_3&>0\\[0.15cm]
(\hat{n}_1\times\hat{n}_3)\cdot \hat{n}_4&>0\\[0.15cm]
(\hat{n}_2\times\hat{n}_1)\cdot M_1\hat{n}_4&>0\\[0.15cm]
(\hat{n}_3\times\hat{n}_2)\cdot M_3^{-1}\hat{n}_4&>0
\ea\right..
\ee
These four inequalities pick four signs each associated to a face $\ell$ of the tetrahedron which corresponds to $\epsilon_\ell=\sgn \sin(|s|a_\ell)$. 
Therefore, the Gram matrix $\Gram(\theta_{\ell_1\ell_2})$ for a tetrahedron given by the dihedral angles $\{\cos\theta_{\ell_1\ell_2}\}$, hence $\Gram(M_\ell)$ used in Theorem \ref{theorem:Minkowski}, can be written in terms of the holonomies using \eqref{eq:cos_H}.

\section{Mathematical tools: $\UQ$ and representations}
\label{sec:math}

In this appendix, we review concisely $\UQ$ for both cases $|q|=1,q^p\neq1$ and $q^p=1$ algebraically and their representation theories. 

\subsection{$\UQ$ with $|q|=1,q^p \neq 1$ as a quasitriangular ribbon Hopf-$*$ algbera}
\label{subsec:algebraic_strucutre}
Let us start by introducing the Hopf algebra $\UQ$ 
with $|q|=1,q^p \neq 1$. 
It is generated by identity $e$ and $H,X,Y$ subject to relations:
\be
[H,X]=2 X\,,\quad
[H,Y]=-2Y\,,\quad
[X,Y]=\f{q^H-q^{-H}}{q-q^{-1}}\,.
\label{eq:Uqsl_commutator}
\ee
The first two commutation relations can be equivalently written in terms of the $q^{\frac{H}{2}}$ (and its inverse $q^{-\frac{H}{2}}$):
\be
q^{\frac{H}{2}}Xq^{-\frac{H}{2}}=qX\,,\quad
q^{\frac{H}{2}}Yq^{-\frac{H}{2}}=q^{-1}Y\,.
\ee
The definition of $\UQ$ is completed by its co-algebra structure and an antipode. 
The co-structure is defined as
\be
\begin{split}
&\Delta(e)=e\otimes e\,,\quad \Delta(q^{\pm \frac{H}{2}})=q^{\pm \frac{H}{2}}\otimes q^{\pm \frac{H}{2}}\,,\quad
\Delta(X)= X\otimes q^{\frac{H}{2}}+q^{\frac{-H}{2}}\otimes X\,,\quad \Delta(Y)= Y\otimes q^{\frac{H}{2}}+q^{\frac{-H}{2}}\otimes Y\,,\\
&\epsilon(e)=\epsilon(q^{H})=1\,, \quad\epsilon(X)=\epsilon(Y)=0\,,
\end{split}    
\ee
where $\epsilon:\UQ\rightarrow\bC$ is the counit and $\Delta:\UQ\rightarrow \UQ\otimes \UQ$ is the coproduct satisfying the co-associativity,
\be
(\Delta\circ \Id)\circ \Delta =( \Id\circ\Delta)\circ \Delta\,.
\label{eq:co-associativity}
\ee
The antipode $S:\UQ\rightarrow\UQ$ acts on the generators as
\be
S(q^{\frac{H}{2}})=q^{-\frac{H}{2}}\,,\quad
S(X)=-q X\,,\quad S(Y)=-q^{-1}Y\,.   
\ee
It is compatible with the coproduct and counit through
\be
\sum_a S(\xi^{(1)}_a)\xi^{(2)}_a = \sum_a\xi^{(1)}_a S(\xi^{(2)}_a)=\epsilon(\xi)\,,\quad \forall \xi\in \UQsl\,,
\label{eq:S_compatible}
\ee
where $\Delta(\xi)=\sum_a \xi^{(1)}_a\otimes \xi^{(2)}_a$. 
Note that the antipode is not an involution. Instead, the following relation holds.
\be
S^2(\xi) = q^H \xi q^{-H}\quad \forall \xi\in\UQ\,.
\label{eq:S2}
\ee
The Hopf algebra $\UQ$ is of {\it quasitriangular} type.
The quasitriangularity is realized by a {\it quantum $\cR$-matrix} $R\in\UQ\otimes\UQ$ which is a solution to the {\it quantum Yang-Baxter equation} (QYBE)
\be
 R_{12}R_{13}R_{23}=R_{23}R_{13}R_{12}\,.
 \label{eq:QYBE}
\ee
Here we have used the same notation as for the classical $r$-matrix, \ie $R_{12}=\sum_a R^{(1)}_a\otimes R^{(2)}_a \otimes \id, R_{13}=\sum_a R^{(1)}_a \otimes \id\otimes R^{(2)}_a$ and $R_{23}=\sum_a  \id\otimes R^{(1)}_a\otimes R^{(2)}_a $. 
The permuted coproduct $\Delta'$ is related to $\Delta$ through $R$:
\be
\Delta':=\sigma\circ \Delta = R\Delta R^{-1}\,,
\label{eq:DeltaP}
\ee
where $\sigma$ is the permutation operator \ie $\sigma(\xi\otimes\eta)=\eta\otimes\xi\,,\forall \xi,\eta\in\UQ$. 
Similarly, we also denote the $\cR$-matrix with two vector subspace permuted as
\be
R':=\sigma\circ R\equiv\sum_a R^{(2)}_a\otimes R^{(1)}_a\,.
\ee
There are some properties of $R$-matrix that we use to derive the theorem of this paper:
\be
\begin{aligned}
 (S\otimes\Id)(R)=(\Id\otimes S)(R)=R^{-1}\quad&,\quad(S\otimes S)(R)=R,
 \\
 (\epsilon\otimes\Id)(R)=(\Id\otimes\epsilon)(R)=e\quad&,\quad
 (\epsilon\otimes\Id)(R')=(\Id\otimes\epsilon)(R')=e.
\end{aligned}
\ee
The deformation parameter $q\equiv e^{\hbar}$
in $\UQ$ is inherited by $R$ hence $R=R(\hbar)$. 
Taking the $\hbar$ expansion of $R(\hbar)$, the classical $r$-matrix \eqref{eq:r_su2} is recovered in the first-$\hbar$ order:
\be
R(\hbar)=\id +\hbar r+O(\hbar^2)\,.
\label{eq:R_r}
\ee
In this sense, \eqref{eq:QYBE} is the quantum version of the CYBE \eqref{eq:CYBE}. 
The $R$-matrix for $\UQ$ is defined as
\be
R=\sum_{n\in \N}\frac{(q-q^{-1})^n}{[n]!}q^{-n(n+1)/2}q^{\frac{1}{2}(H\otimes H)+\f{n}{2}(H\otimes e -e\otimes H)}(X^n\otimes Y^n)\,,
\label{eq:R_def}
\ee
where $[n]:=\frac{q^n-q^{-n}}{q-q^{-1}}$ is called the $q$ number and $[n]!:=[1][2]...[n]$ for $n>1$ while $[0]!=[1]!\equiv 1$. 

\medskip

$\UQ$ is the Hopf $*$-algebra consisting of the Hopf algebra $\UQsl$ for $|q|=1$, and the $*$-structure defined as: 
\be
H^*=H\,,\quad X^*=Y\,,\quad Y^*=X\,.
\ee
It follows that $(R)^*\equiv R^{-1}$. 
The $*$-operation is an anti-homomorphism and can be viewed as analogous to the conjugate transpose operation for matrices in the sense that $(\lambda \xi\eta)^*=\bar{\lambda}\eta^*\xi^*\,,\forall \lambda\in\bC,\xi,\eta\in\cU_q(\mathfrak{su}(2))$, where the bar denotes the complex conjugate. 
Therefore, in the case of $|q|=1$, $H^*=H$ can be equivalently written as $(q^{\pm \frac{H}{2}})^*=q^{\mp \frac{H}{2}}$.  

We also require that the $*$-operation on $\cU_q(\mathfrak{su}(2))\otimes \cU_q(\mathfrak{su}(2))$ behaves as \cite{Alekseev:1994pa}   
\be
(\xi\otimes \eta)^{*}=\eta^{*}\otimes \xi^{*} \,,\quad
\forall \xi,\eta\in\cU_q(\mathfrak{su}(2))\,.
\label{eq:star_tensor_product}
\ee
Moreover, the following properties hold for a general element $\xi\in\UQ$.
\be
S(\xi^{*})=S(\xi)^{*}\,,\quad 
\epsilon(\xi^{*})=\overline{\epsilon(\xi)}\,,\quad
\Delta(\xi^{*})=(\Delta(\xi))^{*} \,.
\label{eq:star_coalgbera}
\ee  
\eqref{eq:star_tensor_product} and \eqref{eq:star_coalgbera} lead to a simple result for $R$ under the $*$-operation: $R^{*}=(S\otimes\Id)(R)=R_{q^{-1}}\equiv R^{-1}$.

To describe $\UQ$ as a ribbon Hopf algebra, we also introduce a {\it ribbon element} $v\in\UQ$ which is an invertible central element defined as
\be
v^2=uS(u)\,,\quad
S(v)=v\,,\quad
\epsilon(v)=1\,,\quad
\Delta(v)=(R'R)^{-1}(v\otimes v)\equiv (v\otimes v)(R'R)^{-1}\quad
\text{ with }\,
u:=\sum_a S(R^{(2)}_{a})R^{(1)}_{a}\,.
\label{eq:def_v}
\ee 
It has been known that such a ribbon element exists for $\UQ$ \cite{Reshetikhin:1990pr}. 
The $*$-operation acts on the ribbon element as $v^{*}=v^{-1}$.

In fact, the above construction can also be generalized to the case of a generic $q\in \bC$. The only difference is that $*$-operation is defined differently for the tensor products of $\UQ$ \cite{Alekseev:1994pa,Alekseev:1994au,Alekseev:1995rn}. That is, for generic $q$, $(\xi\otimes \eta)^{*}=\xi^{*}\otimes \eta^{*} \,,
\forall \xi,\eta\in\UQ$ and the follow-up formulas would be changed.

\medskip

\noindent{\bf The representation theory. }
\medskip

 Let us also give the representation theory of $\UQ$ with $|q|=1,q^p\neq 1$. In this case, $\UQ$ is a semi-simple Hopf algebra and its representation is a deformed version of that for $\su(2)$. For a unitary representation $\rho:\UQ\rightarrow\End(V)$ on a Hilbert space $V$, $(\rho(\xi))^{*}=\rho(\xi^{*})$ holds for all $\xi \in \UQ$. 

The tensor product of two representations $\rho_1$ and $\rho_2$ is expressed in terms of the coproduct $\Delta$:
\be
\rho_1 \boxtimes \rho_2(\xi)=(\rho_1 \otimes \rho_2)(\Delta(\xi))  \,,\quad\forall\, \xi\in\UQ\,.
\ee
Note that the tensor product of two unitary representations is not unitary due to our choice \eqref{eq:star_tensor_product} of $*$-operation on $\UQ\otimes\UQ$. Instead, we have
\be
(\rho_1 \otimes \rho_2)(\Delta'(\xi^{*}))=\left( (\rho_1 \otimes \rho_2)(\Delta(\xi))\right)^{*}\,,   
\label{eq:*-rep_product}
\ee
where $\Delta'$ is the permuted coproduct defined in \eqref{eq:DeltaP}.

For every equivalence class $[J]$ of irreducible representations (with label $J\in\N/2$ being a half-integer), there exists a unitary representation $\rho^J$ with carrier space $V^J$. 

The unitary representations of the $\UQ$ generators act on the basis of $V^J$ as
\begin{subequations}
\begin{align}
\rho^J(q^{\frac{H}{2}})e^J_{m}&=q^{m}\,e^J_{m}\,,\\
\rho^{J}(X)e_{m}^{J}&=\sqrt{[J-m]_{q}[J+m+1]_{q}}\,e_{m+1}^{J}\,,\\
\rho^{J}(Y)e_{m}^{J}&=\sqrt{[J+m]_{q}[J-m+1]_{q}}\,e_{m-1}^{J} \,.
\end{align}
\label{eq:rho_generators}
\end{subequations}
The representation of the unit element $e$ in $V^J$ gives the identity matrix of dimension $2J+1$, \ie $\rho^J(e)=\Id_J$. 

Due to semi-simplicity, the tensor product $\rho^I\boxtimes\rho^J$ of unitary representations is decomposable into the direct sum of irreducible representations
\be
(\rho^I\boxtimes \rho^J)=\bigoplus_{K=|I-J|}^{I+J}\rho^K \,.   
\ee
This decomposition determines the Clebsch-Gordon (CG) maps $C[IJ|K]$: $V^{I}\otimes V^{J} \to V^{K}$ up to normalization,
\be
C[IJ|K](\rho^{I} \boxtimes \rho^{J})(\xi)= \rho^{K}(\xi)C[IJ|K]\,,\quad
\forall \xi\in\UQ\,,
\label{eq:CG_def1}
\ee
Taking the $*$-operation on both sides of \eqref{eq:CG_def1}, on can define another CG maps $C[IJ|K]^{*}$:$V^{K}\to V^{I}\otimes V^{J}$,
\be
(\rho^{I} \otimes \rho^{J})\Delta'(\xi)C[IJ|K]^{*}= C[IJ|K]^{*}\rho^{K}(\xi)\,,\quad 
\forall \xi\in\UQ\,.
\label{eq:CG_def2}   
\ee
Denote $R^{IJ}:=\sum_a\rho^I(R^{(1)}_a)\otimes \rho^J(R^{(2)}_a)$ as the the representation of $R$ in $\End(V^I)\otimes \End(V^J)$. Define $\widetilde{R}^{IJ}:=\sigma^{IJ}\circ R^{IJ}$ with $\sigma^{IJ}:V^I\otimes V^J\rightarrow V^J\otimes V^I$ being the permutation operator. $\widetilde{R}^{IJ}$ is called the braiding and it  furnishes an intertwining relation between $\rho^{I}\boxtimes\rho^{J}(\xi)$ and $\rho^{J}\boxtimes\rho^{I}(\xi)$ in the way that
\be
(\rho^{J}\boxtimes\rho^{I})(\xi)\widetilde{R}^{IJ}=\widetilde{R}^{IJ}(\rho^{I}\boxtimes\rho^{J})(\xi)\,.
\label{eq:intertwine}
\ee
It can be proven by taking the representation of the relation \eqref{eq:DeltaP}, which gives
\be
(\rho^{I}\otimes\rho^{J})\Delta'(\xi)R^{IJ}=R^{IJ}(\rho^{I}\otimes\rho^{J})\Delta(\xi)\,. 
\ee
Permuting the two representation spaces for both sides, one obtains the relation \eqref{eq:intertwine}.

Denote the representation of the ribbon element $v$ in $V^I$ as $v^I:=\rho^I(v)$. Since $v$ is a central element, $v^I$ is simply a complex number by Schur's Lemma. 
By the definition of $u$ defined in \eqref{eq:def_v}, the following relation holds \cite{Kassel:1995xr}. 
\be
S^2(\xi)=u\xi u^{-1} \quad \forall \xi\in\UQ\,.
\ee
Recalling the relation \eqref{eq:S2}, $q^{-H}u\equiv q^{-H}u$ must be a central element. Its invertibility is obvious from the definition. The element $v$ satisfies the defining relations of a ribbon element. Therefore, the ribbon element $v=q^{-H}u$. 
\begin{lemma}
\label{lemma:vJ}
 The ribbon element $v=q^{-H}u$ in the irreducible representation $\rho^J$ has value $q^{-2J(J+1)}$.    
\end{lemma}
\begin{proof}
 By Schurs' lemma, the irreducible representation of the central element $v=uq^{-H}$ on space $V^J$ is proportional to the identity $\Id^J$ times some complex factor. We can find the complex factor by the actions of $\UQ$ generators on the highest weight vector $|J,J\rangle$, which are the special case of \eqref{eq:rho_generators} and read 
\be
q^{\f{H}{2}}|J,J\ra = q^{J}|J,J\ra\,,\quad
X|J,J\ra=0\,,\quad
Y|J,J\ra=\sqrt{[2J]_{q}}|J,J-1\ra\,.
\label{eq:UQ_on_highest_weight}
\ee
Then
\begin{equation}
v|J,J\rangle =u q^{-H}|J,J\rangle
=\left[q^{-\frac{H^2}{2}-H}+\left(\sum_{l=1}^{\infty} \frac{(q-q^{-1})^{l}}{[l]_q!}q^{-l(l+1)/2} (S(Y))^{l}q^{\frac{1}{2}(-H^2+2lH)}X^l\right)q^{-H}\right]|J,J\rangle
=q^{-2J(J+1)}|J,J\rangle  \,,
\end{equation}   
where we have used the definition of $u$ given in \eqref{eq:def_v}. The first term in the square bracket is the $l=0$ term of the summation in the R and the second term acts on $|J,J\ra$ trivially due to \eqref{eq:UQ_on_highest_weight}. 
\end{proof}
Take its square root and define $\kappa_J$ by $\kappa_J^2=v^J$. 
The expression of $\kappa_I$ in the irreducible representation $\rho^J$:
\be 
\kappa_J=q^{-J(J+1)}\,.
\label{eq:v_kappa_J}
\ee
Here, we have chosen $\kappa_I$ to be the positive square root of $v^I$ so that the formulas in the rest of this section, \eg \eqref{eq:CG_normalization}, take the forms that matche the literature \cite{Alekseev:1994au,Alekseev:1994pa,Alekseev:1995rn}. 
The normalization of the CG maps is given by a set of $\kappa_I$'s and we prove it (in a different way from the literature) in the following proposition.  
\begin{prop}
Let the CG maps $C[IJ|K]$ and $C[IJ|K]^*$ for $\UQ$ ($|q|=1,q^p\neq 1$) act on the vector space bases as 
\begin{subequations}
\begin{align}
C[IJ|K](e^I_m\otimes e^J_n)&=\sum_{k}\mat{cc|c}{I&J&K\\m&n&k}_q e^K_k\,,
\label{eq:CG_action}
\\
C[IJ|K]^*(e^K_k)&=\sum_{m,n}\overline{\mat{cc|c}{I&J&K\\m&n&k}}_q e^I_m\otimes e^J_n
\equiv \sum_{m,n}\mat{cc|c}
{I&J&K\\m&n&k}_{q^{-1}} e^I_m\otimes e^J_n\,,
\label{eq:CGstar_action}
\end{align}
\end{subequations}
with $\mat{cc|c}{I&J&K\\m&n&k}_q$ and $\mat{cc|c}{I&J&K\\m&n&k}_{q^{-1}}$ being complex coefficients. 
Assume the following normalization and symmetry of the coefficients:
\begin{align}
&\sum_{m,n}\mat{cc|c}{I&J&K\\m&n&k}_q\mat{cc|c}{I&J&K'\\m&n&k'}_q=\delta_{KK'}\delta_{kk'}\,,
\label{eq:normalization_1}\\
&\sum_{k}\mat{cc|c}{I&J&K\\m&n&k}_q\mat{cc|c}{I&J&K\\m'&n'&k}_q=\delta_{mm'}\delta_{nn'}\,,
\label{eq:normalization_2}\\
&\mat{cc|c}{I&J&K\\m&n&k}_{q^{-1}} =(-1)^{I+J-K}\mat{cc|c}{J&I&K\\n&m&k}_q\,.
\label{eq:CG_qq-1}
\end{align}
Then the CG maps satisfy the following normalization relation
\be
C[IJ|K](R')^{IJ}C[IJ|K']^{*}= \delta_{KK'}\frac{\kappa_I \kappa_J}{\kappa_K} \Id_{K}\,.
\label{eq:CG_normalization}
\ee
\end{prop}
\begin{proof}
We first prove that $\mat{cc|c}{I&J&K\\m&n&k}_q$ satisfying the conditions \eqref{eq:normalization_1}--\eqref{eq:CG_qq-1} takes the same form of the CG coefficients for $\UQ$ with $q$ real. 

Let $\rho^I\otimes \rho^J(\Delta(X))$ act on the basis vector $ e^I_m\otimes e^J_n$ in $V^I\otimes V^J$. According to \eqref{eq:CG_action}, one obtains
\be
\begin{split}
&C[IJ|K]((\rho^I\otimes \rho^J(\Delta(X)))( e^I_m\otimes e^J_n) )\\
=&\sum_{k} \lb q^{n}\sqrt{[I-m]_q[I+m+1]_q}\mat{cc|c}{I&J&K\\m+1&n&k}_q
+q^{-m}\sqrt{[j-n]_q[j+n+1]_q}\mat{cc|c}{I&J&K\\m&n+1&k}_q \rb e^K_k\,.
\end{split}
\ee
By the definition \eqref{eq:CG_def1}, this is equivalent to 
\be
\rho^K(X)C[IJ|K](e^I_m\otimes e^J_n)
=\sum_k \sqrt{[K-k]_q[K+k+1]_q} \mat{cc|c}{I&J&K\\m&n&k}_q e^K_{k+1}\,.
\ee
Similarly, replacing $\rho^K(X)$ by $\rho^K(Y)$, a similar identity can be obtained. 
These are exactly the recursion relations of the CG coefficients of $\UQ$ with $q$ real \cite{Biedenharn:1996vv}:
\begin{multline}
\sqrt{[K\mp k]_q[K\pm k+1]}\mat{cc|c}{I&J&K\\m&n&k\pm 1}_q \\
=q^{-m}\sqrt{[J\pm n]_q[J\mp n+1]_q} \mat{cc|c}{I&J&K\\m&n\mp 1&k}_q
+q^{n}\sqrt{[I\pm m]_q[I\mp m+1]_q} \mat{cc|c}{I&J&K\\m\mp 1&n&k}_q\,.
\label{eq:CG_recursion}
\end{multline}
The recursion relation and the normalization \eqref{eq:normalization_1} and \eqref{eq:normalization_2} of the coefficients $\mat{cc|c}{I&J&K\\m&n&k}_q$ proves that it takes the same expression as the CG coefficient of $\UQ$ with $q$ real:
\begin{multline}
\mat{c c|c}{j_1 & j_2 & J\\ m_1 & m_2 & M}_{q}
= \Delta(j_{1}j_{2}J)\left([2J+1]_{q}![J+M]_{q}![J-M]_{q}![j_{1}+m_{1}]_{q}![j_{1}-m_{1}]_{q}![j_{2}+m_{2}]_{q}![j_{2}-m_{2}]_{q}!\right)^{1/2}\\
\sum_{s=0}^{j_1+j_2-J}\frac{q^{\frac{1}{2}(j_{1}+j_{2}-J)(j_{1}+j_{2}+J+1)+\frac{j_{1}m_{2}-j_{2}m_{1}}{2}}q^{s(j_{1}+j_{2}+J+1)}(-1)^{s}}{[s]_{q}![j_{1}+j_{2}-J-s]_{q}[j_{1}-m_{1}-s]_{q}![j_{2}+m_{2}-s]_{q}![J-j_{2}+m_{1}-s]_{q}![J-j_{1}-m_{2}+s]_{q}!}   \,,
\label{eq:CG_expression}
\end{multline}
where $\Delta(j_1j_2J)$ is defined as
\be
\Delta(j_1j_2J)=\left(\frac{[j_{1}-j_{2}+J]_{q}![J-j_{1}+j_{2}]_{q}![j_{1}+j_{2}-J]_{q}!}{[J+j_{1}+j_{2}+1]_{q}!}\right)^{1/2}.
\ee
The CG coefficients satisfy the following relation \cite{A.N.Kirillov:CGcoefficients}.
\be
\sum_{m,n}(R^{IJ})^{mn}_{m'n'}\mat{cc|c}{I&J&K\\m&n&k}_q
=(-1)^{I+J-K} q^{K(K+1)-I(I+1)-J(J+1)}\mat{cc|c}{J&I&K\\n'&m'&k}_q\,.
\label{eq:CG_R}
\ee
Combining \eqref{eq:CG_R} and the normalization \eqref{eq:normalization_1} of the CG coefficients, one proves the normalization relation \eqref{eq:CG_normalization}.
\end{proof}

From the normalization and $\Delta(e)=e \otimes e$, we obtain completeness for CG maps: 
\be
\sum_{K} \frac{\kappa_K}{\kappa_I \kappa_J} (R')^{IJ}C[IJ|K]^{*}C[IJ|K]= e^{I}\otimes e^{J} \,.
\label{eq:CG_completeness}
\ee
Due to the fact that the recoupling of CG coefficients produces the ``$6j$-symbols'' $\Mat{ccc}{ K & J & P\\ I & L & Q} _q$ of $\UQ$ \cite{Biedenharn:1996vv}, the sequence of CG maps follow the identity 
\be
 C[IP|L]C[JK|P]=\sum_{|I-J|\leq Q \leq I+J}
\Mat{ccc}{ K & J & P\\ I & L & Q} _qC[QK|L]C[IJ|Q] \,.
\label{eq:CG-6j-identity_1}
\ee
Such an identity takes the same shape as in the $\SU(2)$ recoupling theory but the CG maps and the $6j$-symbols are all $q$-deformed. 

We define the {\it quantum trace} of any element $X^{J}\in \End(V^{J})$ by
\be
\tr^{J}_{q}(X^{J})=\tr^{J}(X^{J}\rho^{J}(g))\,,   
\label{eq:quantum_trace_g}
\ee
where $g:=u^{-1}v$ is a group-like element satisfying
\be
\Delta(g)=g\otimes g\,,\quad
g^{*}=g^{-1}\,,\quad
S(g)=g^{-1}\,,\quad
g S(\xi)=S^{-1}(\xi)g   \,.
\label{eq:g_properties}
\ee
The antipode $S$ of $\UQ$ furnishes a conjugation in the set of equivalence classes of
irreducible representations. We use $[\Bar{J}]$ to denote the class conjugate to $[J]$. Another definition of quantum trace is to use the CG map $C[\Bar{J}J|0]:V^{\Bar{J}}\otimes V^{J}\to 0$ \cite{Alekseev:1994au}.
\be
Tr^J_q(X^J)=\frac{d_J}{v_J}C[\Bar{J}J|0]\stackrel{2}{X^J} (R')^{\Bar{J}J} C[\Bar{J}J|0]^*\,.
\label{eq:quantum_trace_CG}
\ee
This two definitions of quantum trace are equivalent. The proof is given in \cite{Alekseev:1994au}.  

 The group-like element $g$ is simply $q^{-H}$ since $v=uq^{-H}$. The quantum trace of $\Id_J$ defines the {\it quantum dimension}, denoted as $d_J$, of the irreducible representation $\rho^{J}$:
\be
d_J=\tr^{J}_{q}(\Id_J)\equiv tr^{J}(q^{-H})=[2J+1]_q\,,\quad 
d_I d_J=\sum_{K=|I-J|}^{I+J} d_K\,.
\label{eq:quantum_dimension}
\ee
In general, quantum dimension $d_J$ is different from the dimension of the carrier space $V^{J}$. The numbers $v^I$ and $d_I$ are symmetric under conjugation, i.e. $v^I=v^{\Bar{I}}$ and $d_I=d_{\Bar{I}}$ \cite{Alekseev:1994au}.

\subsection{$\TUQ$ as a weak quasitriangular ribbon quasi-Hopf-$*$ algebra}
\label{subsec:TUQ}

When $q$ is a root-of-unity, say $q^p=1 \,(p\in \N^+,p\geq2)$, $\UQ$ is no longer a Hopf algebra but a quasi-Hopf algebra. More importantly, the semi-simplicity is lost, bringing complexity to its algebraic structure and the representation theory. However, one can construct a ``truncated algebra'', denoted as $\TUQ$, that is canonically associated with $\UQ$ with $q$ a root-of-unity but {\it is} semi-simple. Even though the algebraic structure of $\TUQ$ is also more complicated than that of $\UQ$ with $q^p=1$, its representation theory is {\it almost} the same as in the case of $\UQ$ with $q^p\neq 1$ as described in Section \ref{subsec:algebraic_strucutre}, which brings enormous benefits in constructing the quantum theory. Therefore, in this section, we only sketch its algebraic structure as a weak quasitriangular ribbon quasi-Hopf-$*$ algebra and focus on the representation theory of $\TUQ$ (more precisely the difference in the representation theory from $\UQ$ with $q^p\neq 1$). 

Algebraically, $\TUQ$ is a weak quasitriangular ribbon quasi-Hopf-$*$ algebra generated by the same generators as $\UQ$. To distinguish notations from those of $\UQ$, we denote its coproduct, counit, antipode and quantum $\cR$-matrix as $\Delta^T,\epsilon^T,S^T$ and $R^T$ respectively. As a quasi-Hopf algebra, the associativity is relaxed up to conjugation controled by an element $\varphi\in \TUQ\otimes \TUQ\otimes \TUQ$:
\be
(\Id\otimes \Delta^T)\Delta^T(\xi)\varphi =\varphi(\Delta^T \otimes \id)\Delta^T(\xi)\,, \quad \forall\xi\in \TUQ\,.
\ee
Such $\varphi$ satisfies the identity
\be
(\Id\otimes \Id\otimes \Delta^T)(\varphi)(\Delta^T\otimes \Id\otimes \Id)(\varphi)=(e\otimes \varphi)(\Id\otimes \Delta^T \otimes\Id)(\varphi)(\varphi\otimes e)\,,
\ee
where $e$ is the identity of $\TUQ$. 
Moreover, the compatibility of the antipode $S^T$ with the coproduct and counit is given additionally in terms of $\alpha,\beta\in \TUQ$:
\begin{subequations}
\begin{align}
\sum_a S^T(\xi^{(1)}_a)\alpha\xi^{(2)}_a=\epsilon^T(\xi)\alpha\,,\quad &\quad
\sum_a \xi^{(1)}_a\beta S^T(\xi^{(2)}_a)=\epsilon^T(\xi)\beta\,,\quad \forall \xi\in\TUQ\,,\\
\sum_a \varphi^{(1)}_a\beta S^T(\varphi^{(2)}_a)\alpha \varphi^{(3)}_a=e\,,\quad & \quad \sum_a S^T(\phi^{(1)}_a)\alpha\phi^{(2)}_a\beta S^T(\phi^{(3)}_a)=e\,,
\end{align}
\end{subequations}
where we have used the standard notations $\Delta(\xi)=\sum_a\xi^{(1)}_a\otimes \xi^{(2)}_a$, $\varphi=\sum_a \varphi^{(1)}_a\otimes\varphi^{(2)}_a\otimes\varphi^{(3)}_a$ and $\phi:=\varphi^{-1}=\sum_a\phi^{(1)}_a\otimes\phi^{(2)}_a\otimes\phi^{(3)}_a$. 

Its quasitriangular structure is given by $R^T$ satisfying the quasi-QYBE
\be
R^T_{12} \varphi_{312} R^T_{13} \varphi_{132}^{-1} R^T_{23} \varphi = \varphi_{321} R^T_{23} (\varphi_{231}^{-1}) R^T_{13} \varphi_{213} R^T_{12}\,.
\ee
Here, $\varphi_{312}=\sum_a\varphi^{(3)}_a\otimes\varphi^{(1)}_a\otimes\varphi^{(2)}_a$, \etc. 
The compatibility of $R^T$ and $\Delta^T$ is the following. 
\begin{subequations}
\begin{align}
   \Delta^{T\,'}(\xi) R^T&=R^T \Delta^T(\xi)\,,\\
(\Id \otimes \Delta^T)(R^T)&=\varphi_{231}^{-1}R^T_{13}\varphi_{213}R^T_{12}\varphi^{-1}\,,\\
(\Delta^T \otimes \Id)(R^T)&=\varphi_{312}R^T_{13}\varphi_{132}^{-1}R^T_{23}\varphi \,,
\end{align}
\end{subequations}
where $\Delta^{T\,'}=\sigma\circ\Delta^T$. 

$\TUQ$ is a {\it weak} quasi-Hopf algebra in the sense that $\Delta^T(e)\neq e\otimes e$. Instead, the coproduct of $e$ is given in terms of $\varphi$:
\be
\Delta^T(e)=( \epsilon^T\otimes \Id \otimes \Id)(\varphi)=(\Id\otimes \epsilon^T \otimes \Id)(\varphi)=( \Id\otimes \Id \otimes \epsilon^T)(\varphi)\,.
\ee
Moreover, we only have quasi-inverse for $\varphi$ and $R^T$ such that
\begin{align}
\varphi \varphi^{-1}=(\Id\otimes\Delta^T)\Delta^T(e)\,,\quad&\quad\varphi^{-1}\varphi =(\Delta^T\otimes \Id)\Delta^T(e)\,,\\
R^TR^{T\,-1}=\Delta^{T\,'}(e)\,,\quad&\quad R^{T\,-1}R^T=\Delta^T(e)\,.
\end{align}

The $*$-structure of $\TUQ$ satisfies the same properties as in \eqref{eq:star_tensor_product} and \eqref{eq:star_coalgbera}. In addition, one needs to define the $*$-operation for $\varphi,\alpha,\beta$ so that it is consistent with all the above relations:   
\be
\varphi^{*}=(\sum_a \varphi_{a}^{(1)}\otimes \varphi_{a}^{(2)} \otimes \varphi_{a}^{(3)})^{*}=\sum_a (\varphi_{a}^{(3)})^{*}\otimes (\varphi_{a}^{(2)})^{*} \otimes (\varphi_{a}^{(1)})^{*}=\varphi\,,\quad
\alpha^{*}=\beta\,.
\ee
Finally, the ribbon element of $\TUQ$ is defined in the same way as in \eqref{eq:def_v} but with $S^T,\epsilon^T,\Delta^T$ and $R^T$ \cite{Altschuler:1992zz}. 

For our purpose, the precise formulas of $\varphi,\alpha$ and $\beta$ are not important hence the above definition of $\TUQ$ as a weak quasitriangular ribbon quasi-Hopf-$*$ algebra is only formal. What is relevant to us is the representation theory of $\TUQ$, which shares great similarity with that of $\UQ$ with $q^p\neq 1$. In practice, we build up the representation theory of $\TUQ$ from that of $\UQ$, which we describe below. 

To begin with, we define the physical representations $\rho_{\phys}^J$ among all the irreducible representations $\rho^J$ of $\UQ$ as those with labels $J=0,\f12,\cdots, \f{k}{2}$  \footnote{The positive integer $k$ is the Chern-Simons level for Chern-Simons theory described by $\UQ$ with $q^p=1$. In general, $p=k+n$ with $n$ being the dual Coxeter number. For the case of $\SU(2)$ being the gauge group, $n=2$.}. 
Then 
\be
\TUQ\equiv \UQ/\mathcal{J}\,,
\ee
where $\mathcal{J}$ is the ideal which is annihilated by all the physical representations $\rho_\phys^J$.
$k$ then provides a cutoff for the physical representation which renders finiteness in the representation theory of $\TUQ$. 

In a tensor product space, such a cutoff can be described by a projector $P\in\TUQ\otimes\TUQ$ such that $P^{IJ}:=\rho^I\otimes\rho^J(P)$ projects on the physical representations $\rho_\phys^{K}, |I-J|\leq K \leq \min \{I+J,k-I-J \}$.
The quantum $\cR$-matrix and the coproduct $\Delta^T$ of $\TUQ$ are related to those of $\UQ$ by this projector through\footnotemark{}
\be
 R^T=RP\,,\quad 
 \Delta^T(\xi)=P \Delta(\xi)\,,\quad \forall \xi\in\TUQ\,.
\ee

In the remaining, we only work on the physical representations and omit the subscript $\phys$ for conciseness. 

To recover semi-simplicity, one introduces the truncated tensor product $\underline{\boxtimes}$ for two physical representations which correspond to $\Delta^T$:
\be
 \rho^I\underline{\boxtimes}\rho^J=\bigoplus_{|I-J|\leq K \leq u(I,J)} \rho^{K}\,,\quad
 u(I,J)=\min \{I+J,k-I-J \}    \,.
\ee
The uplifting fact is that, as long as we stay in the physical representation, the recoupling theory for two representations remains the same as in $\UQ$ with $q^p\neq 1$. That is, the CG maps $C[IJ|K], C[IJ|K]^*$ described in Section \ref{subsec:algebraic_strucutre} can be used safely in the representation theory of $\TUQ$ for $0\leq I,J,K\leq \f{k}{2}$. When several CG maps are acted in a sequence, the order of recouplings matters due to non-associativity. In this case, the representation of $\varphi$ comes into play. Define $\varphi^{IJK}:=\rho^I\otimes\rho^J\otimes\rho^K(\varphi)$. The following identity holds.
\be
 C[IP|L]C[JK|P]\varphi^{IJK}=\sum_{|I-J|\leq Q\leq u(I,J)}
\Mat{ccc}{ K & J & P\\ I & L & Q} _qC[QK|L] C[IJ|Q]\,,
\label{eq:CG-6j-identity}
\ee
where $\Mat{ccc}{ K & J & P\\ I & L & Q} _q$ is the 6j-symbol of $\UQ$ evaluated in physical representations. This is a modification version of \eqref{eq:CG-6j-identity_1}. The proof and more details can be found in \cite{MACK1992185}.

\section{Representation theory of loop algebra $\cL_{0,1}$ for $|q|=1,q^p\neq 1$}
\label{app:reps_q_generic}

Since the loop algebra $\cL_{0,1}$ is isomorphic to $\UQ$ for $|q|=1,q^p\neq1$, its irreducible representation is the same as the one of $\UQ$, in particular, we focus on the unitary irreducible representations of $\UQ$ labeled by $J\in \N/2$. The representation of the generator $\bX^I$ of $\cL_{0,1}$ in the space $V^J$ is given by:
\be
\rho^{J}(\bX^I)=\kappa_I^{-1} (R'R)^{IJ}\,. 
\label{eq:rep_X}
\ee
Moreover, the isomorphism can be extended to $\cL_{0,4}$. The representations of $\cL_{0,4}$ are enumerated by tuples $(K_1, K_2, K_3, K_4)$ of irreducible representations of $\UQ$, where $K_\nu$ is the representation assigned to the $\nu$-th puncture realized in the vector space
\be
V^{K_1K_2K_3K_4}=V^{K_1}\otimes V^{K_2}\otimes V^{K_3}\otimes V^{K_4}\,.
\label{eq:V1234}
\ee
The expression of the generator $\bX_{\nu}^{I}$ of $\cL_{0,4}$ in the representation space $V^{K_1K_2K_3K_4}$ is given by: 
\be
\rho^{K_1K_2K_3K_4}(\bX_{\nu}^{I})=\kappa_{I}^{-1} \lb R'_{12}R'_{13}\cdots R'_{1\nu} R'_{1,\nu+1}R_{1,\nu+1}R^{'\,-1}_{1\nu}\cdots R^{'\,-1}_{13}R^{'\,-1}_{12}\rb^{I K_1..K_\nu}\otimes e^{K_{\nu+1}}\otimes...\otimes e^{K_{4}}\,. 
\ee
The invariant sub-algebra $\cA_{0,4}$ equivalently can be defined as a commutant of the image of $\UQ$ under the embedding, \ie $\cA_{0,4}=\{A|\bX^I_0 A=A\bX^I_0\,,\,\forall\, \bX^I_0\in\cL_{0,4}\}$. 
 The representation $\rho^{K_1K_2K_3K_4}$ can be naturally restricted to invariant sub-algebra $\cA_{0,4}$ of $\cL_{0,4}$. Such representation is reducible. The tensor product of representations is decomposable due to the semi-simplicity of $\UQ$ for $q^p\neq 1$:
 \be
 V^{K_1K_2K_3K_4}=\bigoplus_{J\text{ admissible}}V^J \otimes W^J(K_1, K_2, K_3, K_4)\,. 
 \ee   
Take the representation $ V^{K_1K_2K_3K_4}(\bX_0^I)$ of $\bX_0^I$ which can be decomposed into irreducible representations as above. In each subspace $V^J$, the commutant of $V^J(\bX_0^I)$, which gives the $J$ representation of $\cA_{0,4}$, must live in the multiplicity space $W^J(K_1, K_2, K_3, K_4)$. 

In addition, by the definition \eqref{eq:def_M0} of the moduli algebra,  
the representation of the moduli algebra is further restricted to 
$W^0(K_1,K_2,K_3,K_4)$.
 So we obtain the irreducible representation of moduli algebra realized in the space \cite{Alekseev:1993gh}
 \be
 W^0(K_1, K_2, K_3, K_4)=\Inv_q(V^{K_1K_2K_3K_4})\,.
 \ee
 To endow the representation space with the structure of  Hilbert space, one needs to define the inner product, which we choose to be the same as \eqref{eq:Scalar_reps}\footnote{In general, the inner product will be different than the one we choose here when one changes the star structure of the algebra.} since we focus on $|q|=1,q^p\neq 1$.

\section{Isomorphism between graph algebra $\cL_{0,4}$ and $\UQ^{\otimes 4}$}
\label{app:isomorphism}
Recall that the isomorphism between loop algebra and $\UQ$ described in Section \ref{subsec:moduli_algebra} allows one to set $\bX^I=\kappa_I^{-1}\bX_{+}^I(\bX_{-}^{-1})^{I}$, where $\bX_{+}^{I}=(R')^I$ and $\bX^{I}_{-}=(R^{-1})^{I}$ \cite{Alekseev:1993gh}. We have the following commutation relation for $\bX_{+}^{I},\bX^{J}_{-}$.
\be
\begin{aligned}
\stackrel{1}{\bX_{\pm}^{I}}\stackrel{2}{\bX_{\pm}^{J}}R^{IJ}=R^{IJ}\stackrel{2}{\bX_{\pm}^{J}}\stackrel{1}{\bX_{\pm}^{I}}&,\quad(R')^{IJ}\stackrel{1}{\bX_{\pm}^{I}}\stackrel{2}{\bX_{\pm}^{J}}=\stackrel{2}{\bX_{\pm}^{J}}\stackrel{1}{\bX_{\pm}^{I}}(R')^{IJ},\\
\stackrel{1}{\bX_{-}^{I}}\stackrel{2}{\bX_{+}^{J}}R^{IJ}=R^{IJ}\stackrel{2}{\bX_{+}^{J}}\stackrel{1}{\bX_{-}^{I}}&,\quad (R')^{IJ}\stackrel{1}{\bX_{+}^{I}}\stackrel{2}{\bX_{-}^{J}}=\stackrel{2}{\bX_{-}^{J}}\stackrel{1}{\bX_{+}^{I}}(R')^{IJ}\,.
\label{eq:commutation_X}
\end{aligned}
\ee
The element $\kappa_I^{-1}X_{+}^I(X_{-}^{-1})^{I}$ satisfies both the commutation relation \eqref{eq:commutation_rep_1} and the functionality condition \eqref{eq:holonomy_recoupling}
of loop algebra. With the isomorphism defined above, we see that the matrix elements of $M^I$ for any $I$ are generated by the generators $q^{\frac{H}{2}}, X, Y$ of $\UQ$. This specifies the isomorphism between loop algebra and $\UQ$ for $q$ not a root of unity. 
(We remind the readers that this isomorphism is, strictly speaking, between $\cL_{0,1}$ and $\UQsl$ not necessarily with a $*$-structure to give $\UQ$, but it does not affect the discussion here as the $*$-structure is not involved when we use the isomorphism at the algebraic level. See footnote \ref{footnote:isomorphism}.)
\begin{lemma}
The isomorphism  $\bX^I=\kappa^{-1}_I\bX^I_+(\bX^{-1}_-)^I$ satisfies the ``loop equation'':
\be
(R^{-1})^{IJ}\stackrel{1}{\bX^I}R^{IJ}\stackrel{2}{\bX^J}=\stackrel{2}{\bX^J}(R^{\prime})^{IJ}\stackrel{1}{\bX^I}(R^{\prime-1})^{IJ}\;.
\label{eq:loop_eqn}
\ee
\end{lemma}
\begin{proof}
Let us start from the {\it l.h.s.} of the equation, which can be explicitly written as
\begin{multline}
 (R^{-1})^{IJ}(\stackrel{1}{\bX_+^I}(\stackrel{1}{\bX_{-}^{-1}})^{I})R^{IJ}(\stackrel{2}{\bX_+^I}(\stackrel{2}{\bX_{-}^{-1}})^{I})
=(R^{-1})^{IJ}\stackrel{1}{\bX_+^I}(\stackrel{2}{\bX_+^I})R^{IJ}(\stackrel{1}{\bX_{-}^{-1}})^{I}(\stackrel{2}{\bX_{-}^{-1}})^{I}\\
=\stackrel{2}{\bX_+^I}\stackrel{1}{\bX_+^I}(R')^{IJ}(R'^{-1})^{IJ}(\stackrel{1}{\bX_{-}^{-1}})^{I}(\stackrel{2}{\bX_{-}^{-1}})^{I}
=\stackrel{2}{\bX_+^I}\stackrel{1}{\bX_+^I}(R')^{IJ}(\stackrel{2}{\bX_{-}^{-1}})^{I}(\stackrel{1}{\bX_{-}^{-1}})^{I}(R'^{-1})^{IJ}\\
=\stackrel{2}{\bX_+^I}(\stackrel{2}{\bX_{-}^{-1}})^{I}(R')^{IJ}\stackrel{1}{\bX_+^I}(\stackrel{1}{\bX_{-}^{-1}})^{I}(R'^{-1})^{IJ}\,,
\end{multline}
where the last expression is the {\it r.h.s.} of \eqref{eq:loop_eqn}. 
We apply relation \eqref{eq:commutation_X} several times to obtain the final result.
\end{proof}
\begin{lemma}
The isomorphism  $\bX^I=\kappa^{-1}_I\bX^I_+(\bX^{-1}_-)^I$ satisfies the functoriality condition.
\be
\stackrel{1}{\bX^{I}}R^{IJ}\stackrel{2}{\bX^{J}}=\sum_{K}C[IJ|K]^{*}\bX^{K}C[IJ|K]\;.
\ee
\end{lemma}
\begin{proof}
 Let us start from the {\it l.h.s.} of the equation:   
\be
\begin{split}
 \kappa^{-1}_I\kappa^{-1}_JR'^I_{13}R^I_{13}R^{IJ}_{12}R'^J_{23}R^J_{23}
 &=\kappa^{-1}_I\kappa^{-1}_JR'^I_{13}R'^J_{23}R^{IJ}_{12}R^I_{13}R^J_{23}\\
 &=\kappa^{-1}_I\kappa^{-1}_J(\Delta'\otimes \Id)(R')R^{IJ}_{12}R^I_{13}R^J_{23}\\
 &=\kappa^{-1}_I\kappa^{-1}_JR^{IJ}_{12}R'^I_{23}R'^J_{13}R^I_{13}R^J_{23}\\
 &=\kappa^{-1}_I\kappa^{-1}_JR^{IJ}_{12}(\rho^I\otimes\rho^J\otimes\Id)(\Delta\otimes \Id)(R'R)\\
 &=\kappa^{-1}_I\kappa^{-1}_J\sum_{L}\frac{\kappa_{I}\kappa_{J}}{\kappa_{L}}C[IJ|L]^{*}C[IJ|L](\rho^I\otimes\rho^J\otimes\Id)(\Delta\otimes \Id)(R'R)\\
 &=\sum_{L}\kappa^{-1}_LC[IJ|L]^{*}(R'R)^{L}C[IJ|L]\\
 &=\sum_{L}C[IJ|L]^{*}\bX^{L}C[IJ|L]\;.
\end{split}
\label{eq:proof_functorial}
\ee
In the first equality, we have used the quasi-triangularity $(\Id\otimes \Delta)(R)=R_{13}R_{12}$ and \eqref{eq:DeltaP}. Then we have applied several times of the quasi-triangularity and \eqref{eq:QYBE}. For the fifth equality, we have rewritten $R_{12}^{IJ}$ in terms of the CG maps, \ie
\be
\sum_{L}\frac{\kappa_{I}\kappa_{J}}{\kappa_{L}}C[IJ|L]^{*}C[IJ|L]=R_{12}^{IJ}\;.
\ee
The second to the last line of \eqref{eq:proof_functorial} is obtained by using the definitions \eqref{eq:CG_def1} of CG maps.
\end{proof}

The $*$-structure of loop algebra is defined as:
\be
(\bX^{I})^{*}=\kappa^{-1}\lb R^{I}(\bX^{I})^{-1}(R^{-1})^{I}\rb\kappa\;,
\label{eq:X_star}
\ee
where $\kappa$ is the central element of $\UQ$. 
This $*$-structure is induced by the $*$-structure of $\UQ$ for $|q|=1, q^p\neq 1$ in the sense that the $*$-operation on the matrix elements of $\bX^I$ is closed. Here, we illustrate this induction in the fundamental representation. 
Recall that $\bX^I=\kappa^{-1}_I\bX^I_+(\bX_-^{-1})^I$ with $\bX^I_+={R'}^I$ and $\bX_-=(R^{-1})^I$. Then
\be
X^{\frac{1}{2}}=\kappa_{\f12}^{-1}\mat{cc}{q^H & (q-q^{-1})q^{-\frac{1}{2}}q^{\frac{1}{2}H}Y\\
(q-q^{-1})q^{-\frac{1}{2}}Xq^{\frac{1}{2}H}& (q-q^{-1})^{2}q^{-1}XY+q^{-H}}\,.
\ee
We then calculate the {\it l.h.s.} of \eqref{eq:X_star} that
\be
\begin{aligned}
(\bX^{\frac{1}{2}})^{*}&=\kappa_{\f12}\mat{cc}{(q^{H})^{*} & 
((q-q^{-1})q^{-\frac{1}{2}}q^{\frac{1}{2}H}Y)^{*} \\((q-q^{-1})q^{-\frac{1}{2}}Xq^{\frac{1}{2}H})^{*} & ((q-q^{-1})^{2}q^{-1}XY+q^{-H})^{*}}\\
&=\kappa_{\f12}\mat{cc}{q^{-H} &
(q^{-1}-q)q^{\frac{1}{2}}Xq^{-\frac{H}{2}} \\(q^{-1}-q)q^{\frac{1}{2}}q^{-\frac{H}{2}}Y
& (q^{-1}-q)^{2}qXY+q^{H}}\,,
\end{aligned}
\label{eq:lhs_X}
\ee
where we have used the fact that $(\kappa_I)^*=\kappa^{-1}_I$. 
On the other hand, the {\it r.h.s.} of \eqref{eq:X_star} at the fundamental representation reads
\be\begin{split}
\kappa^{-1}\lb R^{\f12}(\bX^{\f12})^{-1}(R^{-1})^{\f12}\rb\kappa
&=\kappa^{-1}\kappa_{\f12}\mat{cc}{q^{-\frac{1}{2}H} & 0\\
(q^{-1}-q)q^{\frac{1}{2}}X & q^{\frac{1}{2}H}}
\mat{cc}{q^{-\frac{1}{2}H} & (q^{-1}-q)q^{\frac{1}{2}}Y\\
0 & q^{\frac{1}{2}H}}\kappa\\
&=\kappa^{-1}\kappa_{\f12}\mat{cc}{
q^{-H} &
(q^{-1}-q)q^{\frac{1}{2}}q^{-\frac{H}{2}}Y \\
(q^{-1}-q)q^{\frac{1}{2}}Xq^{-\frac{H}{2}} 
& (q^{-1}-q)^{2}qXY+q^{H}
}\kappa\\
&=\kappa_{\f12}\mat{cc}{
q^{-H} &
(q^{-1}-q)q^{\frac{1}{2}}q^{-\frac{H}{2}}Y \\
(q^{-1}-q)q^{\frac{1}{2}}Xq^{-\frac{H}{2}} 
& (q^{-1}-q)^{2}qXY+q^{H}
}\,.
\end{split}
\label{eq:rhs_X}
\ee
In the last step, we have eliminated $\kappa$ by the fact that $\kappa$ is a central element of $\UQ$ and thus commutes with all the generators of $\UQ$. 
Comparing \eqref{eq:lhs_X} and \eqref{eq:rhs_X}, we see that 
$(\bX_{ab}^{\frac{1}{2}})^{*}=\kappa^{-1}\lb R^{\frac{1}{2}}(\bX^{\frac{1}{2}})^{-1}(R^{-1})^{\frac{1}{2}}\rb_{ba}\kappa$ with $a,b=\pm \f12$, which means that the $*$-operation is closed in the fundamental representation. This can be generalized to any other irreducible representation by using the recoupling relation \eqref{eq:holonomy_recoupling} of the generators of loop algebra. 

\medskip

The isomorphism between loop algebra and $\UQ$ can also be generalized to graph algebra, which we now describe. 
The graph algebra $\cL_{0,4}$ is isomorphic to $\UQ^{\otimes 4}$ by the following relations \cite{Alekseev:1993gh}.
\be
\begin{split}
\bX_{\nu}^{I}=\kappa_{I}^{-1}K_{\nu}^{I}\bX_{\nu\nu}^I K_{\nu}^{I-1}\quad,&\quad \bX_{\nu\nu}^I=\bX^I_{\nu, +}(\bX_{\nu, -}^{-1})^I\;,\\
K^{I}_{i}=\bX_{1,+}^{I}\cdots\bX_{i-1,+}^{I}\quad,&\quad K_{1}=e\;.
\end{split}
\ee
These definitions are used in \eqref{eq:iso_M_graph_algebra}.
The matrix elements of $\bX_{\nu\nu}^I$ subject to the following relations \cite{Alekseev:1993gh}.
\begin{align}
(R^{-1})^{IJ}\stackrel{1}{\bX^I_{\nu\nu}}R^{IJ}\stackrel{2}{\bX^{J}_{\nu\nu}}&=\stackrel{2}{\bX^J_{\nu\nu}}(R^{\prime})^{IJ}\stackrel{1}{\bX^{I}_{\nu\nu}}(R^{\prime-1})^{IJ}\;, 
\label{eq:Xloop_equation}\\
\stackrel{1}{\bX_{\nu\nu}}\stackrel{2}{\bX_{\mu\mu}}&=\stackrel{2}{\bX_{\nu\nu}}\stackrel{1}{\bX_{\mu\mu}}\;,
\label{eq:X_commute}
\end{align}
which also apply for all representations.
\begin{lemma}
\label{lemma:X1X2-relation}
 The isomorphism between $\cL_{0,4}$ and $\UQ^{\otimes 4}$ preserves the the commutation relation of the graph algebra.
 \be
 \begin{aligned}
 (R^{-1})^{IJ}\,\stackrel{1}{\bX_{\nu}^I}\, R^{IJ}\,\stackrel{2}{\bX_{\nu}^J} &= \stackrel{2}{\bX_{\nu}^J}, (R')^{IJ}\, \stackrel{1}{\bX_{\nu}^I}\, ({R'}^{-1})^{IJ}\,,\\
(R^{-1})^{IJ}\stackrel{1}{\bX_{\nu}^I}R^{IJ}\stackrel{2}{\bX_{\mu}^J}&=\stackrel{2}{\bX_{\mu}^J}(R^{-1})^{IJ} \stackrel{1}{\bX_{\nu}^I}R^{IJ}\quad,\quad \ell_\nu\prec\ell_\mu,\\
(R')^{IJ}\,\stackrel{1}{\bX_{\nu}^I},(R'^{-1})^{IJ}\,\stackrel{2}{\bX_{\mu}^J} &= \stackrel{2}{\bX_{\mu}^J}\,(R')^{IJ}\,\stackrel{1}{\bX_{\nu}^I}\,(R'^{-1})^{IJ}\,,\quad  \ell_{\nu}\succ\ell_{\mu}.
\end{aligned}
\ee
\end{lemma}
\begin{proof}
 Let us start from the {\it l.h.s.} of the equation for $\ell_\nu\prec\ell_\mu$:
 \be
 \begin{aligned}
  (R^{-1})^{IJ}\stackrel{1}{\bX^I_{\nu}}R^{IJ}\stackrel{2}{\bX^J_{\mu}}&=(R^{-1})^{IJ}\stackrel{1}{(\bX_{1,+})^I}\cdots\stackrel{1}{(\bX_{\nu-1,+})^I}\stackrel{1}{(\bX_{\nu,+})^I}\stackrel{1}{(\bX_{\nu,-}^{-1})^I}\stackrel{1}{(\bX_{\nu-1,+}^{-1})^I}\cdots\stackrel{1}{(\bX_{1,+}^{-1})^I}\;\cdot\\
&\quad\cdot\;(R)^{IJ}\stackrel{2}{(\bX_{1,+})^J}\cdots\stackrel{2}{(\bX_{\mu-1,+})^J}\stackrel{2}{(\bX_{\mu,+})^J}\stackrel{2}{(\bX_{\mu,-}^{-1})^J}\stackrel{2}{(\bX_{\mu,+}^{-1})^J}\cdots\stackrel{2}{(\bX_{1,+}^{-1})^J}\\
&=(R^{-1})^{IJ}\stackrel{1}{(\bX_{1,+})^I}\stackrel{2}{(\bX_{1,+})^J}\cdots\stackrel{1}{(\bX_{\nu-1,+})^I}\stackrel{1}{(\bX_{\nu,+})^I}\stackrel{1}{(\bX_{\nu,-}^{-1})^I}\stackrel{1}{(\bX_{\nu-1,+}^{-1})^I}\cdots\stackrel{1}{(\bX_{2,+}^{-1})^I}\;\cdot\\
&\quad\cdot\;(R)^{IJ}\stackrel{2}{(\bX_{2,+})^J}\cdots\stackrel{2}{(\bX_{\mu-1,+})^J}\stackrel{2}{(\bX_{\mu,+})^J}\stackrel{2}{(\bX_{\mu,-}^{-1})^J}\stackrel{2}{(\bX_{\mu,+}^{-1})^J}\cdots\stackrel{1}{(\bX_{1,+}^{-1})^I}\stackrel{2}{(\bX_{1,+}^{-1})^J}\;.
 \end{aligned}
 \ee
 The second equality is obtained by using relations \eqref{eq:X_commute} and \eqref{eq:commutation_X}. And continuing to use both relations several times and adding $(R^{-1})^{IJ}R^{IJ}$ in the derivation, we have
 \be
 \begin{aligned}
   &\phantom{=}(R^{-1})^{IJ}\stackrel{1}{\bX^I_{\nu}}R^{IJ}\stackrel{2}{\bX^J_{\mu}}\\
  &=\stackrel{2}{(\bX_{1,+})^J}\stackrel{1}{(\bX_{1,+})^I}\stackrel{2}{(\bX_{2,+})^J}\stackrel{1}{(\bX_{2,+})^I}\cdots\stackrel{2}{(\bX_{\nu-1,+})^J}\stackrel{1}{(\bX_{\nu-1,+})^I}\stackrel{2}{(\bX_{\nu,+})^J}\stackrel{1}{(\bX_{\nu,+})^I}(R^{-1})^{IJ}R^{IJ}
 \cdots\stackrel{2}{(\bX_{\mu-1,+})^J}\stackrel{2}{(\bX_{\mu,+})^J}\;\cdot \\
 &\quad\cdot\;\stackrel{2}{(\bX_{\mu,-}^{-1})^J}\stackrel{2}{(\bX_{\mu-1,+}^{-1})^J}\cdots(R^{-1})^{IJ}R^{IJ}\stackrel{1}{(\bX_{\nu,-}^{-1})^I}\stackrel{2}{(\bX_{\nu,+}^{-1})^J}\stackrel{1}{(\bX_{\nu-1,+}^{-1})^I}\stackrel{2}{(\bX_{\nu-1,+}^{-1})^J}\cdots\stackrel{1}{(\bX_{2,+}^{-1})^I}\stackrel{2}{(\bX_{2,+}^{-1})^J}\stackrel{1}{(\bX_{1,+}^{-1})^I}\stackrel{2}{(\bX_{1,+}^{-1})^J}\;.   
 \end{aligned}
 \ee
 Continue to use relations \eqref{eq:X_commute} and \eqref{eq:commutation_X} several times, we obtain the desired result:
 \be
 \begin{aligned}
(R^{-1})^{IJ}\stackrel{1}{\bX^I_{\nu}}R^{IJ}\stackrel{2}{\bX^J_{\mu}}
 &=\stackrel{2}{(\bX_{1,+})^J}\cdots\stackrel{2}{(\bX_{\mu-1,+})^J}\stackrel{2}{(\bX_{\mu,+})^J}\stackrel{2}{(\bX_{\mu,-}^{-1})^J}\stackrel{2}{(\bX_{\mu,+}^{-1})^J}\cdots\stackrel{2}{(\bX_{1,+}^{-1})^J}(R^{-1})^{IJ}\;\cdot\\
 &\quad\cdot\;\stackrel{1}{(\bX_{1,+})^I}\cdots\stackrel{1}{(\bX_{\nu-1,+})^I}\stackrel{1}{(\bX_{\nu,+})^I}\stackrel{1}{(\bX_{\nu,-}^{-1})^I}\stackrel{1}{(\bX_{\nu-1,+}^{-1})^I}\cdots\stackrel{1}{(\bX_{1,+}^{-1})^I}(R)^{IJ}\\
 &=\stackrel{2}{\bX^J_{\mu}}(R^{-1})^{IJ}\stackrel{1}{\bX^I_{\nu}}R^{IJ}\;.
 \end{aligned}
 \ee
  For $\ell_\nu\succ\ell_\mu$ and $\ell_\nu=\ell_\mu$, the proofs are similar.
 \end{proof}

\begin{lemma}
 The embedding $\bX^I_{0}=\kappa^{-1}_I\bX^I_{0,+}(\bX^I_{0,-})^{-1}=\kappa^{-1}_I(\prod_{\nu=1}\bX^I_{\nu,+})(\prod_{\nu=1}\bX^I_{\nu,-})^{-1}$ satisfies the defining relations of loop algebra:
 \begin{align}
(R^{-1})^{IJ}\stackrel{1}{\bX_0^I}R^{IJ}\stackrel{2}{\bX_0^J}&=\stackrel{2}{\bX_0^J}(R^{\prime})^{IJ}\stackrel{1}{\bX_0^I}(R^{\prime-1})^{IJ}\;,\\
\stackrel{1}{\bX_0^{I}}R^{IJ}\stackrel{2}{\bX_0^{J}}&=\sum_{K}C[IJ|K]^{*}\bX_0^{K}C[IJ|K]\;.
 \end{align}
\end{lemma}
\begin{proof}
Thr proof is similar to that of Lemma \ref{lemma:X1X2-relation}. Starting from the {\it l.h.s.} of the equation, we have
\be
\begin{aligned}
 (R^{-1})^{IJ}\stackrel{1}{\bX_0^I}R^{IJ}\stackrel{2}{\bX_0^J}&=(R^{-1})^{IJ}\kappa^{-1}_I\stackrel{1}{(\bX_{1,+})^I}\stackrel{1}{(\bX_{2,+})^I}\cdots\stackrel{1}{(\bX_{\nu,+})^I}\stackrel{1}{(\bX^{-1}_{\nu,-})^I}\stackrel{1}{(X^{-1}_{\nu-1,-})^I}\cdots\stackrel{1}{(X^{-1}_{1,-})^I}R^{IJ}\;\cdot\\
& \quad\cdot\;\kappa^{-1}_J\stackrel{2}{(\bX_{1,+})^J}\stackrel{2}{(\bX_{2,+})^J}\cdots\stackrel{2}{(\bX_{\nu,+})^J}\stackrel{2}{(\bX^{-1}_{\nu,-})^{J}}\stackrel{2}{(\bX^{-1}_{\nu-1,-})^{J}}\cdots\stackrel{2}{(\bX^{-1}_{1,-})^{J}}\\
 &=\kappa^{-1}_I\kappa^{-1}_J\stackrel{2}{(\bX_{1,+})^J}\stackrel{2}{(\bX_{2,+})^J}\cdots\stackrel{2}{(\bX_{\nu,+})^J}\stackrel{1}{(\bX_{1,+})^I}\stackrel{1}{(\bX_{2,+})^I}\cdots\stackrel{1}{(\bX_{\nu,+})^I}(R^{-1})^{IJ}R^{IJ}(R')^{IJ}(R'^{-1})^{IJ}\;\cdot\\
 &\quad\cdot\;\stackrel{1}{(\bX^{-1}_{\nu,-})^I}\stackrel{2}{(\bX^{-1}_{\nu,-})^J}\stackrel{1}{(X^{-1}_{\nu-1,-})^I}\stackrel{2}{(X^{-1}_{\nu-1,-})^J}\cdots\stackrel{1}{(X^{-1}_{1,-})^I}\stackrel{2}{(X^{-1}_{1,-})^J}\;.
\end{aligned}
\ee
To get the second equality, one needs to apply relations \eqref{eq:X_commute} and \eqref{eq:commutation_X}  several times and inserts $(R')^{IJ}(R'^{-1})^{IJ}$ in the equation. At this point, we apply relations \eqref{eq:X_commute} and \eqref{eq:commutation_X}
several times again.
\be
\begin{aligned}
 (R^{-1})^{IJ}\stackrel{1}{\bX_0^I}R^{IJ}\stackrel{2}{\bX_0^J}
 &=\kappa^{-1}_I\kappa^{-1}_J\stackrel{2}{(\bX_{1,+})^J}\stackrel{2}{(\bX_{2,+})^J}\cdots\stackrel{2}{(\bX_{\nu,+})^J}\stackrel{1}{(\bX_{1,+})^I}\stackrel{1}{(\bX_{2,+})^I}\cdots\stackrel{1}{(\bX_{\nu,+})^I} (R')^{IJ}\;\cdot\\
 &\quad\cdot\;\stackrel{2}{(\bX^{-1}_{\nu,-})^{J}}\stackrel{2}{(\bX^{-1}_{\nu-1,-})^{J}}\cdots\stackrel{2}{(\bX^{-1}_{1,-})^{J}}\stackrel{1}{(\bX^{-1}_{\nu,-})^{I}}\stackrel{1}{(\bX^{-1}_{\nu-1,-})^{I}}\cdots\stackrel{1}{(\bX^{-1}_{1,-})^{I}}(R'^{-1})^{IJ}\\
 &=\kappa^{-1}_I\kappa^{-1}_J\stackrel{2}{(\bX_{1,+})^J}\stackrel{2}{(\bX_{2,+})^J}\cdots\stackrel{2}{(\bX_{\nu,+})^J}\stackrel{2}{(\bX^{-1}_{\nu,-})^{J}}\stackrel{2}{(\bX^{-1}_{\nu-1,-})^{J}}\cdots\stackrel{2}{(\bX^{-1}_{1,-})^{J}}(R')^{IJ}\;\cdot\\
 &\quad\cdot\;\stackrel{1}{(\bX_{1,+})^I}\stackrel{1}{(\bX_{2,+})^I}\cdots\stackrel{1}{(\bX_{\nu,+})^I}\stackrel{1}{(\bX^{-1}_{\nu,-})^I}\stackrel{1}{(X^{-1}_{\nu-1,-})^I}\cdots\stackrel{1}{(X^{-1}_{1,-})^I} (R'^{-1})^{IJ}\\
 &=\stackrel{2}{\bX_0^J}(R^{\prime})^{IJ}\stackrel{1}{\bX_0^I}(R^{\prime-1})^{IJ}\;.
\end{aligned}
\ee

To prove the second relation, we again start from the {\it l.h.s.}:
\be
\begin{aligned}
\stackrel{1}{\bX_0^{I}}R^{IJ}\stackrel{2}{\bX_0^{J}}&=\kappa^{-1}_I\kappa^{-1}_J\stackrel{1}{(\bX_{1,+})^I}\stackrel{1}{(\bX_{2,+})^I}\cdots\stackrel{1}{(\bX_{\nu,+})^I}\stackrel{1}{(\bX^{-1}_{\nu,-})^I}\stackrel{1}{(X^{-1}_{\nu-1,-})^I}\cdots\stackrel{1}{(X^{-1}_{1,-})^I}R^{IJ}\;\cdot\\
& \quad\cdot\;\stackrel{2}{(\bX_{1,+})^J}\stackrel{2}{(\bX_{2,+})^J}\cdots\stackrel{2}{(\bX_{\nu,+})^J}\stackrel{2}{(\bX^{-1}_{\nu,-})^{J}}\stackrel{2}{(\bX^{-1}_{\nu-1,-})^{J}}\cdots\stackrel{2}{(\bX^{-1}_{1,-})^{J}}\\
&=\kappa^{-1}_I\kappa^{-1}_J\stackrel{1}{(\bX_{1,+})^I}\stackrel{2}{(\bX_{1,+})^J}\stackrel{1}{(\bX_{2,+})^I}\stackrel{2}{(\bX_{2,+})^J}\cdots\stackrel{1}{(\bX_{\nu,+})^I}\stackrel{2}{(\bX_{\nu,+})^J}R^{IJ}\;\cdot\\
&\quad\cdot\;\stackrel{1}{(\bX^{-1}_{\nu,-})^I}\stackrel{2}{(\bX^{-1}_{\nu,-})^J}\stackrel{1}{(X^{-1}_{\nu-1,-})^I}\stackrel{2}{(X^{-1}_{\nu-1,-})^J}\cdots\stackrel{1}{(X^{-1}_{1,-})^I}\stackrel{2}{(X^{-1}_{1,-})^J}\;.
\end{aligned}
\ee
We have applied the relations \eqref{eq:X_commute} and \eqref{eq:commutation_X}
several times to obtain the second equality. At this point, we use the definitions of $\bX^I_{\nu,+}$ and $\bX^I_{\nu,-}$, then apply the quasi-triangularity several times and rewrite $R_{12}^{IJ}$ in terms of the CG maps, which gives  
\be
\begin{aligned}
&\quad \stackrel{1}{\bX_0^{I}}R^{IJ}\stackrel{2}{\bX_0^{J}}\\
&=\kappa^{-1}_I\kappa^{-1}_J(R'_{13})^I(R'_{23})^J(R'_{14})^I(R'_{24})^J\cdots(R'_{1\nu+2})^I(R'_{2\nu+2})^JR^{IJ}_{12}R^I_{1\nu+2}R^J_{2\nu+2}...R^I_{14}R^J_{24}R^I_{13}R^J_{23}\\
&=\kappa^{-1}_I\kappa^{-1}_J (\rho^I\otimes\rho^J\otimes\Id)(\Delta'\otimes\Id)(R'_{12})(\rho^I\otimes\rho^J\otimes\Id)(\Delta'\otimes\Id)(R'_{13})...(\rho^I\otimes\rho^J\otimes\Id)(\Delta'\otimes\Id)(R'_{1\nu+1})\;\cdot\\
&\quad\cdot\;R^{IJ}_{12}(\rho^I\otimes\rho^J\otimes\Id)(\Delta\otimes\Id)(R_{1\nu+1})\cdots(\rho^I\otimes\rho^J\otimes\Id)(\Delta\otimes\Id)(R_{13})(\rho^I\otimes\rho^J\otimes\Id)(\Delta\otimes\Id)(R_{12})\\
&=\kappa^{-1}_I\kappa^{-1}_J R^{IJ}_{12}(\rho^I\otimes\rho^J\otimes\Id)(\Delta\otimes\Id)(R'_{12})(\rho^I\otimes\rho^J\otimes\Id)(\Delta\otimes\Id)(R'_{13})\cdots(\rho^I\otimes\rho^J\otimes\Id)(\Delta\otimes\Id)(R'_{1\nu+1})\;\cdot\\
&\quad\cdot\;(\rho^I\otimes\rho^J\otimes\Id)(\Delta\otimes\Id)(R_{1\nu+1})\cdots(\rho^I\otimes\rho^J\otimes\Id)(\Delta\otimes\Id)(R_{13})(\rho^I\otimes\rho^J\otimes\Id)(\Delta\otimes\Id)(R_{12})\\
&=\kappa^{-1}_I\kappa^{-1}_J\sum_{L}\frac{\kappa_{I}\kappa_{J}}{\kappa_{L}}C[IJ|L]^{*}C[IJ|L](\rho^I\otimes\rho^J\otimes\Id)(\Delta\otimes\Id)(R'_{12}R'_{13}\cdots R'_{1\nu+1}R_{1,\nu+1}\cdots R_{13}R_{12})\\
&=\sum_{L}\kappa^{-1}_{L}C[IJ|L]^{*}(R'_{12}R'_{13}\cdots R'_{1\nu+1}R_{1,\nu+1}\cdots R_{13}R_{12})^{L}C[IJ|L]\\
&=\sum_{L}C[IJ|L]^{*}\bX^L_0C[IJ|L]\;.
\end{aligned}
\ee
We have, therefore, completed all the proofs.
\end{proof}

\section{Some detailed calculations}
\label{app:details}

In this appendix, we collect some detailed calculations used in the main text, which are formulated into lemmas. 
Lemma \ref{lemma:trace_lemma} is the proof of \eqref{eq:trace_lemma_0}. 
Lemmas \ref{lemma:flatness_gen} and \ref{lemma:flatness_trun} are the proofs of \eqref{eq:flatness} first for $\UQ$ with $q^p=1$ assuming semi-simplicity, which we call the non-truncated case, then for $\TUQ$ with semi-simplicity property, which we call the truncated case. Lemmas \ref{lemma:tr_RpR_gen} and \ref{lemma:tr_RpR_trun} are useful in the discussion of the area operator in Section \ref{sec:area}.

 \begin{lemma}
 \label{lemma:trace_lemma}
The quantum trace has the following property
\be
\tr_q^I(Y^I )=\tr_q^I((R^{-1})^{I} Y^I R^{I})=\tr_q^I(R'^{I}Y^I(R'^{-1})^{I})\,,\quad \forall\ Y^I\in \End(V^I)\otimes e\otimes \UQ\,,
\label{eq:trace_lemma}
\ee
where $e$ is the identity element in $\UQ$. 
\end{lemma}
\begin{proof}
Let $Y^I=y^I\otimes e\otimes \xi$ where $y^i\in\End(V^I)$ and  $\xi\in\UQ$.
We start from the {\it r.h.s.} of the first equation and take the $J$ representation for the second tensor space which gives $\rho^J(Y^I):=y^I\otimes e^J\otimes \xi$. Then we calculate
\be\begin{split}
\tr_q^I\lb (R^{-1})^{IJ}\rho^J(Y^I)R^{IJ}\rb 
&= \tr^I \lb \lb\sum_a  S(R^{(1)}_a)\otimes R^{(2)}_a\rb^{IJ} Y^I \lb  \sum_b  R^{(1)}_b\otimes R^{(2)}_b\rb^{IJ} \rho^I(g)\rb\\
&=\sum_{ab} \tr \lb S(R_a^{(1)})^Iy^I R_b^{(1),I}\rho^I(g) \rb \otimes R_a^{(2),J} R_b^{(2),J} \otimes \xi\\
&=\sum_{ab}\tr\lb y^IR_b^{(1),I}S^{-1}(R^{(1)}_a)^I\rho^I(g)\rb\otimes R_a^{(2),J} R_b^{(2),J} \otimes \xi\\
&=\sum_{ab}\tr\lb y^IS^{-1}\lb (R^{(1)}_a)^I S( R_b)^{(1),I}\rb\rho^I(g)\rb\otimes R_a^{(2),J} R_b^{(2),J} \otimes \xi\\
&=\tr_q^I\lb\rho^J( Y^I) ( S^{-1}\otimes \Id ) ( R^{IJ} (R^{-1})^{IJ})\rb 
\equiv \tr^I_q(\rho^J(Y^I))\,,
\end{split}
\label{eq:trace_property}
\ee
where we have used the cyclic property of trace and the intertwining property, \ie $S^{-1}(\xi)g=gS(\xi)\,,\forall \xi\in\UQ$, of the group-like element $g$ (the third equation of \eqref{eq:g_properties}) to obtain the third line and used the anti-homomorphism property $S^{-1}(\xi\eta)=S^{-1}(\eta)S^{-1}(\xi)$ of the antipode to obtain the fourth line. As \eqref{eq:trace_property} is true for all representation $J$, the first equality in \eqref{lemma:trace_lemma} is true. 

The second equality follows the same logic. 
\be\begin{split}
\tr_q^I\lb R'^{IJ}\rho^J(Y^I)(R'^{-1})^{IJ}\rb
&=\sum_{ab} \tr\lb R_a^{(2),I}y^IS(R_b^{(2),I})\rho^I(g)\rb \otimes R_a^{(1),J}R_b^{(1),J} \otimes \xi  \\
&=\sum_{ab} \tr\lb y^IS\lb S^{-1}(R_a^{(2),I})\rho^{I}(g^{-1})R_b^{(2),I}\rb\rb \otimes R_a^{(1),J}R_b^{(1),J} \otimes \xi \\
&=\sum_{ab} \tr\lb y^IS\lb \rho^{I}(g^{-1}) S(R_a^{(2),I}) R_b^{(2),I}\rb\rb \otimes R_a^{(1),J}R_b^{(1),J} \otimes \xi \\
&=\sum_{ab} \tr\lb y^IS\lb  S(R_a^{(2),I}) R_b^{(2),I}\rb \rho^{I}(g)\rb \otimes R_a^{(1),J}R_b^{(1),J} \otimes \xi \\
&=\tr^I_q\lb \rho^J(Y^I) (S\otimes\Id)( (R^{-1})^{IJ} R^{IJ} )  \rb \equiv \tr^I_q(\rho^J(Y^I))\,,
\end{split}\ee
where we have used the cyclic property of trace, the anti-homomorphism property of antipode and the property $S(g)=g^{-1}$ (the second equation of \eqref{eq:g_properties}) of $g$ to obtain the second line and the intertwining property of $g$ to obtain the third line. 
\end{proof}

\begin{lemma}[Flatness for the non-truncated case]
\label{lemma:flatness_gen}\cite{Alekseev:1994au}
The elements $\chi_{0}^{0}$ and $M_{0}^{I}$ satisfy the following relation. 
 \be
\begin{split}
\chi_{0}^{0} \bM_{0}^{I}=\kappa_{I}^{-1} \chi_{0}^{0} e^{I}\;,
\end{split}
\label{eq:flatness_non-truncated}
\ee
where the complex number $\kappa_{I}^{-1}$ gives a quantum correction. When $q\to 1$, Eq.\eqref{eq:flatness_non-truncated} recovers the classical closure condition \eqref{eq:closure_SU2}.
\end{lemma}
\begin{proof}
 We start from the functoriality relations of holonomies, \ie 
 \be
 \stackrel{1}{\bM^I_0}(R)^{IJ}\stackrel{2}{\bM^I_0}=\sum_K C[IJ|K]^*\bM^K_0C[IJ|K]\;.
 \ee
 Acting with the character $c^I_0$ on $\bM_0^J$, we calculate that 
 \be
 \begin{aligned}
c^I_0\bM^J_0&=\kappa_{I}\tr_q^I(\bM^I_0)\bM^J_0\\
&=\kappa_{I}\tr_q^I(\sum_K (R^{-1})^{IJ}C[IJ|K]^*\bM^K_0C[IJ|K])\\
&=\kappa_{I}\sum_K \frac{d_I}{v_I}C[\Bar{I}I|0] (R^{-1})^{IJ}C[IJ|K]^*\bM^K_0C[IJ|K](R')^{\Bar{I}I}C[\Bar{I}I|0]^{*}\\
&=\kappa_{I} \sum_K \frac{d_{I}}{v_{I}}\frac{v_{k}}{v_{I}v_{J}}\frac{d_{k}v_{I}}{d_{I}v_{k}}\frac{d_{k}v_{I}}{d_{I}v_{k}}C[\Bar{K}K|0](R')^{J\Bar{K}} C[J\Bar{K}|\Bar{I}]^{*}C[J\Bar{K}|\Bar{I}]\stackrel{3}{\bM^{K}_{0}}(R')^{\Bar{K}K}C[\Bar{K}K|0]^{*}\\
&=\kappa_{I}\sum_K\frac{d^2_K}{d_Iv_Jv_K}C[\Bar{K}K|0](R')^{J\Bar{K}} C[J\Bar{K}|\Bar{I}]^{*}C[J\Bar{K}|\Bar{I}]\stackrel{3}{\bM^{K}_{0}}(R')^{\Bar{K}K}C[\Bar{K}K|0]^{*}\;,
 \label{eq:flatness_proof}
 \end{aligned}
 \ee
where we have used the quantum trace identity \eqref{eq:trace_lemma} in the second line. The fourth line is obtained by using $\frac{v_K}{v_Iv_J}(R')^{IJ}C[IJ|K]^*=(R^{-1})^{IJ}C[IJ|K]^*$ and the following relations given and proved in \cite{Alekseev:1994au}:
 \be
\begin{aligned}
C[IJ|K](R')^{\Bar{I}I}C[\Bar{I}I|0]^{*}&=\frac{d_{K}v_{I}}{d_{I}v_{K}}(A) C[J\Bar{K}|I](R')^{\Bar{K}K}C[\Bar{K}K|0]^{*}\;,\\
C[\Bar{I}I|0](R')^{IJ} C[IJ|K]^{*}&=\frac{d_{K}v_{I}}{d_{I}v_{K}}(A^{-1})C[\Bar{K}K|0](R')^{J\Bar{K}} C[J\Bar{K}|\Bar{I}]^{*}\;,
\end{aligned}
\ee
 where A is a number.  
 With $\chi_{0}^{0}= \sum_I \mathcal{N}^{2} d_{I} c_{0}^{I}$, it follows that
 \be
 \begin{aligned}
   \chi^0_0 \bM^J_0&= \mathcal{N}^{2} \sum_{I,K} \kappa_{I}\frac{d^2_K}{v_Jv_K}C[\Bar{K}K|0](R')^{J\Bar{K}} C[J\Bar{K}|\Bar{I}]^{*}C[J\Bar{K}|\Bar{I}]\stackrel{3}{\bM^{K}_{0}}(R')^{\Bar{K}K}C[\Bar{K}K|0]^{*}\\ 
   &=\mathcal{N}^{2} \sum_{I,K} \kappa_{I}\frac{d^2_K}{v_Jv_K}C[\Bar{K}K|0]\frac{\kappa_J\kappa_{\Bar{K}}}{\kappa_{\Bar{I}}}(e^J\otimes e^{\Bar{K}})\stackrel{3}{\bM^{K}_{0}}(R')^{\Bar{K}K}C[\Bar{K}K|0]^{*}\\
   &=\mathcal{N}^{2} \sum_{K} \frac{d^2_K}{\kappa_J\kappa_K}C[\Bar{K}K|0](e^J\otimes e^{\Bar{K}})\stackrel{3}{\bM^{K}_{0}}(R')^{\Bar{K}K}C[\Bar{K}K|0]^{*}\\
   &=\kappa^{-1}_J e^J\sum_{K}\mathcal{N}^{2} d_K \frac{d_K}{\kappa_K}C[\Bar{K}K|0] e^{\Bar{K}}\stackrel{2}{\bM^{K}_{0}}(R')^{\Bar{K}K}C[\Bar{K}K|0]^{*}\\
   &=\kappa^{-1}_Je^J \chi^0_0\;.
 \end{aligned}
 \ee
 The second line is obtained by the completeness of the CG maps:
 \be
 \sum_{\Bar{I}} \frac{\kappa_{\Bar{I}}}{\kappa_J \kappa_{\Bar{K}}}
 (R')^{J\Bar{K}}C[J\Bar{K}|\Bar{I}]^*C[J\Bar{K}|\Bar{I}]=e^J\otimes e^{\Bar{K}}\;.
 \ee
\end{proof}
\begin{lemma}[Flatness for the truncated case]
\label{lemma:flatness_trun}
 With the truncation, the elements $\chi_{0}^{0}$ and $M^{I}_{0}$ still satisfy the following relation:
\be
\chi_{0}^{0} \bM^{I}_{0} = (\kappa_{I})^{-1} \chi_{0}^{0} e^{I}\;. 
\label{eq:flatness_truncated}
\ee 
The complex number $\kappa_{I}^{-1}$ gives a quantum correction, one can see when $q\to 1$ we recover the classical closure condition.
\end{lemma}
\begin{proof}:
The proof follows Ref.\cite{Alekseev:1994au}. All proofs in \eqref{eq:flatness_proof} can be translated to the truncated case with the substitution rules \eqref{eq:substitution_rule} rather straightforwardly. The only difference is that we do not have the completeness of the CG maps for the truncated case. The quasi-associativity of tensor product of representations gives \cite{Alekseev:1994au}:
\be
e^JC[\Bar{K}K|0]=\sum_{\Bar{I}}\mathcal{F}(IJK)C[\Bar{I}K|J]C[J\Bar{K}|\Bar{I}]\;,
\ee
where $\mathcal{F}(IJK)$ is a coefficient obtained from a subset of 6j-symbols. With the normalization of the CG maps, the expression above can be reformulated as
\be
\begin{aligned}
e^JC[\Bar{K}K|0](R')^{J\Bar{K}}C[J\Bar{K}|\Bar{L}]^*&=\sum_{\Bar{I}}\mathcal{F}(IJK)C[\Bar{I}K|J]C[J\Bar{K}|\Bar{I}](R')^{J\Bar{K}}C[J\Bar{K}|\Bar{L}]^*\\
&=\sum_{\Bar{I}}\mathcal{F}(IJK)C[\Bar{I}K|J]\frac{\kappa_J\kappa_{\Bar{K}}}{\kappa_{\Bar{L}}}\delta_{\Bar{I},\bar{L}}\\
&=\mathcal{F}(LJK)C[\Bar{L}K|J]\frac{\kappa_J\kappa_{\Bar{K}}}{\kappa_{\Bar{L}}}\;.
\end{aligned}
\ee
Using this property, the {\it l.h.s.} of \eqref{eq:flatness_truncated} can be written as
\be
\begin{aligned}
  \chi^0_0\bM^J_0&= \mathcal{N}^{2} \sum_{I,K} \kappa_{I}\frac{d^2_K}{v_Jv_K}C[\Bar{K}K|0](R')^{J\Bar{K}} C[J\Bar{K}|\Bar{I}]^{*}C[J\Bar{K}|\Bar{I}]\stackrel{3}{\bM^{K}_{0}}(R')^{\Bar{K}K}C[\Bar{K}K|0]^{*}\\
  &=\mathcal{N}^{2} \sum_{I,K} \kappa_{I}\frac{d^2_K}{v_Jv_K}\mathcal{F}(IJK)C[\Bar{I}K|J]\frac{\kappa_J\kappa_{\Bar{K}}}{\kappa_{\bar{I}}}C[J\Bar{K}|\Bar{I}]\stackrel{3}{\bM^{K}_{0}}(R')^{\Bar{K}K}C[\bar{K}K|0]^{*}\\
  &=\mathcal{N}^{2} \sum_{K} \kappa_{I}\frac{d^2_K}{v_Jv_K}\frac{\kappa_J\kappa_{\Bar{K}}}{\kappa_{\bar{I}}}e^JC[\Bar{K}K|0]\stackrel{2}{\bM^{K}_{0}}(R')^{\Bar{K}K}C[\Bar{K}K|0]^{*}\\
  &=\kappa^{-1}_J e^J \chi^0_0 \;,
\end{aligned}
\ee
which completes the proof and all proofs can be translated to the truncated case with the help of the substitution
rules \eqref{eq:substitution_rule} . 
\end{proof}
\begin{lemma}
\label{lemma:tr_RpR_gen} 
Let $\rho^{j_{1}}$ and $\rho^{j_{2}}$ be two irreducible representations of $\mathcal{U}_q(\mathfrak{su}(2))$. The matrix s is defined as $s_{IJ}=\tr_q^{j_1}\otimes \tr_q^{j_2}((R'R)^{j_1j_2})$ and it is evaluated as $[(2j_1+1)(2j_2+1)]_q$.
\end{lemma}
\begin{proof}
A direct calculation gives
\be
\begin{split}
\tr_q^{j_1}\otimes \tr_q^{j_2}((R'R)^{j_1j_2})
&=\sum_{J=j_{1}-j_{2}}^{j_{1}+j_{2}}q^{2J(J+1)-2j_{1}(j_{1}+1)-2j_{2}(j_{2}+1)}[2J+1]_q\\
&=\sum_{J=j_{1}-j_{2}}^{j_{1}+j_{2}}q^{2J(J+1)-2j_{1}(j_{1}+1)-2j_{2}(j_{2}+1)}\frac{q^{2J+1}-q^{-(2J+1)}}{q-q^{-1}}\\
&=\frac{1}{q-q^{-1}}\sum_{J=j_{1}-j_{2}}^{j_{1}+j_{2}}q^{2J(J+1)-2j_{1}(j_{1}+1)-2j_{2}(j_{2}+1)}(q^{2J+1}-q^{-(2J+1)})\\
&=\frac{1}{q-q^{-1}}\sum_{J=j_{1}-j_{2}}^{j_{1}+j_{2}}q^{-2j_{1}(j_{1}+1)-2j_{2}(j_{2}+1)}(q^{2(J+1)^{2}-1}-q^{2J^{2}-1})\;.
\label{eq:RpR_proof}
\end{split}
\ee
Lemma \ref{lemma:RpR_qgen} and the decomposition $\tr_q^{j_1}\otimes \tr_q^{j_2}= \sum_J N^{j_1j_2}_J \tr_q^J$, due to the semi-simplicity property of $\mathcal{U}_q(\mathfrak{su}(2))$, are used to get the first equality. The multiplicity is $1$ for the admissible tuple $(j_1,j_2,J)$ in $\mathcal{U}_q(\mathfrak{su}(2))$ case. Therefore,
\be
\begin{split}
\tr_q^{j_1}\otimes \tr_q^{j_2}((R'R)^{j_1j_2})
&=\frac{q^{-2j_{1}(j_{1}+1)-2j_{2}(j_{2}+1)}}{q-q^{-1}}\left(\sum_{J'=j_{1}-j_{2}+1}^{j_{1}+j_{2}+1}q^{2J'^{2}-1}-\sum_{J=j_{1}-j_{2}}^{j_{1}+j_{2}}q^{2J^{2}-1}\right) \\
&=\frac{q^{-2j_{1}(j_{1}+1)-2j_{2}(j_{2}+1)}}{q-q^{-1}}q^{-1}\left(q^{2(j_{1}+j_{2}+1)^{2}}-q^{2(j_{1}-j_{2})^{2}}\right)\\
&=\frac{1}{q-q^{-1}}\left(q^{4j_{1}j_{2}+2j_{1}+2j_{2}+1}-q^{-(4j_{1}j_{2}+2j_{1}+2j_{2}+1)}\right)=[(2j_{1}+1)(2j_{2}+1)]_q\;.
\end{split}
\ee
\end{proof}
\begin{lemma}
\label{lemma:tr_RpR_trun}
  For $\mathcal{U}^T_q(\mathfrak{su}(2))$, the tensor product of two representations $\rho^I\boxtimes\rho^J$ is decomposable and its irreducible components $\rho^K$ are all physical representations, \ie $|I-J|\leq K\leq u(I,J)$.
  \be
  \sum_{K=|I-J|}^{u(I,J)} \frac{v_Iv_J}{v_K}\tr_q^K(\Id_K)=[(2I+1)(2J+1)]_q.
  \ee
 \end{lemma}
 \begin{proof}
  For the case $|I-J|\le K\le I+J$,
  the derivation of the $s_{IJ}=\tr_{q}^{I}\otimes \tr_{q}^{J}((R'R)^{IJ})$ is the same as $q$ generic case.
  For the case $|I-J|\le K\le p-2-I-J$, the proof in the \eqref{eq:RpR_proof} can be directly translated to the truncated case. Based on the \eqref{eq:RpR_proof}, we continue the proof for the truncated case: 
  \be
  \begin{aligned}
  \sum_{K}\tr_{q}^{K}(\frac{v_{I}v_{J}}{v_{K}}\Id_{K})
 &=\frac{1}{q-q^{-1}}q^{-2I(I+1)-2J(J+1)}\left(\sum_{K'=I-J+1}^{p-2-I-J+1}q^{2K'^{2}-1}-\sum_{K=I-J}^{p-2-I-J}q^{2K^{2}-1}\right)\\
 &=\frac{1}{q-q^{-1}}q^{-2I(I+1)-2J(J+1)}\left(q^{2(p^{2}-2p(I+J+1)+(I+J+1)^{2})-1}-q^{2(I-J)^{2}-1}\right)\\
 &=\frac{1}{q-q^{-1}}q^{-2I(I+1)-2J(J+1)}\left(q^{2(I+J+1)^{2}-1}-q^{2(I-J)^{2}-1}\right)\\
 &=\frac{1}{q-q^{-1}}\left(q^{(4IJ+2I+2J+1)}-q^{(-4IJ-2I-2J-1)}\right)\\
 &=[(2I+1)(2J+1)]_q\;.
\end{aligned}
\ee
The third equality is due to $q^p=1$.
 \end{proof}

\nocite*{}
\bibliographystyle{bib-style} 
\bibliography{QI.bib}

\end{document}